\documentclass[12pt, a4paper]{article}

\usepackage{amssymb,amsmath,amsfonts,amsthm, amscd, mathrsfs,helvet, mathtools}
\usepackage[nosort]{cite}
\usepackage{bm, stmaryrd}
\usepackage{graphicx,verbatim}
\usepackage[usenames, dvipsnames]{color}
\usepackage{psfrag}
\usepackage[all]{xy}

\definecolor{lightred}{RGB}{255,127,127}
\definecolor{lightgreen}{RGB}{127,255,127}
\definecolor{lightblue}{RGB}{127,127,255}
\definecolor{linkcolor}{rgb}{0,0,0.6}
\usepackage[ colorlinks=true,
pdfstartview=FitV,
linkcolor= linkcolor,
citecolor= linkcolor,
urlcolor= linkcolor,
hyperindex=true,
hyperfigures=false]
{hyperref}

\theoremstyle{plain}
\newtheorem{theorem}{Theorem}[section]
\newtheorem*{theorem*}{Theorem}

\newtheorem*{proposition*}{Proposition}

\newtheorem{lemma}[theorem]{Lemma}
\newtheorem{corollary}[theorem]{Corollary}

\definecolor{brightBlue}{rgb}{0,0,1}

\definecolor{Violet}{rgb}{0.47,0,1}

\DeclareMathOperator{\res}{res}

\newcommand{\noi}{\noindent}
\newcommand{\dd}{\text{d}}
\newcommand{\p}{\partial}
\newcommand{\g}{\mathfrak{g}}

\newcommand{\so}{\mathfrak{so}}
\newcommand{\spc}{\mathfrak{sp}}
\renewcommand{\sl}{\mathfrak{sl}}

\newcommand{\C}{\mathbb{C}}
\newcommand{\s}{\sigma}

\newcommand{\E}{\mathcal{E}}
\newcommand{\M}{\mathcal{M}}
\newcommand{\R}{\mathbb{R}}

\newcommand{\Nc}{\mathcal{N}}
\newcommand{\W}{\mathcal{W}}
\newcommand{\Z}{\mathbb{Z}}

\newcommand{\Zc}{\mathcal{Z}}
\newcommand{\Y}{\mathcal{Y}}
\newcommand{\Tc}{\mathscr{T}}
\newcommand{\J}{\mathscr{J}}
\newcommand{\K}{\mathscr{K}}
\newcommand{\Hc}{\mathcal{H}}
\newcommand{\Cc}{\mathcal{C}}
\newcommand{\Pc}{\mathcal{P}}
\newcommand{\Q}{\mathcal{Q}}
\newcommand{\Id}{\text{Id}}
\newcommand{\Lc}{\mathcal{L}}
\newcommand{\Lct}{\widetilde{\mathcal{L}}}
\newcommand{\Rc}{\mathcal{R}}
\newcommand{\Rct}{\widetilde{\mathcal{R}}}
\newcommand{\ti}[1]{_{\bm{\underline{#1}}}}
\newcommand{\Tr}{\text{Tr}}
\newcommand{\Pexp}{\text{P}\overleftarrow{\text{exp}}}
\newcommand{\Tp}{^{\mathsf{T}}}

\numberwithin{equation}{section}

\begin{document}

\begin{center}
\vspace*{2em}
{\large\bf Local charges in involution and\\[1mm]
hierarchies in integrable sigma-models}\\
\vspace{1.5em}
S. Lacroix$\,{}^1$, M. Magro$\,{}^1$, B. Vicedo$\,{}^2$

\vspace{1em}
\begingroup\itshape
{\it 1) Univ Lyon, Ens de Lyon, Univ Claude Bernard, CNRS,\\ Laboratoire de Physique,
F-69342 Lyon, France}\\
\vspace{1em}
{\it 2) School of Physics, Astronomy and Mathematics,
University of Hertfordshire,}\\
{\it College Lane,
Hatfield AL10 9AB,
United Kingdom}
\par\endgroup
\vspace{1em}
\begingroup\ttfamily
Sylvain.Lacroix@ens-lyon.fr, Marc.Magro@ens-lyon.fr, Benoit.Vicedo@gmail.com
\par\endgroup
\vspace{1.5em}
\end{center}

\begin{abstract}
Integrable $\s$-models, such as the principal chiral model, $\Z_T$-coset models for $T \in \Z_{\geq 2}$ and their various integrable deformations, are examples of non-ultralocal integrable field theories 
described by $r/s$-systems with twist function. In this general setting, and when the Lie algebra $\g$ underlying the $r/s$-system is of classical type, we construct an infinite algebra of local conserved charges in involution, extending the approach of Evans, Hassan, MacKay and Mountain developed for the principal chiral model and symmetric space $\s$-model. In the present context, the local charges are attached to certain `regular' zeros of the twist function and have increasing degrees related to the exponents of the untwisted affine Kac-Moody algebra $\widehat{\g}$ associated with $\g$. The Hamiltonian flows of these charges are shown to generate an infinite hierarchy of compatible integrable equations.
\end{abstract}


\section{Introduction}

One of the hallmarks of integrability in a $(1+1)$-dimensional field theory is the existence of an infinite number of conserved charges. At the classical level, this property can be attributed to the existence of a Lax connection, depending on some auxiliary complex spectral parameter $\lambda$, whose zero curvature equation for all $\lambda$ is equivalent to the equations of motion of the field theory. In principle, the conserved charges can all be obtained by expanding the monodromy of the Lax connection in $\lambda$ around specific points. Depending on the chosen point of expansion, the resulting charges can be either local or non-local in the fields entering the Lax connection.

When passing to the Hamiltonian formalism, one requires additionally that the conserved charges be in involution with respect to the Poisson bracket of the theory, namely that they Poisson commute not only with the Hamiltonian but also between themselves. It is known since the mid-eighties that a sufficient condition guaranteeing the involution of the charges built from the monodromy is that the Poisson bracket of the Lax matrix $\mathcal L(\lambda, x)$, the spatial component of the Lax connection, be of Maillet's general $r/s$-form \cite{Maillet:1985fn,Maillet:1985ek}.
In most cases of interest, the pair of matrices $r\ti{12}(\lambda, \mu)$ and $s\ti{12}(\lambda, \mu)$, giving the form its name, are rational functions on $\C^2$ valued in the two-fold tensor product $\g \otimes \g$ of some finite-dimensional Lie algebra $\g$. The mathematical formalism underlying this particular case was pinned down by Semenov-Tian-Shansky in \cite{SemenovTianShansky:1983ik} where the $r$- and $s$-matrices were understood to be the skew-symmetric and symmetric parts of a single solution $\mathcal R\ti{12}(\lambda, \mu)$ of the classical Yang-Baxter equation. In general, the latter is related to the standard solution $\mathcal R^0\ti{12}(\lambda, \mu)$ on the (twisted) loop algebra over $\g$ by
\begin{equation*}
\Rc\ti{12}(\lambda,\mu) = \Rc^0\ti{12}(\lambda,\mu) \varphi(\mu)^{-1},
\end{equation*}
where $\varphi(\lambda)$ is a rational function 
\cite{SemenovTianShansky:1995ha,Sevostyanov:1995hd,Vicedo:2010qd} (see also \cite{Maillet:1985ec}), 
known as the twist function.    The loop algebra over $\g$ is twisted by some automorphism 
$\sigma$ of order $T \in \Z_{\geq 1}$, with the non-twisted case corresponding to $T=1$. 
Recall that the twist function 
  plays an essential role in characterising a given integrable field theory. In fact, it was 
shown recently in \cite{Vicedo:2017cge} to form an integral part of the Lax matrix itself, 
viewed as a rational function valued in the untwisted affine Kac-Moody algebra associated 
with $\g$. 

\medskip

The principal chiral model serves as the prototypical example of an integrable $\sigma$-model which fits the general formalism of $r/s$-systems with twist function \cite{Sevostyanov:1995hd}. As such, it has proved extremely fruitful over the past couple of years to try and reinterpret some of its properties in terms of its twist function. Indeed, once a given property has been understood algebraically at the level of the twist function, it almost immediately generalises to other integrable field theories described within this formalism.
The first illustration of this general philosophy came about from the desire to generalise the Faddeev-Reshetikhin construction \cite{Faddeev:1982rn}, initially developed in the context of the $SU(2)$ principal chiral model, to other integrable $\sigma$-models with twist function. Specifically, it was shown in \cite{Delduc:2012qb} that the key initial step of this construction can be naturally reformulated in the language of twist functions. This led to a proposal for extending the Faddeev-Reshetikhin approach to a wide range of other models of interest, including the symmetric and semi-symmetric space $\sigma$-models, as well as the Green-Schwarz superstring on $AdS_5 \times S^5$ in \cite{Delduc:2012vq}. The prospect of extending the subsequent steps in the Faddeev-Reshetikhin construction to these other models remains an exciting open problem.
It is interesting to note that, despite its usefulness in the case of the principal chiral model, the relevance of the twist function was first appreciated in \cite{Vicedo:2010qd}, following \cite{Magro:2008dv,Vicedo:2009sn}, on the much more elaborate $AdS_5 \times S^5$ superstring both within the Green-Schwarz \cite{Metsaev:1998it} and the pure spinor \cite{Berkovits:2000fe} formulations.

Another important application of the formalism of the twist function ties in with the great effort made in recent years \cite{Klimcik:2002zj,Klimcik:2008eq,Delduc:2013fga,Delduc:2013qra,Sfetsos:2013wia,
Delduc:2014kha,Kawaguchi:2014qwa,Delduc:2014uaa,Hollowood:2014rla,Hollowood:2014qma,Hoare:2014oua,Sfetsos:2015nya} towards deforming some well known integrable field theories, such as the principal chiral model as well as symmetric and semi-symmetric space $\sigma$-models, while preserving their integrability.
Specifically, it was realised in \cite{Delduc:2013fga} that the so called Yang-Baxter $\sigma$-model, first introduced by {Klim$\check{\text{c}}$\'{\i}k} in \cite{Klimcik:2002zj} as a certain one-parameter deformation of the principal chiral model on any real Lie group $G_0$, could be naturally obtained by deforming the poles of the twist function of the principal chiral model. This led to an immediate broadening of the landscape of Yang-Baxter type deformations to include also one-parameter deformations of the symmetric and semi-symmetric space $\sigma$-models, incorporating, in particular, the Green-Schwarz superstring on $AdS_5 \times S^5$ in \cite{Delduc:2014kha}. In fact, many other deformations were also understood a posteriori to arise in this fashion \cite{Vicedo:2015pna}.
It is worth noting in passing that the bi-Yang-Baxter $\sigma$-model \cite{Klimcik:2008eq,Klimcik:2014bta} is special in this regard. Although it was originally devised as a two-parameter deformation of the principal chiral model on any real Lie group $G_0$, it can equally be regarded as a two-parameter deformation of the symmetric space $\sigma$-model on $G_0 \times G_0 / G_{0, \rm diag}$ with $G_{0, \rm diag}$ the diagonal subgroup of $G_0 \times G_0$.
It was shown in \cite{Delduc:2015xdm} that it is this latter formulation which fits within the general framework of $r/s$-systems with twist function. The original description of the model as a double deformation of the principal chiral model can be recovered by fixing the $G_{0, \rm diag}$ gauge symmetry, at the expense of losing the formulation in terms of a twist function \cite{Delduc:2015xdm} (see, however, subsection \ref{Sec:BYB} below).

Aside from providing a systematic way of constructing integrable deformations,
the idea of deforming the pole structure of the twist function of a given integrable field theory has also played a pivotal role in establishing and characterising the symmetry algebras of the resulting deformed models. Specifically, the principal chiral model and (semi-)symmetric space $\sigma$-models all have a double pole in their twist function which splits up into a pair of simple poles when their Yang-Baxter type deformation is switched on. It was shown in \cite{Delduc:2013fga, Delduc:2014kha}, based on earlier work \cite{Kawaguchi:2010jg,Orlando:2010yh,
Kawaguchi:2011mz,Kawaguchi:2011pf,Kawaguchi:2012ve,Kawaguchi:2012gp,Orlando:2012hu,Kawaguchi:2013gma} (see also more recent related results \cite{Kameyama:2014bua,Itsios:2014vfa,Hollowood:2015dpa}), that the charges extracted from the leading order in the expansion of the monodromy at this pair of simple poles satisfy all the relations of a Poisson algebra $\mathscr U_q(\g)$, the semiclassical counterpart of the quantum group $U_{\widehat{q}}(\g)$ with $\widehat{q} = q^\hbar$. This general feature of Yang-Baxter type deformations of double poles in the twist function was subsequently related to Poisson-Lie $G$-symmetries in \cite{Delduc:2016ihq}. Amongst the charges spanning the Poisson algebra $\mathscr U_q(\g)$, those associated to non-Cartan generators are all non-local. In the example of the Yang-Baxter $\sigma$-model, the level zero charges together with two additional non-local charges coming from the next order in the expansion of the monodromy around the simple poles of the twist function, have been shown \cite{Delduc:2017brb} to satisfy all the defining relations of the semiclassical counterpart of the quantum affine algebra $U_{\widehat{q}}(\widehat{\g})$.

\medskip

The purpose of the present article is to provide another application of the general formalism of $r/s$-systems with twist function. Specifically, we will describe how, in this general framework, infinite towers of local charges can be associated with certain zeros of the twist function, all of which are in pairwise involution. Following the same spirit as recalled above, the starting point of our approach was to reinterpret the construction of local charges in the principal chiral model due to Evans, Hassan, MacKay and Mountain \cite{Evans:1999mj} in the present language of twist functions. In fact, this construction had soon been generalised to include also the (supersymmetric) principal chiral model with a Wess-Zumino term in \cite{Evans:2000hx}, symmetric space $\sigma$-models in \cite{Evans:2000qx} as well as supersymmetric coset $\sigma$-models in \cite{Evans:2005zd}. Each of these generalisations can be regarded as further evidence that such a construction should hold for any integrable field theory with twist function, while at the same time providing indications on how to do so. In the remainder of this introduction we will briefly summarise the main results of the paper.

Let us first note that in all of the integrable $\sigma$-models with twist function described above, every zero of $\varphi(\lambda)$ is such that $\varphi(\lambda) \mathcal L(\lambda, x)$ is regular there. In a general integrable field theory with twist function $\varphi(\lambda)$ we shall say that any zero of $\varphi(\lambda)$ with this property is \emph{regular}. We denote by $\mathcal Z$ the set of regular zeros of $\varphi(\lambda)$ in $\C$. As discussed in subsection \ref{SubSec:DGAM}, the regularity property of the zeros of the twist function in an integrable $\sigma$-model is related to a general condition used for describing these models as dihedral affine Gaudin models \cite{Vicedo:2017cge}. We shall further distinguish between two types of zeros: cyclotomic ones and 
non-cyclotomic ones. This notion depends on the order $T$ of the automorphism $\sigma$. 
In a model with $T=1$, every point is by definition non-cyclotomic, 
whereas in  a model with $T>1$, every point is non-cyclotomic except for 
the origin and infinity. As explained in subsection \ref{Sec:Infinity}, throughout our analysis the point at infinity will be treated in much the same way as the origin by using an inversion of the spectral parameter.

\medskip

To every $\lambda_0 \in \mathcal Z$, or every $\lambda_0 \in \mathcal Z \cup \{ \infty \}$ if infinity is also a regular zero, we will associate a subset of integers $\mathcal E_{\lambda_0} \subset \Z_{\geq 2}$ and a corresponding tower of local charges $\mathcal Q^{\lambda_0}_n$ labelled by $n \in \mathcal E_{\lambda_0}$. The first main property of these charges which we will establish is that any two such charges $\mathcal Q^{\lambda_0}_n$ and $\mathcal Q^{\mu_0}_m$ for any $\lambda_0, \mu_0 \in \mathcal Z$ and $n \in \mathcal E_{\lambda_0}$, $m \in \mathcal E_{\mu_0}$ are in involution. Moreover, if infinity is a regular zero and either $\lambda_0$ or $\mu_0$ is taken to be the point at infinity, the corresponding local charges will only Poisson commute up to a certain field $\mathcal C(x)$ which will coincide with the coset constraint in $\Z_T$-coset $\sigma$-models. Following the standard terminology from the theory of constrained Hamiltonian systems, we will refer to equalities as being \emph{weak} when they hold only after setting this particular field to zero, see subsection \ref{Sec:AlgebraLoc}. Furthermore, we show that in every example of integrable $\sigma$-model considered, the Hamiltonian can be expressed as a particular linear combination of the collection of quadratic local charges $\mathcal Q^{\lambda_0}_2$ for $\lambda_0 \in \mathcal Z \cup \{ \infty \}$ and the momentum of the model. It then follows that all of the local charges are conserved.

Let us briefly outline the construction of the local charges by considering first the case when $\lambda_0 \in \mathcal Z$ is non-cyclotomic. If the Lie algebra $\g$ is of type B, C or D then the density of the local charge $\mathcal Q^{\lambda_0}_n$ is obtained simply by evaluating
\begin{equation} \label{tr phi L intro}
\Tr \big( \varphi(\lambda)^n \mathcal L(\lambda, x)^n \big)
\end{equation}
at the regular zero $\lambda_0$. When $\g$ is of type A, on the other hand, the density of the local charge $\mathcal Q^{\lambda_0}_n$ is given instead by a certain polynomial in the above expressions, determined as in \cite{Evans:1999mj} with the help of a generating function. In either case, $\mathcal E_{\lambda_0}$ is given here by the set of exponents of the affine Kac-Moody algebra $\widehat{\g}$ associated with $\g$, shifted by one (we do not treat the case of the Pfaffian in type D). In the example of the principal chiral model on a real Lie group $G_0$ treated in \cite{Evans:1999mj}, the twist function has simple zeros at $\pm 1$ and the evaluation of $\varphi(\lambda) \mathcal L(\lambda, x)$ at $\lambda = \pm 1$ produces the currents $j_\pm = g^{-1} \partial_\pm g$ of the theory, where $g$ is the $G_0$-valued field of the principal chiral model and $\partial_\pm$ are the partial derivatives along light-cone coordinates on the worldsheet. We recover in this way the higher spin local charges in involution of the principal chiral model constructed in \cite{Evans:1999mj}.

When the regular zero $\lambda_0 \in \mathcal Z$ is cyclotomic, \emph{i.e.} $\lambda_0 = 0$, it may 
happen, as a result of the equivariance properties of both the Lax matrix and twist function, that 
the evaluation of \eqref{tr phi L intro} at the point $\lambda_0$ vanishes identically. More precisely, 
the first non-vanishing term in the power series expansion of \eqref{tr phi L intro} around 
$\lambda = 0$ is of order $\lambda^{r_n}$ for some $0 \leq r_n \leq T-1$. If the Lie algebra 
$\g$ is of type B, C or D, or also of type A with an inner automorphism $\sigma$, then we 
define the density of the local charge $\mathcal Q^0_n$ as the coefficient of this leading term. The 
case when $\g$ is of type A and the 
automorphism $\sigma$ is not inner is treated in a similar fashion 
to the case of a non-cyclotomic point in type A, with the densities of the local charges $\mathcal Q^0_n$ being obtained by means of a generating function. In each case it turns out that we need to restrict attention to indices $n$ such that $0 \leq r_n < T-1$. As a result, and in contrast to the case of a non-cyclotomic regular zero, some exponents of the affine Kac-Moody algebra $\widehat{\g}$ are `dropped' in the construction of the subset $\mathcal E_0$, specifically those such that $r_n = T-1$. In the case of a symmetric space $\sigma$-model, for which $T=2$ so that only charges for which $r_n = 0$ are kept, we recover in this way the local charges found in \cite{Evans:2000qx}.

\medskip

The collection of local charges $\Q^{\lambda_0}_n$, $\lambda_0 \in \mathcal Z$, $n \in \mathcal E_{\lambda_0}$ in involution generate an infinite set of Poisson commuting Hamiltonian flows $\left\lbrace \Q^{\lambda_0}_n, \cdot \right\rbrace$ on the phase space of the model. To every such flow we then associate a corresponding $\g$-valued connection
$\nabla^{\lambda_0}_n = \left\lbrace \Q^{\lambda_0}_n, \cdot \right\rbrace + \M^{\lambda_0}_n (\lambda,x)$ for some $\g$-valued matrix $\M^{\lambda_0}_n (\lambda,x)$ depending on the spectral parameter $\lambda$.
The second main property of the local charges $\Q^{\lambda_0}_n$, $\lambda_0 \in \mathcal Z$, $n \in \mathcal E_{\lambda_0}$ which we establish is that the connection $\nabla^{\lambda_0}_n$ for any $\lambda_0 \in \mathcal Z$ and $n \in \mathcal E_{\lambda_0}$ commutes with the connection $\nabla_x = \partial_x + \mathcal L(\lambda, x)$. In this sense, the local charges generate a hierarchy of integrable equations. We use this result to deduce that the local charges $\Q^{\lambda_0}_n$, $\lambda_0 \in \mathcal Z$, $n \in \mathcal E_{\lambda_0}$ are in involution with the non-local charges extracted from the monodromy of $\Lc(\lambda,x)$.
Moreover, we go on to show that when $\g$ is of type B, C or D, any two such connections $\nabla^{\lambda_0}_n$ and $\nabla^{\mu_0}_m$ for $\lambda_0, \mu_0 \in \mathcal Z$ and $n \in \mathcal E_{\lambda_0}$, $m \in \mathcal E_{\mu_0}$ also commute with one another. Finally, we have also checked these results in the case of type A for low values of $n$ and $m$ and on this basis we conjecture it to hold in general. If infinity is a regular zero then the majority of these results still hold in the weak sense when we consider also the local charges associated with infinity.

\medskip

This article is organised as follows. The general framework of $r/s$-systems with twist function which we employ throughout the article is reviewed in section \ref{Sec:FraandGenRes}. In particular, we introduce the notion of a regular zero in the complex plane which plays 
a central role in our analysis.  In subsection \ref{Sec:Model}, we define the $\Rc$-matrices entering the $r/s$-systems of interest and discuss their equivariance properties, as well as those of the Lax matrix and the twist function. We will make extensive use of these properties when discussing local charges extracted from cyclotomic regular zeros of the twist function. The list of examples we shall consider is given in paragraph \ref{Sec:Examples}. In subsection \ref{Sec:Infinity}, we define the notion of a regular zero at infinity and relate it to that of a regular zero at the origin by inversion of the spectral parameter. Finally, we establish some general results in subsection \ref{subsec-pbtrpow}. Section \ref{Sec:NonCycZero} is devoted to the procedure for extracting local charges in involution in the case of a non-cyclotomic regular zero. In particular, we present in subsection \ref{Sec:GenNonCyc} an explicit construction of the currents $\K_n^{\lambda_0}$ for type A algebras using generating functions in the spirit of \cite{Evans:1999mj}. Section \ref{Sec:CycZero} deals with charges at cyclotomic zeros. We explain how the equivariance properties of the various objects affect the construction of local conserved charges in involution. Here the Lie algebras of type B, C and D can still be treated 
uniformly but in type A we need to consider separately the cases when the 
automorphism $\sigma$ is inner or not. The generating function for Lie algebras of type A 
with non-inner automorphism is presented in subsection \ref{Sec:CycGenerating}. A list of properties of these local charges is collated in section \ref{Sec:PrOfLocCha}, including the fact that the local charges extracted from different regular zeros Poisson commute (weakly when the point at infinity is involved). Moreover, we show that all the local charges commute with the field $\mathcal C(x)$ which will play the role of the constraint in $\Z_T$-coset $\sigma$-models, therefore showing that they are gauge invariant. We also discuss the reality conditions of all the local charges. The Hamiltonian flows of the local charges $\Q^{\lambda_0}_n$ are studied in detail in 
section \ref{Sec:IntHierZeroCurv}. The main result that any two of the $\g$-valued connections $\nabla^{\lambda_0}_n$ and $\nabla^{\mu_0}_m$ satisfy a zero curvature equation is established in subsection \ref{sec: ZC eq}. Finally, in section \ref{Sec:Applications} we apply all these results to the examples listed in paragraph \ref{Sec:Examples}. In paragraph \ref{SubSec:DGAM} we interpret some of our results and assumptions in terms of dihedral affine Gaudin models \cite{Vicedo:2017cge}. We end with an outlook and two appendices.

\section{Framework and general results} \label{Sec:FraandGenRes}

\subsection{Non-ultralocal models with twist function}
\label{Sec:Model}

\subsubsection{Non-ultralocal models}

In this section we begin by outlining the general framework which we will work in. We consider a two dimensional field theory with spatial coordinate $x$ taking values on the circle $S^1$ or the real line $\R$. The dynamics of the model is described by a Hamiltonian $\Hc$ and a Poisson bracket $\lbrace\cdot,\cdot\rbrace$ on the phase space. Let $\Pc$ denote the conserved momentum of the model, whose Poisson bracket generates the derivative $\p_x$ with respect to $x$ on the phase space.

We suppose that the model is integrable with Lax matrix $\Lc$ valued in a Lie algebra $\g$ and depending on a complex spectral parameter $\lambda$. Moreover, we also suppose that the Poisson bracket of the Lax matrix with itself assumes the form of Maillet's $r/s$-system, \textit{i.e.}
\begin{align}\label{Eq:PBR}
\left\lbrace \Lc(\lambda,x)\ti{1}, \Lc(\mu,y)\ti{2} \right\rbrace & =
\left[ \Rc\ti{12}(\lambda,\mu), \Lc(\lambda,x)\ti{1} \right] \delta_{xy} - \left[ \Rc\ti{21}(\mu,\lambda), \Lc(\mu,y)\ti{2} \right] \delta_{xy} \\
 & \hspace{50pt} - \; \bigl( \Rc\ti{12}(\lambda,\mu) + \Rc\ti{21}(\mu,\lambda) \bigr) \delta'_{xy},  \notag
\end{align}
using the standard tensorial notations $\bm{\underline{i}}$.
In this equation, $\Rc\ti{12}$ is a $\g\otimes\g$-valued matrix depending on the spectral parameters $\lambda$ and $\mu$, $\delta_{xy}$ is the Dirac distribution and $\delta'_{xy}=\p_x \delta_{xy}$. When the $\Rc$-matrix is skew-symmetric, the term containing $\delta'_{xy}$ vanishes and the model is said to be ultralocal. In general, $\Rc$ is non skew-symmetric and the model is then non-ultralocal.

In this article, we will focus on the case where the Lie algebra $\g$ is simple. More precisely, we will restrict to the classical types A, B, C and D of the Cartan classification, seen in their defining representations\footnote{Here $J_n$ is the standard symplectic structure on $\mathbb{C}^{2n}$ given by $J_n = \bigg( \begin{matrix} 0 & \text{Id}\\ - \text{Id} & 0 \end{matrix} \bigg)$ and $M\Tp$ denotes the transpose of $M$.}:
\begin{table}[h]
\begin{center}
\begin{tabular}{lcl}
Type & ~~ & Algebra \\
\hline \hline
A    &    & $\sl(n,\C)=\left\lbrace M \in M_n(\C) \; | \; \Tr(M)=0 \right\rbrace$ \\
B,D  &    & $\so(n,\C)=\left\lbrace M \in M_n(\C) \; | \; M\Tp+M=0 \right\rbrace$ \\
C    &    & $\spc(2n,\C)=\left\lbrace M \in M_{2n}(\C) \; | \; M\Tp J_n + J_n M =0 \right\rbrace$
\end{tabular}
\caption{Defining representations of classical Lie algebras.\label{Tab:Alg}}
\end{center}\vspace{-16pt}
\end{table}

\subsubsection{$\Rc$-matrices and twist functions}
\label{Sec:RMatTwist}

A sufficient condition to ensure the Jacobi identity of the Poisson bracket \eqref{Eq:PBR} is for $\Rc$ to verify the classical Yang-Baxter Equation (CYBE)
\begin{equation}
\label{Eq:CYBE}
\left[ \Rc\ti{12}(\lambda_1,\lambda_2), \Rc\ti{13}(\lambda_1,\lambda_3) \right] + \left[ \Rc\ti{12}(\lambda_1,\lambda_2), \Rc\ti{23}(\lambda_2,\lambda_3) \right] + \left[ \Rc\ti{32}(\lambda_3,\lambda_2), \Rc\ti{13}(\lambda_1,\lambda_3) \right] = 0.
\end{equation}
Let us recall the family of solutions of this equation that we shall consider. Let $T$ be a positive 
integer, $\s:\g\rightarrow\g$ an automorphism of $\g$ of finite order $T$ and $\omega$ a primitive $T^{\rm th}$ root of unity. We fix a basis $T^a$ of $\g$ and its dual basis $T_a$ normalised such that $\Tr(T^a T_b) = \delta^a_b$. We define the Casimir tensor on $\g$ to be
\begin{equation}\label{Eq:Cas}
C\ti{12} = T^a \otimes T_a,
\end{equation}
where here and throughout we use summation convention on repeated Lie algebra indices. The 
standard $\Rc$-matrix
\begin{equation}\label{Eq:RCyc}
\Rc^0\ti{12}(\lambda,\mu) = \frac{1}{T} \sum_{k=0}^{T-1} \frac{\s^k\ti{1}C\ti{12}}{\mu-\omega^{-k}\lambda}
\end{equation}
is then a solution of the CYBE \eqref{Eq:CYBE}. Note that $\s\ti{1}C\ti{12} = \s^{-1}\ti{2}C\ti{12}$, as the Killing form on $\g$ is $\s$-invariant. In this article, we will consider a matrix $\Rc$ obtained by ``twisting'' this cyclotomic matrix, namely
\begin{equation}\label{Eq:DefR}
\Rc\ti{12}(\lambda,\mu) = \Rc^0\ti{12}(\lambda,\mu) \varphi(\mu)^{-1},
\end{equation}
where $\varphi$ is a rational function, called the twist function of the model. It is easy to check that $\Rc$ is also a solution of the CYBE 
\eqref{Eq:CYBE}. We shall call a Poisson bracket \eqref{Eq:PBR} with such an  
$\Rc$-matrix an $r/s$-system with twist function. 

As $\s^T=\Id$, the eigenvalues of $\s$ are of the form $\omega^p$ where, by convention, we take $p\in\lbrace 0, \ldots, T-1 \rbrace$. We denote by $\g^{(p)}$ the corresponding eigenspace and by $\pi^{(p)}$ the projection on $\g^{(p)}$ in the direct sum $\g=\bigoplus_{p=0}^{T-1} \g^{(p)}$. One then has the identity
\begin{equation*}
\pi^{(p)} = \frac{1}{T} \sum_{k=0}^{T-1} \omega^{-kp} \s^k.
\end{equation*}
Defining $C\ti{12}^{(p)}=\pi^{(p)}\ti{1}C\ti{12}=C\ti{21}^{(-p)}$, we can rewrite $\Rc^0$ as
\begin{equation}\label{Eq:RCas}
\Rc^0\ti{12}(\lambda,\mu) = \sum_{p=0}^{T-1} \frac{\lambda^p\mu^{T-1-p}}{\mu^T-\lambda^T} C\ti{12}^{(p)}.
\end{equation}

\subsubsection{Equivariance properties}
\label{Sec:Equi}

As  $\s$ is of order $T$, it defines an action of the cyclic group $\Z_T=\Z/T\Z$ on $\g$. On the other hand, $\Z_T$ can be seen as acting on the complex numbers $\C$ \textit{via} multiplication by $\omega$. We then remark that the matrix $\Rc^0$ is equivariant under these two actions, in the sense that
\begin{equation}\label{Eq:EquiR}
\s\ti{1} \Rc^0\ti{12}(\lambda,\mu) = \Rc^0\ti{12}(\omega\lambda,\mu).
\end{equation}
We will suppose that the Lax matrix $\Lc$ possesses a similar equivariance property, namely
\begin{equation}\label{Eq:EquiL}
\s\bigl( \Lc(\lambda,x) \bigr) = \Lc(\omega\lambda,x).
\end{equation}
The compatibility of these two properties with the Poisson bracket \eqref{Eq:PBR} imposes that
\begin{equation}\label{Eq:TwistEqui}
\varphi(\omega\lambda)=\omega^{-1}\varphi(\lambda),
\end{equation}
from which we deduce that $\lambda\varphi(\lambda)$ is invariant under the action of $\Z_T$. Thus, there exists a rational function $\zeta$ such that
\begin{equation}\label{Eq:DefZeta}
\lambda \varphi(\lambda) = \zeta(\lambda^T).
\end{equation}~

In this article, we will be interested in the zeros of the twist function $\varphi$ in $\C$. We will say that such a zero $\lambda_0$ is \textit{regular} if $\varphi(\lambda)\Lc(\lambda,x)$ is holomorphic at $\lambda=\lambda_0$. By virtue of the equivariance properties \eqref{Eq:EquiL} and \eqref{Eq:TwistEqui}, if $\lambda_0$ is a regular zero, all points of the orbit $\Z_T\lambda_0$ are also regular zeros. Let us pick (arbitrarily) one of them. We then form a set $\Zc$ of regular zeros of $\varphi$ such that for every pair of distinct points $\lambda_0$ and $\mu_0$ in $\Zc$, the orbits $\Z_T\lambda_0$ and $\Z_T\mu_0$ are disjoint. As explained in subsection \ref{Sec:Infinity}, we will  also be interested in the case where the differential form $\varphi(\lambda)\dd \lambda$ has a zero at infinity, \textit{i.e.} where
\begin{equation}\label{Eq:Psi}
\psi(\alpha) = - \frac{1}{\alpha^2}\varphi\left(\frac{1}{\alpha}\right)
\end{equation}
has a zero at $\alpha=0$. We will also see that the appropriate notion of a regular zero at infinity corresponds to requiring that $\frac{1}{\alpha}\varphi\left(\frac{1}{\alpha}\right)\Lc\left(\frac{1}{\alpha},x\right)$ be holomorphic at $\alpha=0$.

\subsubsection{Examples}
\label{Sec:Examples}

To end this subsection we list some examples of models which fit the framework of this article. We consider a real Lie group $G_0$ whose Lie algebra $\g_0$ is a real form of $\g$. The space of fields valued in the cotangent bundle $T^*G_0$ is naturally equipped with a canonical Poisson bracket. One can parametrise this phase space by two fields $g$ and $X$, respectively $G_0$-valued and $\g_0$-valued. The Poisson bracket then reads
\begin{subequations}
\begin{align*}
\left\lbrace g(x)\ti{1}, g(y)\ti{2} \right\rbrace & = 0, \\
\left\lbrace X(x)\ti{1}, g(y)\ti{2} \right\rbrace & = g(x)\ti{2} C\ti{12} \delta_{xy}, \\
\left\lbrace X(x)\ti{1}, X(y)\ti{2} \right\rbrace & = \left[ C\ti{12}, X(x)\ti{1} \right] \delta_{xy}.
\end{align*}
\end{subequations}

\paragraph{Models with $T=1$.}
The first class of examples of non-ultralocal models with twist function consists of the principal chiral 
model (PCM) and its integrable deformations (dPCM). Among these models are the so-called 
Yang-Baxter deformation (or $\eta$-deformation) \cite{Klimcik:2002zj, Klimcik:2008eq, Delduc:2013fga}, the PCM with Wess-Zumino term and the 
combination of these two deformations \cite{Delduc:2014uaa}. The Lax matrix of these models all have the form
\begin{equation}\label{Eq:LaxPCM}
\Lc_{\text{dPCM}}(\lambda,x) = \frac{j_1(x) + \lambda j_0(x)}{1-\lambda^2}.
\end{equation}
The $\g_0$-valued fields $j_0$ and $j_1$ are expressed in terms of $(g,X)$, in a model dependent way. In the simplest case, the PCM, we have
\begin{subequations}
\begin{align*}
j_0 &= g^{-1}Xg, \\
j_1 &= -g^{-1}\p_x g.
\end{align*}
\end{subequations}
For the deformed models, these relations are modified and involve two parameters $\eta$ and $k$ (which are zero for the PCM).

In all these models, the Lax matrix \eqref{Eq:LaxPCM} satisfies the Maillet bracket \eqref{Eq:PBR} with an $\Rc$-matrix 
of the form \eqref{Eq:DefR}. The matrix $\Rc^0$ is given by equation \eqref{Eq:RCyc} for $T=1$ (and thus $\s=\Id$) and 
the twist function is \cite{Delduc:2014uaa}
\begin{equation}\label{Eq:TwistdPCM}
\varphi_{\text{dPCM}}(\lambda) = \frac{1-\lambda^2}{(\lambda-k)^2+A^2},
\end{equation}
where $A$ is a certain function of $\eta$ and $k$, which vanishes when $\eta=0$. The Hamiltonian and the momentum of the model are given by
\begin{subequations}\label{Eq:HamMomPCM}
\begin{align}
\Hc_{\text{dPCM}} &= \frac{B}{2} \int \dd x \; \Tr\bigl( (A^2+k^2+1)(j_0^2 + j_1^2) + 4k j_0j_1 \bigr), \\
\Pc_{\text{dPCM}} &= B \int \dd x \; \Tr\bigl( k(j_0^2 + j_1^2) + (A^2+k^2+1) j_0j_1 \bigr),
\end{align}
\end{subequations}
with $B$ a global factor depending on $A$ and $k$ \textit{via} the relation $B=-\dfrac{1}{4}\varphi'_{\text{dPCM}}(1)\varphi'_{\text{dPCM}}(-1)$. The twist function $\varphi_{\text{dPCM}}$ has two zeros at $+1$ and $-1$. Moreover, these are regular zeros, \textit{i.e.} $\varphi(\lambda)\Lc(\lambda,x)$ is regular at $\lambda=\pm 1$. Note that the evaluation of the latter at $\pm 1$ gives
\begin{equation}\label{Eq:ChiralFieldsdPCM}
J_\pm(x) = \frac{j_1(x) \pm j_0(x)}{(k\pm 1)^2+A^2}.
\end{equation}

There exists another two-parameter deformation of the PCM, the so-called bi-Yang-Baxter model 
\cite{Klimcik:2008eq}. It depends on 
two parameters $\eta$ and $\tilde{\eta}$, such that the case $\eta=\tilde{\eta}=0$ corresponds to the PCM. It is 
integrable \cite{Klimcik:2014bta,Delduc:2015xdm} and has a Lax matrix which satisfies a Poisson bracket 
of Maillet type \eqref{Eq:PBR}. The associated $\Rc$-matrix is not of the form \eqref{Eq:DefR} and so the bi-Yang-Baxter model does not quite fit within the above framework. However, we will see in subsection \ref{Sec:BYB} how the method discussed here can be adapted to apply also to the bi-Yang-Baxter model.

\paragraph{Models with $T>1$.} The second class of examples that we will consider are the $\s$-models 
on $\Z_T$-coset spaces \cite{Young:2005jv}. Consider first the case of symmetric spaces, \textit{i.e.} of $\Z_2$-coset spaces, 
which possess a family of one-parameter deformations (that we shall call d$\Z_2$ models). These models are constructed as follows. Let $\s$ be an involutive automorphism of $G_0$ and let $H=G_0^\s$ denote 
the fixed-point subgroup of $G_0$ under the action of $\s$. The symmetric space is then the quotient $G_0/H$. Differentiating $\s$ at the identity, we induce an automorphism of $\g_0$. We extend it linearly to the complex algebra $\g$ and thus obtain an automorphism of $\g$ of order 2, that we shall still denote $\s$. As a vector space, $\g$ then decomposes as the direct sum of the eigenspaces $\g^{(0)}$ and $\g^{(1)}$, corresponding to the eigenvalues $+1$ and $-1$ (cf. paragraph \ref{Sec:RMatTwist}).

The symmetric space $\s$-model and its deformation are expressed in terms of fields $A^{(k)}$ and $\Pi^{(k)}$ 
valued in the eigenspaces $\g^{(k)}$, for each $k=0, 1$. These fields are defined in terms of the canonical fields
$g$ and $X$ in a model dependent way (in particular, this definition depends on the deformation parameter $\eta$). 
The Lax matrix of the model (deformed or not) is then \cite{Delduc:2013fga}
\begin{equation*}
\Lc_{\dd\Z_2}(\lambda,x)  = A^{(0)}(x) + \frac{1}{2}\left(\frac{1}{\lambda}+\lambda\right)A^{(1)}(x) + \frac{1}{2}(\lambda^2-1) \Pi^{(0)}(x) + \frac{1}{2}\left(\lambda - \frac{1}{\lambda}\right) \Pi^{(1)}(x).
\end{equation*}
The Poisson bracket of this Lax matrix with itself is of Maillet type \eqref{Eq:PBR} with a $\Rc$-matrix of 
the form \eqref{Eq:DefR}. The matrix $\Rc^0$ is given by equation \eqref{Eq:RCyc} with $T=2$ and $\s$ the above 
involutive automorphism of $\g$. The twist function is \cite{Delduc:2013fga}
\begin{equation}\label{Eq:TwistZ2}
\varphi_{\dd\Z_2}(\lambda) = \frac{2\lambda}{(1-\lambda^2)^2 + \eta^2 (1+\lambda^2)^2}.
\end{equation}
The regular zeros of these models are thus the origin $0$ and the infinity $\infty$.\\

To conclude the list of examples let us discuss briefly the general $\Z_T$-coset model, for any $T\in\Z_{\geq 2}$. Let $\s$ be an automorphism of the complexified group $G$, of order $T$. The coset space is then $G_0/H$, where $H=G^\s\cap G_0$. We will still denote by $\s$ the corresponding automorphism of $\g$, which induces an eigenspace decomposition $\g=\bigoplus_{p=0}^{T-1} \g^{(p)}$ (see paragraph \ref{Sec:RMatTwist} for details). The phase space of the model is expressed in terms of fields $A^{(k)}$ and $\Pi^{(k)}$, for $k=0,\ldots,T-1$, which are defined in terms of the canonical fields $g$ and $X$ and belong to $\g^{(k)}$. The Lax matrix is 
then \cite{Ke:2011zzb, Vicedo:2017cge} 
\begin{equation}\label{Eq:LaxZT}
\Lc_{\Z_T}(\lambda,x) = \sum_{k=1}^{T} \frac{(T-k) + k\lambda^{-T}}{T}\lambda^k A^{(k)}(x)  + \sum_{k=1}^{T} \frac{1-\lambda^{-T}}{T} \lambda^k \Pi^{(k)}(x).
\end{equation}
Note that in this equation, and in general, we consider the exponents $(k)$ only modulo $T$, so that $A^{(T)}=A^{(0)}$ for example.

All $\Z_T$-coset models possess a gauge symmetry under the action of the subgroup $H$ of $G_0$. In the Hamiltonian formulation presented above, $A^{(0)}$ plays the role of a gauge field and $\Pi^{(0)}$ is the constraint associated with the gauge symmetry. In those models also, the Poisson bracket of the Lax matrix \eqref{Eq:LaxZT} with itself is non-ultralocal with a twist function. The associated matrix $\Rc^0$ is the one in equation \eqref{Eq:RCyc}, with the automorphism $\s$ introduced above. The twist function 
is \cite{Ke:2011zzb, Vicedo:2017cge}
\begin{equation}\label{Eq:TwistZT}
\varphi_{\Z_T}(\lambda) = \frac{T\lambda^{T-1}}{(1-\lambda^T)^2}.
\end{equation}
As for the $\Z_2$-coset, the regular zeros are $0$ and $\infty$.

\subsection{Infinity and inversion of the spectral parameter}
\label{Sec:Infinity}

In this article we will construct a tower of local charges associated with each regular zero of the twist function. As mentioned in paragraph \ref{Sec:Equi}, the set of regular zeros can include the point at infinity, although the sense in which infinity can be a regular zero is slightly different from the definition of finite regular zeros. In this subsection, we show how the notion of a regular zero at infinity is related to that of a regular zero at the origin through inversion of the spectral parameter, \textit{i.e.} by the change of parameter $\lambda \mapsto \alpha=\lambda^{-1}$.
Under such a change of spectral parameter we have
\begin{equation*}
\varphi(\lambda) \dd\lambda = \psi(\alpha)\dd\alpha,
\end{equation*}
where $\psi(\alpha)$ is defined in equation \eqref{Eq:Psi}. Suppose that infinity is a zero of the twist function, \textit{i.e.} that $\psi(0)=0$, and define
\begin{equation}\label{Eq:DefP}
P(\alpha,x)=\frac{1}{\alpha}\varphi\left(\frac{1}{\alpha}\right)\Lc\left(\frac{1}{\alpha},x\right).
\end{equation}
We will say that infinity is a \textit{regular zero} if $P(\alpha,x)$ is regular at $\alpha=0$. In the remainder of this subsection we will assume this to be the case. We then set
\begin{equation}\label{Eq:DefC}
\Cc(x) = P(0,x).
\end{equation}
From the equivariance properties \eqref{Eq:EquiL} and \eqref{Eq:TwistEqui} of $\Lc$ and $\varphi$, we deduce that $\Cc$ is valued in the grading $\g^{(0)}$. Let us note here that in the $\Z_T$-coset models, described in paragraph \ref{Sec:Examples}, this field $\Cc$ coincides with the gauge constraint $\Pi^{(0)}$.

Starting from the Poisson bracket \eqref{Eq:PBR} and using the form \eqref{Eq:DefR} of the $\Rc$-matrix, we find
\begin{align*}
 \left\lbrace \Lc(\lambda,x)\ti{1}, P(\alpha,y)\ti{2} \right\rbrace & =  \left[\alpha^{-1}\Rc^0\ti{12}\left(\lambda,\alpha^{-1}\right), \Lc(\lambda,x)\ti{1} \right] \delta_{xy} - \left[ \Rc^0\ti{21}\left(\alpha^{-1},\lambda \right)\varphi(\lambda)^{-1}, P(\alpha,x)\ti{2} \right] \delta_{xy} \\
 & \hspace{50pt} - \Bigl( \alpha^{-1}\Rc^0\ti{12}\left(\lambda,\alpha^{-1}\right) - \alpha\psi(\alpha)\Rc^0\ti{21}\left(\alpha^{-1},\lambda \right) \varphi(\lambda)^{-1} \Bigr)  \delta'_{xy}
\end{align*}
Using the expression \eqref{Eq:RCyc} of $\Rc^0$, we have
\begin{equation}\label{Eq:RAsymptoticInfinity}
\alpha^{-1}\Rc^0\ti{12}\left(\lambda,\alpha^{-1}\right) \;\xrightarrow{\alpha\to 0}\; \frac{1}{T} \sum_{k=0}^{T-1} \s^k\ti{1}C\ti{12} = C^{(0)}\ti{12}, \;\;\;\;\; \Rc^0\ti{21}\left(\alpha^{-1},\lambda \right) \;\xrightarrow{\alpha\to 0}\; 0.
\end{equation}
As $P(\alpha,x)$ and $\alpha\psi(\alpha)$ are regular at 0, taking the limit $\alpha\to 0$ in the above Poisson bracket, we then obtain
\begin{equation}\label{Eq:PBLC}
\left\lbrace \Lc(\lambda,x)\ti{1}, \Cc(y)\ti{2} \right\rbrace = \bigl[ C^{(0)}\ti{12}, \Lc(\lambda,x)\ti{1} \bigr] \delta_{xy} - C^{(0)}\ti{12}\delta'_{xy}.
\end{equation}
Applying the same kind of reasoning we also find
\begin{equation}\label{Eq:PBCC}
\left\lbrace \Cc(x)\ti{1}, \Cc(y)\ti{2} \right\rbrace = \bigl[ C^{(0)}\ti{12}, \Cc(x)\ti{1} \bigr] \delta_{xy}.
\end{equation}
Let us define a new Lax matrix
\begin{equation*}
\Lct(\lambda,x) = \Lc(\lambda,x) - \lambda^{-1}\varphi(\lambda)^{-1}\Cc(x).
\end{equation*}
From the fact that $[ C^{(k)}\ti{12}, Z\ti{1} ] = -[ C^{(k)}\ti{12}, Z\ti{2} ]$ for any $Z\in\g^{(0)}$, we find that
\begin{equation*}
\lambda \left[ \Rc^0\ti{21}(\mu,\lambda), Z\ti{2} \right] - \mu \left[ \Rc^0\ti{12}(\lambda,\mu), Z\ti{1} \right]  = \left[ C^{(0)}\ti{12}, Z\ti{2} \right].
\end{equation*}
Using this identity and the Poisson brackets \eqref{Eq:PBR}, \eqref{Eq:PBLC} and \eqref{Eq:PBCC}, we prove that the Poisson bracket of $\Lct$ with itself is also of the $r/s$-form, namely
\begin{align}\label{Eq:PBRt}
\left\lbrace \Lct(\lambda,x)\ti{1}, \Lct(\mu,y)\ti{2} \right\rbrace &=
\left[ \Rct\ti{12}(\lambda,\mu), \Lct(\lambda,x)\ti{1} \right] \delta_{xy} - \left[ \Rct\ti{21}(\mu,\lambda), \Lct(\mu,y)\ti{2} \right] \delta_{xy} \\
 &\hspace{50pt}  - \; \Bigl( \Rct\ti{12}(\lambda,\mu) + \Rct\ti{21}(\mu,\lambda) \Bigr) \delta'_{xy}, \notag
\end{align}
where $\Rct\ti{12}(\lambda,\mu)=\Rct^0\ti{12}(\lambda,\mu)\varphi(\mu)^{-1}$ and
\begin{equation}\label{Eq:DefRct}
\Rct^0\ti{12}(\lambda,\mu) = \Rc^0\ti{12}(\lambda,\mu) - \mu^{-1} C^{(0)}\ti{12}.
\end{equation}
We now define
\begin{equation*}
\Lc^\infty(\alpha,x) = \Lct\left(\frac{1}{\alpha},x\right).
\end{equation*}
The following theorem is the main result of this subsection.
\begin{theorem}\label{Thm:PBLcI}
The Poisson bracket of $\Lc^\infty$ with itself reads
\begin{align}\label{Eq:PBLcI}
\left\lbrace \Lc^\infty(\alpha,x)\ti{1}, \Lc^\infty(\beta,y)\ti{2} \right\rbrace &=
\left[ \Rc^\infty\ti{12}(\alpha,\beta), \Lc^\infty(\alpha,x)\ti{1} \right] \delta_{xy} - \left[ \Rc^\infty\ti{21}(\beta,\alpha), \Lc^\infty(\beta,y)\ti{2} \right] \delta_{xy} \\
 &\hspace{50pt}  - \; \Bigl( \Rc^\infty\ti{12}(\alpha,\beta) + \Rc^\infty\ti{21}(\beta,\alpha) \Bigr) \delta'_{xy}, \notag
\end{align}
where
\begin{equation*}
\Rc^\infty\ti{12}(\alpha,\beta) = \Rc^0\ti{21}(\alpha,\beta)\psi(\beta)^{-1}
\end{equation*}
satisfies the classical Yang-Baxter equation \eqref{Eq:CYBE}.
\end{theorem}
\begin{proof}
Using equation \eqref{Eq:RCas}, we find that
\begin{equation}\label{Eq:RtCas}
\Rct^0\ti{12}(\lambda,\mu) = \sum_{k=1}^{T} \frac{\lambda^k\mu^{T-1-k}}{\mu^T-\lambda^T} C\ti{12}^{(p)}.\end{equation}
The theorem follows from the Poisson bracket \eqref{Eq:PBRt} and the identity
\begin{equation}\label{Eq:RtInvR}
\Rct^0\ti{12}\left(\frac{1}{\alpha},\frac{1}{\beta}\right)= - \beta^2 \Rc^0\ti{21}(\alpha,\beta),
\end{equation}
which is a consequence of equation \eqref{Eq:RtCas}.
\end{proof}

To interpret Theorem \ref{Thm:PBLcI}, let us note that the matrix $\Rc^0\ti{21}$ is nothing but the matrix $\Rc^0\ti{12}$ for the automorphism $\s^{-1}$. Moreover, from the equivariance properties \eqref{Eq:EquiL} and \eqref{Eq:TwistEqui}, we find that the corresponding properties of $\Lc^\infty$ and $\psi$ are
\begin{equation}\label{Eq:EquiInfinity}
\s^{-1}\bigl(\Lc^\infty(\alpha,x)\bigr) = \Lc^\infty(\omega\alpha,x)
\;\;\;\;\;\;\; \text{ and } \;\;\;\;\;\;\;
\psi(\omega\alpha)=\omega^{-1}\psi(\alpha).
\end{equation}
The Poisson bracket of $\Lc^\infty$ is thus an $r/s$-system with twist function $\psi$, automorphism $\s^{-1}$ and spectral parameter $\alpha=\lambda^{-1}$. Moreover, the point $\alpha=0$ is a regular zero of this $r/s$-system. Indeed, we supposed that $\alpha$ was a zero of $\psi(\alpha)$ and one can check explictly that $\psi(\alpha)\Lc^\infty(\alpha,x)$ is regular at $\alpha=0$.\\

It is worth noting that the procedure just described is involutive, in the following sense. If $\varphi(\lambda)\Lc(\lambda,x)$ is regular at $\lambda=0$, one can check that $\alpha=\infty$ (which corresponds to $\lambda=0$) is a regular zero of the $r/s$-system of $\Lc^\infty$ and, moreover, that the corresponding field $\Cc^\infty$ obtained by evaluating $\lambda^{-1}\psi(\lambda^{-1})\Lc^\infty(\lambda^{-1},x)$ at $\lambda=0$ is equal to $\Cc$. Re-inverting the spectral parameter $\alpha$ to $\lambda=\alpha^{-1}$, we can thus construct a ``new'' Lax matrix $\Lc^\infty(\lambda^{-1},x)-\lambda \psi(\lambda^{-1})^{-1}\Cc(x)$. According to Theorem \ref{Thm:PBLcI}, this Lax matrix should satisfy an $r/s$-system with twist function $\varphi$ and automorphism $\s$. A direct computation reveals that this Lax matrix is actually equal to the initial Lax matrix $\Lc$.\\

Let us end this subsection by illustrating the inversion of spectral parameter on the example of $\Z_T$-coset models. As noted above, for these models the field $\Cc$ coincides with the constraint $\Pi^{(0)}$. After performing the change of spectral parameter $\lambda \mapsto \alpha=\lambda^{-1}$, we find a twist function
\begin{equation*}
\psi_{\Z_T}(\alpha) = -\frac{T\alpha^{T-1}}{(1-\alpha^T)^2} = -\varphi_{\Z_T}(\alpha).
\end{equation*}
Note that the property $\psi(\alpha)=-\varphi(\alpha)$ is also true for the twist function \eqref{Eq:TwistZ2} of the $\eta$-deformed $\Z_2$-model. The new Lax matrix is
\begin{equation*}
\Lc^\infty_{\Z_T}(\alpha,x) = \sum_{k=1}^{T} \frac{(T-k) + k\alpha^{-T}}{T}\alpha^k A^{(T-k)}(x) - \sum_{k=1}^{T} \frac{1-\alpha^{-T}}{T} \alpha^k \Pi^{(T-k)}(x).
\end{equation*}
Comparing this to the initial Lax matrix \eqref{Eq:LaxZT}, we see that it simply corresponds (up to a minus sign on terms involving $\Pi^{(k)}$) to changing every grading $(k)$ to $(T-k)$, which is equivalent to considering the automorphism $\s^{-1}$ instead of $\s$.

\subsection{Poisson brackets of traces of powers of $\Lc$}
\label{subsec-pbtrpow}
Recall that we consider the Lie algebra $\g$ in its defining matrix representation (see Table \ref{Tab:Alg}). We may therefore take powers of elements of $\g$ and traces of these matrices. In the following sections, we will extract local charges in involution from the traces of powers of the Lax matrix $\Lc$. In this subsection, we will establish general results on the Poisson brackets of powers of $\Lc$ and their traces.

\begin{lemma}\label{Lem:PBPow}
Suppose that $X$ and $Y$ are $\g$-valued quantities such that
\begin{equation*}
\left\lbrace X\ti{1}, Y\ti{2} \right\rbrace = \left[a\ti{12},X\ti{1}\right] + \left[b\ti{12},Y\ti{2} \right] + c\ti{12}.
\end{equation*}
Then the Poisson brackets of powers of $X$ and $Y$ are
\begin{equation*}
\left\lbrace X^n\ti{1}, Y^m\ti{2} \right\rbrace = \left[a^{(nm)}\ti{12},X\ti{1}\right] + \left[b^{(nm)}\ti{12},Y\ti{2} \right] + c^{(nm)}\ti{12},
\end{equation*}
where, for $t=a,b,c$, we defined
\begin{equation*}
t^{(nm)}\ti{12} = \sum_{k=0}^{n-1} \sum_{l=0}^{m-1} X\ti{1}^k Y\ti{2}^l \, t\ti{12} \, X\ti{1}^{n-1-k} Y\ti{2}^{m-1-l}.
\end{equation*}
\end{lemma}
\begin{proof}
The Poisson bracket being a derivation, we can use the Leibniz rule yielding
\begin{equation*}
\left\lbrace X^n\ti{1}, Y^m\ti{2} \right\rbrace = \sum_{k=0}^{n-1} \sum_{l=0}^{m-1} X\ti{1}^k Y\ti{2}^l \, \left\lbrace X\ti{1}, Y\ti{2} \right\rbrace \, X\ti{1}^{n-1-k} Y\ti{2}^{m-1-l}.
\end{equation*}
We conclude observing that $ X\ti{1}^k Y\ti{2}^l $ and $X\ti{1}^{n-1-k} Y\ti{2}^{m-1-l}$ commute with $X\ti{1}$ and $Y\ti{2}$ and using the identity
\begin{equation*}
M_1[M_2,N]M_3 = [M_1M_2M_3,N],
\end{equation*}
true for any matrices $M_1$, $M_2$, $M_3$ and $N$ such that $[M_1,N]=[M_3,N]=0$.
\end{proof}

\begin{corollary}\label{Cor:PBTr}
Suppose that $X$ and $Y$ are $\g$-valued quantities such that
\begin{equation*}
\left\lbrace X\ti{1}, Y\ti{2} \right\rbrace = \left[a\ti{12},X\ti{1}\right] + \left[b\ti{12},Y\ti{2} \right] + c\ti{12}.
\end{equation*}
Then we have
\begin{subequations}
\begin{align*}
\bigl\lbrace \emph{\Tr}(X^n), \emph{\Tr}(Y^m) \bigr\rbrace &= nm \, \emph{\Tr}\ti{12} \bigl( c\ti{12}X^{n-1}\ti{1}Y^{m-1}\ti{2} \bigr), \\
\bigl\lbrace X, \emph{\Tr}(Y^m) \bigr\rbrace &= m\left[ \emph{\Tr}\ti{2}\bigl(a\ti{12}Y^{m-1}\ti{2}\bigr), X \right]  + m \, \emph{\Tr}\ti{2}\bigl(c\ti{12}Y^{m-1}\ti{2}\bigr).
\end{align*}
\end{subequations}
\end{corollary}
\begin{proof}
Starting with Lemma \ref{Lem:PBPow}, the corollary follows from the cyclicity of the trace and the vanishing of traces of commutators.
\end{proof}

Let us now apply these results to the Lax matrix $\Lc$. We work in the framework described in subsection \ref{Sec:Model}. We define
\begin{equation}\label{Eq:DefS}
S_n(\lambda,x) = \varphi(\lambda)^n \Lc(\lambda,x)^n
\end{equation}
and
\begin{equation}\label{Eq:DefT}
\Tc_n(\lambda,x) = \Tr\bigl( S_n(\lambda,x) \bigr).
\end{equation}
Starting with the Poisson bracket \eqref{Eq:PBR} and the expression \eqref{Eq:DefR} of the $\Rc$-matrix, we apply Corollary \ref{Cor:PBTr}. We find that
\begin{equation}\label{Eq:PBT}
\left\lbrace \Tc_n(\lambda,x), \Tc_m(\mu,y) \right\rbrace = -nm \, \Tr\ti{12} \Bigl( U\ti{12}(\lambda,\mu) S_{n-1}(\lambda,x)\ti{1}S_{m-1}(\mu,y)\ti{2} \Bigr) \delta'_{xy},
\end{equation}
with
\begin{equation}\label{Eq:DefU}
U\ti{12}(\lambda,\mu) = \varphi(\lambda)\Rc^0\ti{12}(\lambda,\mu) + \varphi(\mu)\Rc^0\ti{21}(\mu,\lambda).
\end{equation}

\section{Charges at non-cyclotomic zeros}
\label{Sec:NonCycZero}

The purpose of this section is to describe the procedure for extracting local charges in involution from non-cyclotomic regular zeros of the twist function $\varphi$. Let us first explain what we mean here by a \emph{non-cyclotomic} point. If $T=1$, \textit{i.e.} if $\s=\Id$ and there 
is no cyclotomic invariance, we define any point as being non-cyclotomic. If $T\in\Z_{>1}$, a non-cyclotomic point is a point which is not fixed by the action of the cyclic group $\Z_T$, \textit{i.e.} which is not the origin or infinity.\\

Throughout this section we fix a non-cyclotomic regular zero $\lambda_0$. We will focus here on the case where $\lambda_0$ is different from infinity. The case $\lambda_0=\infty$ is treated by the same method, just replacing $\Lc$ by $\Lc^\infty$ and $\varphi$ by $\psi$ (cf. subsection \ref{Sec:Infinity}). The fact that $\lambda_0$ is a regular zero implies that $S_n(\lambda,x)$ and $\Tc_n(\lambda,x)$, defined in equations \eqref{Eq:DefS} and \eqref{Eq:DefT}, are both holomorphic at $\lambda=\lambda_0$. Thus, we can define the current
\begin{equation}\label{Eq:DefJNonCyc}
\J_n^{\lambda_0}(x) = \Tc_n(\lambda_0,x).
\end{equation}

Let us briefly comment on the explicit expression of these currents in the case of the PCM. As explained in paragraph \ref{Sec:Examples}, the PCM has two regular zeros at $+1$ and $-1$. The corresponding currents are
\begin{equation*}
\J_{n,\text{PCM}}^{\pm 1}(x) = \Tr\bigl(j_\pm^n(x)\bigr),
\end{equation*}
where $j_\pm(x)=\frac{1}{2}\bigl(j_1(x) \pm j_0(x)\bigr)$. These currents are the one investigated in \cite{Evans:1999mj}, from which local charges in involution for the PCM are constructed. In this section, we will follow the method developed in \cite{Evans:1999mj}, generalising it to any current \eqref{Eq:DefJNonCyc} associated with a non-cyclotomic regular zero $\lambda_0$ of the model.

\subsection{Poisson algebra of the currents}

We begin by computing the Poisson bracket of the currents $\J_n^{\lambda_0}(x)$ and $\J_m^{\lambda_0}(y)$. Specifically, we would like to evaluate equation \eqref{Eq:PBT} at $\lambda=\mu=\lambda_0$. Since $\lambda_0$ is a regular zero, $S_{n-1}(\lambda_0,x)$ and $S_{m-1}(\lambda_0,y)$ are well defined. Thus, it remains to determine $U\ti{12}(\lambda_0,\lambda_0)$. Starting with the definition \eqref{Eq:DefU} of $U$ and using $\varphi(\lambda_0)=0$, one has
\begin{equation*}
U\ti{12}(\lambda,\lambda_0) = \varphi(\lambda)\Rc^0\ti{12}(\lambda,\lambda_0).
\end{equation*}
Recall from equation \eqref{Eq:RCyc} that $\Rc^0\ti{12}(\lambda,\lambda_0)$ is not regular at $\lambda=\lambda_0$, so that we cannot simply evaluate the above equation at $\lambda=\lambda_0$. However, as $\lambda_0$ is a non-cyclotomic point, the matrix $\Rc^0$ has the following local behaviour
\begin{equation}\label{Eq:RAsymptotic}
\Rc^0\ti{12}(\lambda,\lambda_0) = -\frac{1}{T}\frac{C\ti{12}}{\lambda-\lambda_0} + A^{\lambda_0}\ti{12}(\lambda),
\end{equation}
where $A^{\lambda_0}\ti{12}(\lambda)$ is regular at $\lambda=\lambda_0$. Using again $\varphi(\lambda_0)=0$, we then obtain
\begin{equation*}
U\ti{12}(\lambda_0,\lambda_0) = -\frac{\varphi'(\lambda_0)}{T} C\ti{12},
\end{equation*}
where $\varphi'$ denotes the derivative of $\varphi$ with respect to the spectral parameter $\lambda$. Thus, one has
\begin{equation}\label{Eq:PBJ}
\left\lbrace \J^{\lambda_0}_n(x), \J^{\lambda_0}_m(y) \right\rbrace = \frac{nm}{T} \varphi'(\lambda_0) \, \Tr\ti{12} \Bigl( C\ti{12} S_{n-1}(\lambda_0,x)\ti{1}S_{m-1}(\lambda_0,y)\ti{2} \Bigr) \delta'_{xy}.
\end{equation}
Recall the completeness relation
\begin{equation}\label{Eq:CompRel}
\Tr\ti{2}( C\ti{12} Z\ti{2} ) = Z,
\end{equation}
for $Z$ in $\g$. We cannot directly apply this identity to equation \eqref{Eq:PBJ} as $S_{m-1}(\lambda_0,y)$ does not belong to $\g$ in general (recall that $S_{m-1}$ is defined as the $(m-1)^{\rm st}$ power of a matrix in $\g$).

Following~\cite{Evans:1999mj}, we will show in the next subsections how to circumvent this difficulty. We will treat separately the case where $\g$ is of type B, C or D and the case where $\g$ is of type A.

\subsection{Type B, C and D algebras}
\label{Sec:NonCycZeroBCD}

Let us first consider the case where $\g$ is of type B, C or D, \textit{i.e.} where $\g$ is an orthogonal or a symplectic algebra (cf. Table \ref{Tab:Alg}). One can check that, for these algebras, if $X$ belongs to $\g$, $X^n$ also belongs to $\g$ if $n$ is odd. Moreover, all matrices in $\g$ are traceless. We then deduce that the currents $\J_n^{\lambda_0}$ are zero for $n$ odd. Thus, we will only extract local charges from the traces of even powers of $\Lc$, \textit{i.e.} from the currents $\J^{\lambda_0}_{2n}$.

The Poisson bracket of such currents is given by equation \eqref{Eq:PBJ}. The right hand side contains $\Tr\ti{2}\bigl(C\ti{12} S_{2m-1}(\lambda_0,y)\ti{2} \bigr)$, and since $2m-1$ is odd we have $S_{2m-1}(\lambda_0,y) \in \g$. Hence, we can apply the completeness relation \eqref{Eq:CompRel}, which yields
\begin{equation*}
\left\lbrace \J^{\lambda_0}_{2n}(x), \J^{\lambda_0}_{2m}(y) \right\rbrace = 4nm\frac{\varphi'(\lambda_0)}{T} \, \Tr \Bigl( S_{2n-1}(\lambda_0,x) S_{2m-1}(\lambda_0,y) \Bigr) \delta'_{xy}.
\end{equation*}
Using the definition \eqref{Eq:DefS} of $S$, one has
\begin{equation}\label{Eq:DerS}
\Tr\bigl(S_p(\lambda,x)\p_xS_q(\lambda,x)\bigr) = \frac{q}{p+q} \p_x \Tc_{p+q}(\lambda,x).
\end{equation}
Using the identities $f(y)\delta'_{xy}=\p_x \bigl(f(x)\bigr)\delta_{xy}+f(x)\delta'_{xy}$ and \eqref{Eq:DerS}, we obtain
\begin{equation}\label{Eq:PBJTypeBCD}
\left\lbrace \J^{\lambda_0}_{2n}(x), \J^{\lambda_0}_{2m}(y) \right\rbrace = 4nm\frac{\varphi'(\lambda_0)}{T} \left( \J^{\lambda_0}_{2n+2m-2}(x) \delta'_{xy} + \frac{2m-1}{2n+2m-1} \p_x \bigl( \J^{\lambda_0}_{2n+2m-2}(x) \bigr) \delta_{xy} \right).
\end{equation}
Define the local charges
\begin{equation}\label{Eq:DefQJ}
\Q^{\lambda_0}_{2n} = \int \dd x \; \J^{\lambda_0}_{2n}(x),
\end{equation}
where the integration is over the whole domain of the spatial coordinate $x$ (\textit{i.e.} the real line $\R$ or the circle $S^1$).
Once integrated over $y$, the right hand side of \eqref{Eq:PBJTypeBCD} is a total derivative with respect to $x$. Assuming the periodicity of the fields if $x\in S^1$ or that they decrease at infinity if $x\in\R$, we then conclude that
\begin{equation*}
\left\lbrace \Q^{\lambda_0}_{2n}, \Q^{\lambda_0}_{2m} \right\rbrace = 0.
\end{equation*}
In conclusion, we have constructed a tower of local charges $\Q^{\lambda_0}_{2n}$ in involution, as integrals of the currents $\J^{\lambda_0}_{2n}(x)$. These currents are polynomials in the fields appearing in the Lax matrix $\Lc(\lambda,x)$. More precisely, the current $\J^{\lambda_0}_{2n}$ is a homogeneous polynomial of degree $2n$.\\

Up to a global factor, the Poisson bracket \eqref{Eq:PBJTypeBCD} is the same as the bracket (4.16) of~\cite{Evans:1999mj}. Thus, we can apply the methods developed in~\cite{Evans:1999mj}. In particular, this allows to construct a more general tower of local charges $\Q^{\lambda_0}_{2n}(\xi)$ in involution, depending on a free parameter $\xi\in\R$. These charges are defined as integrals
\begin{equation*}
\Q^{\lambda_0}_{2n}(\xi) = \int \dd x \; \K_{2n}^{\lambda_0}(\xi,x)
\end{equation*}
of some currents $\K_{2n}^{\lambda_0}(\xi)$. These currents are given by homogeneous polynomials in the $\J^{\lambda_0}_{2k}$'s, depending on the free parameter $\xi\in\R$. In particular, the first currents $\K_{2n}^{\lambda_0}(\xi)$ are given by:
\begin{align}\label{Eq:KAlpha}
&\K^{\lambda_0}_2(\xi) = \J^{\lambda_0}_2, \;\;\;\;\;\; \K^{\lambda_0}_4(\xi)=\J^{\lambda_0}_4-\frac{3\xi}{2}(\J^{\lambda_0}_2)^2,  \notag \\
&\K^{\lambda_0}_6(\xi) = \J^{\lambda_0}_6 - \frac{15\xi}{4}\J^{\lambda_0}_2\J^{\lambda_0}_4 + \frac{25\xi^2}{8}(\J^{\lambda_0}_2)^3.
\end{align}
The expression of the current $\K_{2n}^{\lambda_0}(\xi)$ is determined (up to a global factor) recursively from equation \eqref{Eq:PBJTypeBCD} by demanding that the charge $\Q_{2n}^{\lambda_0}(\xi)$ be in involution with all the charges $\Q_{2m}^{\lambda_0}(\xi)$ ($m=2,\ldots,n-1$) constructed thus far. It can also be found without recursion with the help of a generating function, which allows a general proof of the involution of the charges $\Q_{2n}^{\lambda_0}(\xi)$: we refer the reader to the subsection \ref{Sec:GenNonCyc} for more details.

Taking $\xi=0$ in equation \eqref{Eq:KAlpha}, we get $\K_{2n}^{\lambda_0}(\xi=0)=\J_{2n}^{\lambda_0}$. Hence, we recover the local charges $\Q^{\lambda_0}_{2n}$ introduced in equation \eqref{Eq:DefQJ} as a special case of this one-parameter family of local charges. For different parameters $\xi$ and $\xi'$, the towers of charges $\Q_{2n}^{\lambda_0}(\xi)$ and $\Q_{2n}^{\lambda_0}(\xi')$ 
do not commute with one another. We thus have to work with a fixed value of $\xi$: in the rest of this article, we will mainly focus on the simplest case $\xi=0$. This choice is justified first by simplicity, but also because the proof of the existence of an integrable hierarchy associated to the charges $\Q_{2n}^{\lambda_0}(\xi)$, presented in section \ref{Sec:IntHierZeroCurv}, works only for the case $\xi=0$.

\subsection{Type A algebras}
\label{Sec:TypeANonCyc}

Let us now consider the case where $\g$ is of type A, \textit{i.e.} where $\g=\sl(d,\C)$ for some $d \in \Z_{\geq 2}$ (see Table \ref{Tab:Alg}). If $X\in\g$, we have $\Tr(X)=0$ by definition, but in general $X^n\notin\g$ and $\Tr(X^n)\neq 0$ for $n \geq 2$. Thus, we consider the currents $\J^{\lambda_0}_n$ for $n\geq 2$. The Poisson bracket between two such currents is given by equation \eqref{Eq:PBJ}. Since in general $S_{m-1}(\lambda_0,y)$ does not belong to $\g$, we cannot use the completeness relation \eqref{Eq:CompRel} to simplify this equation. However, a variant of the identity \eqref{Eq:CompRel} exists for any matrix $Z\in M_d(\C)$. Indeed, using the facts that $Z-\frac{1}{d}\Tr(Z)\Id$ belongs to $\g$ and that $\Tr\ti{2}(C\ti{12})=0$, we find that
\begin{equation}\label{Eq:CompRelSl}
\Tr\ti{2}\bigl(C\ti{12}Z\ti{2}\bigr) = Z - \frac{1}{d}\Tr(Z)\Id.
\end{equation}
Applying this relation to equation \eqref{Eq:PBJ} and using the identities $f(y)\delta'_{xy}=\p_x \bigl(f(x)\bigr)\delta_{xy}+f(x)\delta'_{xy}$ and \eqref{Eq:DerS}, we obtain
\begin{align}\label{Eq:PBJTypeA}
\left\lbrace \J_n^{\lambda_0}(x), \J_m^{\lambda_0}(y) \right\rbrace
&=  nm \frac{\varphi'(\lambda_0)}{T} \left( \J_{n+m-2}^{\lambda_0}(x) \delta'_{xy}  - \frac{1}{d}\J_{n-1}^{\lambda_0}(x)\J_{m-1}^{\lambda_0}(x)\delta'_{xy} \right.  \\
 & \hspace{30pt} \left. + \frac{m-1}{n+m-2} \p_x \left( \J_{n+m-2}^{\lambda_0}(x) \right) \delta_{xy} - \frac{1}{d}\J_{n-1}^{\lambda_0}(x) \p_x \left( \J_{m-1}^{\lambda_0}(x) \right) \delta_{xy} \right). \notag
\end{align}
Integrating both sides over $x$ and $y$, we see that the right hand side does not vanish identically as it did in subsection \ref{Sec:NonCycZeroBCD}. Nevertheless, following the method of~\cite{Evans:1999mj} we will be able to construct new currents $\K_n^{\lambda_0}$ such that the charges
\begin{equation}\label{Eq:DefQK}
\Q^{\lambda_0}_n = \int \dd x \; \K_n^{\lambda_0}(x)
\end{equation}
Poisson commute with one another.\\

The Poisson bracket \eqref{Eq:PBJTypeA} is to be compared to equation (4.5) of~\cite{Evans:1999mj}, from which it differs only by an overall factor. We can therefore directly apply the procedure developed in~\cite{Evans:1999mj} to the present case so as to construct the desired currents $\K^{\lambda_0}_n$'s. The expression for the first $\K^{\lambda_0}_n$'s read
\begin{align}\label{Eq:KJ}
&\K^{\lambda_0}_2 = \J^{\lambda_0}_2, \;\;\; \K^{\lambda_0}_3=\J^{\lambda_0}_3, \;\;\; \K^{\lambda_0}_4=\J^{\lambda_0}_4-\frac{3}{2d}(\J^{\lambda_0}_2)^2,  \;\;\; \K^{\lambda_0}_5 = \J^{\lambda_0}_5 - \frac{10}{3d} \J^{\lambda_0}_2\J^{\lambda_0}_3, \notag \\
&\K^{\lambda_0}_6 = \J^{\lambda_0}_6 - \frac{5}{3d}(\J^{\lambda_0}_3)^2 - \frac{15}{4d}\J^{\lambda_0}_2\J^{\lambda_0}_4 + \frac{25}{8d^2}(\J^{\lambda_0}_2)^3.
\end{align}
These currents are similar to the currents $\K^{\lambda_0}_n(\xi)$ described in \eqref{Eq:KAlpha} for $\g$ of type B, C or D. More precisely, the current \eqref{Eq:KJ} coincide with the currents $\K^{\lambda_0}_n\left(\frac{1}{d}\right)$, recalling that for type B, C and D, the $\J^{\lambda_0}_{2k+1}$'s vanish. As for $\K^{\lambda_0}_n(\xi)$ in type B, C and D, the expression of the current $\K_n^{\lambda_0}$ for type A is determined (up to a global factor) recursively from equation \eqref{Eq:PBJTypeA} by demanding that the charge $\Q_n^{\lambda_0}$ be in involution with all the charges $\Q_m^{\lambda_0}$ ($m=2,\ldots,n-1$) constructed thus far. However, in the present case, one does not have the freedom of a free parameter $\xi$ in the definition of $\K_n^{\lambda_0}$: there is a unique tower of charges in involution $\Q^{\lambda_0}_n$.

As in the case of type B, C and D algebras, the current $\K_n^{\lambda_0}(x)$ is a homogeneous polynomial of degree $n$ in the fields appearing in the Lax matrix $\Lc(\lambda,x)$. And as explained in~\cite{Evans:1999mj}, the degrees $n$ for which the current $\K_n^{\lambda_0}(x)$ is non-zero are the exponents of the untwisted affine Kac-Moody algebra $\widehat{\g}$ plus one.

At this stage, we do not have a proof that the recursive algorithm described above can be applied indefinitely. We shall now recall from~\cite{Evans:1999mj} how to construct explicitly the current $\K_n^{\lambda_0}$ without a recursive algorithm, using generating functions.

\subsection{Generating functions}
\label{Sec:GenNonCyc}

In the previous subsections \ref{Sec:NonCycZeroBCD} and \ref{Sec:TypeANonCyc}, we introduced currents $\K_n^{\lambda_0}(\xi)$ (for types B, C and D) and $\K_n^{\lambda_0}$ (for type A), constructed recursively from the currents $\J^{\lambda_0}_n$ (and which depended on a free parameter $\xi$ for types B, C and D). In this subsection, we will show how to construct these currents using generating functions.

We will mainly focus on the case where $\g$ is of type A and will briefly comment on types B, C and D at the end of the subsection. Let us then suppose that $\g=\sl(d,\C)$, so that we can use the notations and results of subsection \ref{Sec:TypeANonCyc}. We introduce
\begin{equation}\label{Eq:DefF}
F(\lambda,\mu,x)= \Tr \log\bigl( \Id- \mu \varphi(\lambda)\Lc(\lambda,x) \bigr)
\end{equation}
and
\begin{equation}\label{Eq:DefA}
A(\lambda,\mu,x) = \det\bigl( \Id- \mu \varphi(\lambda)\Lc(\lambda,x) \bigr),
\end{equation}
so that $A(\lambda,\mu,x) = \exp\bigl(F(\lambda,\mu,x)\bigr)$. By expanding the matricial logarithm in \eqref{Eq:DefF} as a power series in $\mu$ one finds
\begin{equation} \label{Eq:PowF}
F(\lambda,\mu,x) = - \sum_{k=2}^{\infty} \frac{\mu^k}{k} \Tc_k(\lambda,x),
\end{equation}
with $\Tc_n(\lambda,x)$ defined in equation \eqref{Eq:DefT}. We are interested in the evaluations of $F(\lambda,\mu,x)$ and $A(\lambda,\mu,x)$ at $\lambda=\lambda_0$, which are well defined as $\lambda_0$ is a regular zero. Following~\cite{Evans:1999mj}, we look for $\K_n^{\lambda_0}(x)$ in the form of
\begin{equation}\label{Eq:KGen1}
\K^{\lambda_0}_n(x) = A(\lambda_0,\mu,x)^{p_n} \Bigr|_{\mu^{n}}
\end{equation}
for some rational number $p_n$, where $f(\mu)|_{\mu^n}$ denotes the coefficient of $\mu^n$ in the power series expansion of $f(\mu)$.

The Poisson brackets of the currents $\Tc_n(\lambda_0,x)=\J_n^{\lambda_0}(x)$ are given by equation \eqref{Eq:PBJTypeA}. This allows one to compute $\left\lbrace F(\lambda_0,\mu,x), F(\lambda_0,\nu,y) \right\rbrace$ and $\left\lbrace A(\lambda_0,\mu,x), A(\lambda_0,\nu,y) \right\rbrace$. As equation \eqref{Eq:PBJTypeA} coincides with the equation (4.5) of~\cite{Evans:1999mj} up to a global factor, these Poisson brackets are the same as in~\cite{Evans:1999mj} (equations (4.13) and (4.14)), still up to the global factor. Thus, the procedure of~\cite{Evans:1999mj} applies and we conclude that the Poisson bracket of the local charges \eqref{Eq:DefQK} defined in terms of the currents \eqref{Eq:KGen1} is
\begin{equation*}
\left\lbrace \Q^{\lambda_0}_n, \Q^{\lambda_0}_m \right\rbrace
=  p_n p_m \mu\nu \frac{\varphi'(\lambda_0)}{T} \int \dd x \; A(\lambda_0,\mu,x)^{p_n} \p_x \bigl( A(\lambda_0,\nu,x)^{p_m} \bigr) h_{mn}(\mu,\nu) \Big|_{\mu^n \nu^m},
\end{equation*}
where
\begin{equation} \label{Eq:hnm def}
h_{nm}(\mu,\nu) = \left[ \left( \frac{n-1}{p_n}\nu - \frac{m-1}{p_m}\mu  \right) \frac{1}{\mu-\nu} + \frac{1}{d}\frac{(n-1)(m-1)}{p_n p_m} \right].
\end{equation}
It follows that the charges $\Q_n^{\lambda_0}$ are in involution if we choose, for any $k\in\Z_{\geq 2}$, $p_k=\frac{k-1}{d}$. The corresponding currents are given by
\begin{equation}\label{Eq:KGen2}
\K_n^{\lambda_0}(x) = \left. \exp \left( - \frac{n-1}{d} \sum_{k=2}^{\infty} \frac{\mu^k}{k} \J_k^{\lambda_0}(x) \right) \right|_{\mu^n}.
\end{equation}
One can check that the first currents defined by this generating function are given by equation \eqref{Eq:KJ}, up to overall global factors. The current $\K_n^{\lambda_0}(x)$ is the evaluation at $\lambda=\lambda_0$ of the more general current
\begin{equation}\label{Eq:DefW}
\W_n(\lambda,x) = A(\lambda,\mu,x)^{(n-1)/d} \Bigr|_{\mu^{n}},
\end{equation}
which we will need later. The equation
\begin{equation}\label{Eq:SerieW}
\W_n(\lambda,x) = \left. \exp \left( - \frac{n-1}{d} \sum_{k=2}^{\infty} \frac{\mu^k}{k} \Tc_k(\lambda,x) \right) \right|_{\mu^n}
\end{equation}
allows one to compute $\W_n(\lambda,x)$ as a polynomial in the $\Tc_k(\lambda,x)$. More precisely, $\W_n$ is related to the $\Tc_k$'s in the same way that $\K^{\lambda_0}_n$ is related to the $\J^{\lambda_0}_k$'s.\\

We end this subsection by saying a few words on Lie algebras $\g$ of type B, C or D. In this case, we saw in subsection \ref{Sec:NonCycZeroBCD} that the local charges in involution can be taken as integrals of currents $\K_{2n}^{\lambda_0}(\xi)$, depending on a free parameter $\xi$ (see equation \eqref{Eq:KAlpha}). These currents can be obtained from the $\J^{\lambda_0}_{2k}$'s using a generating function, similar to the one presented above for type A. We will not enter into details here and will just present the final result, based on reference~\cite{Evans:1999mj}. The current $\K^{\lambda_0}_{2n}(\xi)$ can be computed as:
\begin{equation}
\K_{2n}^{\lambda_0}(\xi,x) = \left. \exp \left( - \frac{\xi(n-1)}{2} \sum_{k=1}^{\infty} \frac{\mu^k}{k} \J_{2k}^{\lambda_0}(x) \right) \right|_{\mu^n}.
\end{equation}
Starting from the Poisson bracket \eqref{Eq:PBJTypeBCD}, one can show that the corresponding charges $\Q^{\lambda_0}_{2n}(\xi)$ are in involution, using similar techniques as above for type A. We refer the interested reader to reference~\cite{Evans:1999mj} for details on the proof. An explicit computation shows that the first currents $\K^{\lambda_0}_{2n}(\xi)$ obtained from the above equation are given by equation \eqref{Eq:KAlpha}, up to overall global factors.

\subsection{Summary}
\label{Sec:SummaryNonCyc}

To conclude this section, let us summarise the results that we obtained. In particular, we will use this as an opportunity to extend the notations $\K_n^{\lambda_0}$ and $\W_n$, defined for a type A algebra in the previous subsections, to other types. This will serve to uniformise the notation in the rest of the article.

When $\g$ is of type A, the currents $\K_n^{\lambda_0}(x)$ are given in subsection \ref{Sec:GenNonCyc} through equation \eqref{Eq:KGen2}. We also defined a current $\W_n(\lambda,x)$ depending on the spectral parameter $\lambda$ in equation \eqref{Eq:SerieW}. For a Lie algebra $\g$ of type B, C or D (as treated in subsection \ref{Sec:NonCycZeroBCD}), we introduced currents $\K^{\lambda_0}_n(\xi)$, depending on a free parameter $\xi$. However, as explained at the end of susbection \ref{Sec:NonCycZeroBCD}, we will only use the currents $\J_n^{\lambda_0}(x)=\K_n^{\lambda_0}(\xi=0,x)$ in the rest of this paper. In order to employ uniform notations throughout the paper, we shall define in this case $\K_n^{\lambda_0}(x)=\J_n^{\lambda_0}(x)$ and $\W_n(\lambda,x)=\Tc_n(\lambda,x)$.

With these conventions, independently of the type of $\g$, the current $\K^{\lambda_0}_n(x)$ is the evaluation of $\W_n(\lambda,x)$ at $\lambda=\lambda_0$ and the charge $\Q^{\lambda_0}_n$ is given by
\begin{equation}
\Q_n^{\lambda_0} = \int \dd x \; \K^{\lambda_0}_n(x).
\end{equation}

Recall also that we restrict the degrees $n$ of the currents $\K_n^{\lambda_0}$ to some subset $\E_{\lambda_0}$ of $\Z_{\geq 2}$. In fact, independently of the type of $\g$, $\E_{\lambda_0}$ can (almost) be seen as the set of exponents of the affine algebra $\widehat{\g}$ plus one. This was already observed for type A in subsection \ref{Sec:TypeANonCyc}, based on the results of~\cite{Evans:1999mj}. For types B, C and D, we saw in subsection \ref{Sec:NonCycZeroBCD} that $\E_{\lambda_0}$ is the set of all even numbers, which turns out to coincide with the exponents of $\widehat{\g}$ plus one for types B and C~\cite{Evans:1999mj}. For type D, there are some exponents missing in this construction (the rank modulo the Coxeter number), which are related to the Pfaffian (see~\cite{Evans:1999mj}). Although we do not consider the Pfaffian here, we expect that it should be possible to construct a corresponding local charge in the present framework too.

Having introduced these type-independent notations, we can summarise the results of this section by the following theorem.

\begin{theorem} \label{thm: involution of Qs}
Let $\lambda_0$ be a non-cyclotomic regular zero of the model. Then, for any $m$ and $n$ in $\E_{\lambda_0}$, the charges $\Q^{\lambda_0}_n$ and $\Q^{\lambda_0}_m$ are in involution, \textit{i.e.} we have
\begin{equation*}
\lbrace \Q_n^{\lambda_0}, \Q_m^{\lambda_0} \rbrace = 0.
\end{equation*}
\end{theorem}

The notations and results summarised above will be generalised to the case of 
cyclotomic zeros in the following section. Let us note here that there will be some 
subtlety in the definition of the current $\W_n$ for an algebra of type A in the case 
when the automorphism $\s$ is inner, compared to the definition 
given above. We shall discuss this in subsection \ref{Sec:SummaryCyc}.

\section{Charges at cyclotomic zeros}
\label{Sec:CycZero}

In this section, we explain how to construct towers of local charges in involution attached to cyclotomic regular zeros of the twist function $\varphi$. Recall that a cyclotomic point is a point fixed by the action of the cyclic group $\Z_T$, \textit{i.e.} the origin or infinity. Suppose we are considering a model with a regular zero at infinity. As explained in subsection \ref{Sec:Infinity}, working in the new spectral parameter $\alpha=\lambda^{-1}$ and with the new Lax matrix $\Lc^\infty$ amounts 
to treating, instead, a model with  a regular zero at $\alpha=0$ and automorphism $\s^{-1}$. Hence it is sufficient to describe the extraction of local charges at the origin.

Throughout this section we therefore consider a  model with $T>1$ and 
a regular zero at the origin. We thus have $\varphi(0)=0$ and $\varphi(\lambda)\Lc(\lambda,x)$ regular at 0. Using the equivariance property \eqref{Eq:TwistEqui}, we see that the smallest power of $\lambda$ in $\varphi$ is of the form $\alpha T-1$, for some $\alpha\in\Z_{\geq 1}$. In terms of the function $\zeta$, defined in equation \eqref{Eq:DefZeta}, this implies that $\zeta(\lambda^T)=O(\lambda^{\alpha T})$. We will mostly need the fact that $\zeta(\lambda^T)=O(\lambda^T)$, \textit{i.e.} that $\varphi(\lambda)=O(\lambda^{T-1})$, and more precisely the asymptotic property
\begin{equation}\label{Eq:ZetaAsymptotic}
\zeta(\lambda^T)=\zeta'(0)\lambda^T + O(\lambda^{2T}).
\end{equation}

Recall that in the previous section we extracted local charges by evaluating the traces of powers of $\varphi(\lambda)\Lc(\lambda,x)$ at the regular zeros. In the case of a cyclotomic point, this method is not sufficient to extract all charges, as such traces can vanish. To understand how to construct the whole algebra of local charges, we will first need to establish equivariance properties of $S_n(\lambda,x)$.

\subsection{Equivariance properties}
\label{Sec:EquivT}

Recall the equivariance properties \eqref{Eq:TwistEqui} and \eqref{Eq:EquiL} of $\varphi$ and $\Lc$. In this subsection, we look for a similar relation for $S_n(\lambda,x)$. In general, $S_n(\lambda,x)$ does not belong to the Lie algebra $\g$ since it is defined as the power of an element of $\g$ seen in the fundamental representation. Thus, one cannot consider directly the action of $\s$ on $S_n(\lambda,x)$.

We refer here to the discussion of appendix \ref{App:ExtSigma}. We will restrict to the case where $\s$ is not one of the special automorphisms of $D_4 = \so(8,\C)$. In this case, we can extend $\s$ to a linear endomorphism on the space $F$ of all matrices acting on the fundamental representation, that we shall still denote $\s$ (see details in appendix \ref{App:ExtSigma}). Note that this new endomorphism $\s$ of $F$ is still of order $T$. We will also need the following properties of $\s$. For any $Z\in F$ we have
\begin{subequations}
\begin{align}
\s(Z^n) &= \epsilon^{n-1} \s(Z)^n, \label{Eq:SigmaPow}\\
\Tr\bigl(\s(Z)\bigr) &= \epsilon \Tr(Z), \label{Eq:SigmaTr}
\end{align}
\end{subequations}
for some $\epsilon$ in $\lbrace 1,-1 \rbrace$. Note that $\epsilon$ is always 1 except when $\g=\sl(d,\C)$ and $\s$ has a non-trivial outer part, in which case $T$ is even. We shall write $\epsilon=\omega^{\frac{\eta T}{2}}$, with $\eta$ in $\lbrace 0,1 \rbrace$.

From equations \eqref{Eq:TwistEqui} and \eqref{Eq:EquiL} and the identity \eqref{Eq:SigmaPow}, we deduce that $S_n$ satisfies the equivariance property
\begin{equation}\label{Eq:EquiS}
\s\bigl(S_n(\lambda,x)\bigr) = \omega^{\kappa(n-1)+1} S_n(\omega\lambda,x),
\end{equation}
with $\kappa=1+\frac{\eta T}{2}$. Let us consider the power series expansion
\begin{equation}\label{Eq:PowS}
S_n(\lambda,x) = \sum_{r=0}^{\infty} A_{n,r}(x) \lambda^r.
\end{equation}
We then find
\begin{equation}\label{Eq:EquiA}
\s(A_{n,r})=\omega^{r+\kappa(n-1)+1} A_{n,r}.
\end{equation}
Taking the trace and using equation \eqref{Eq:SigmaTr}, we find
\begin{equation}\label{Eq:EquiTrA}
\Tr(A_{n,r})=\omega^{r+n\kappa} \Tr(A_{n,r}).
\end{equation}
Thus, $\Tr(A_{n,r})$ vanishes except if $r\equiv r_n \, [T]$, where $r_n$ is the remainder of the euclidian division of $-n\kappa$ by $T$. We define
\begin{equation*}
\J^0_n(x) = \Tr \bigl( A_{n,r_n}(x) \bigr).
\end{equation*}
In particular, the first term in the power series expansion of $\Tc_n(\lambda,x)$ is $\lambda^{r_n} \J^0_n(x)$. Note that $\J^0_n(x)$ is the evaluation of $\Tc_n(\lambda,x)$ at $\lambda=0$ if and only if $r_n=0$, \textit{i.e.} if $T$ divides $n\kappa$. Note also, as $-2\kappa \equiv -2 \, [T]$, that $r_2=T-2$. Thus, we find
\begin{equation*}
\J^0_2(x) = \zeta'(0) \res_{\lambda=0} \varphi(\lambda) \Tr \bigl( \Lc(\lambda,x)^2 \bigr),
\end{equation*}
where $\zeta$ was defined in equation \eqref{Eq:DefZeta}. Finally, let us remark that equation \eqref{Eq:EquiTrA} implies that $\Tc_n(\lambda,x)$ has the following equivariance property
\begin{equation}\label{Eq:EquiT}
\Tc_n(\omega\lambda,x) = \omega^{r_n} \Tc_n(\lambda,x).
\end{equation}

\subsection{Poisson algebra of the currents}
\label{Sec:PBCurrentsCyc}

One can extract the Poisson brackets of the currents $\J^0_n(x)$ and $\J^0_m(y)$ as the coefficient of $\lambda^{r_n+r_m}$ in the power series expansion of $\left\lbrace \Tc_n(\lambda,x), \Tc_m(\lambda,y) \right\rbrace$. The latter can be computed from equation \eqref{Eq:PBT}. Specifically, using the identity \eqref{Eq:RCas} we find
\begin{equation*}
\lambda\mu U\ti{12}(\lambda,\mu)
= - \frac{\zeta(\lambda^T)-\zeta(\mu^T)}{\lambda^T-\mu^T} \sum_{k=0}^{T-1} \lambda^k \mu^{T-k} C^{(k)}\ti{12} + \zeta(\mu^T) C^{(0)}\ti{12},
\end{equation*}
with $\zeta$ defined in equation \eqref{Eq:DefZeta}. Taking the limit $\mu \to \lambda$ we obtain
\begin{equation}\label{Eq:UAround0}
U\ti{12} (\lambda,\lambda) = - \lambda^{T-2} \zeta'(\lambda^T) C\ti{12} + \lambda^{-2} \zeta(\lambda^T) C\ti{12}^{(0)},
\end{equation}
so that
\begin{align}\label{Eq:PBTCyc}
\left\lbrace \Tc_n(\lambda,x), \Tc_m(\lambda,y) \right\rbrace
&  = nm \lambda^{T-2} \zeta'(\lambda^T) \Tr\ti{12}\Bigl( C\ti{12} S_{n-1}(\lambda,x)\ti{1} S_{m-1}(\lambda,y)\ti{2} \Bigr) \delta'_{xy} \\
& \hspace{40pt} -nm \lambda^{-2} \zeta(\lambda^T) \Tr\ti{12}\Bigl( C^{(0)}\ti{12} S_{n-1}(\lambda,x)\ti{1} S_{m-1}(\lambda,y)\ti{2} \Bigr) \delta'_{xy}. \notag
\end{align}

The first term of this Poisson bracket has the same structure as the 
Poisson bracket \eqref{Eq:PBJ}. The main difference coming from cyclotomy is thus the second term, which involves the partial Casimir $C\ti{12}^{(0)}$. We recall that we have the partial completeness relation
\begin{equation}\label{Eq:CompRelPart}
\Tr\ti{2}(C^{(0)}\ti{12} Z\ti{2}) = \pi^{(0)} (Z),
\end{equation}
for any $Z\in\g$.\\

The second term in \eqref{Eq:PBTCyc} will therefore involve the projection $S^{(0)}_{n-1}(\lambda,x)$ of $S_{n-1}(\lambda,x)$ onto the grading zero $F^{(0)} = \left\lbrace Z \in F \, | \, \s(Z)=Z \right\rbrace$. To determine these projections, we can make use of the power series expansion \eqref{Eq:PowS} and equation \eqref{Eq:EquiA}. In particular, one finds that $A_{n-1,r+T}$ is in $F^{(0)}$ if and only if $A_{n-1,r}$ also belongs to $F^{(0)}$. Let us then define $q_n$ to be the unique integer between $0$ and $T-1$ such that $A_{n-1,r}$ belongs to $F^{(0)}$ if and only if $r \equiv q_n \, [T]$. Using equation \eqref{Eq:EquiA} we find $q_n \equiv r_n+1 \, [T]$. So $q_n = r_n +1$ if $r_n \leq T-2$ and $q_n=0$ if $r_n=T-1$.\\

\noi To simplify the Poisson bracket \eqref{Eq:PBTCyc}, we will need to distinguish between three cases:
\begin{itemize}
\item $\g$ is of type B, C or D,
\item $\g$ is of type A and $\s$ is inner,
\item $\g$ is of type A and $\s$ is not inner.
\end{itemize}

\subsection{Algebra of type B, C or D}
\label{Sec:CycBCD}

We first consider $\g$ to be of type B, C or D.  Recall that in this case $S_{2n-1}(\lambda,x)$ belongs to the Lie algebra so that $\Tc_{2n-1}(\lambda,x)$ is zero and hence we consider only the currents $\J^0_{2n}(x)$. Moreover, we can use the completeness relations \eqref{Eq:CompRel} and \eqref{Eq:CompRelPart} in \eqref{Eq:PBTCyc}. We then find
\begin{align} \label{Eq:PBTCycBCD}
\left\lbrace \Tc_{2n}(\lambda,x), \Tc_{2m}(\lambda,y) \right\rbrace
& = 4nm \lambda^{T-2} \zeta'(\lambda^T) \Tr\Bigl( S_{2n-1}(\lambda,x) S_{2m-1}(\lambda,y) \Bigr) \delta'_{xy}  \\
& \hspace{40pt} - 4nm \lambda^{-2} \zeta(\lambda^T) \Tr\Bigl( S^{(0)}_{2n-1}(\lambda,x) S^{(0)}_{2m-1}(\lambda,y) \Bigr) \delta'_{xy}. \notag
\end{align}
After integration over $y$, the first term becomes a total derivative with respect to $x$ by virtue of equation \eqref{Eq:DerS} and thus vanishes when integrated over $x$.\\

Recall moreover that the Poisson bracket of $\J^0_{2n}(x)$ with $\J^0_{2m}(y)$ is obtained from \eqref{Eq:PBTCycBCD} by keeping only the term $\lambda^{r_{2n}+r_{2m}}$ in the power series expansion. We note that the smallest power of $\lambda$ in the second term of \eqref{Eq:PBTCycBCD} is $\alpha T-2+q_{2n}+q_{2m}$ (cf. equation \eqref{Eq:ZetaAsymptotic} and above). As we saw in the previous subsection, $q_k=r_k+1$ if $r_k \leq T-2$ and $q_k = 0$ if $r_k=T-1$. In the case where $r_{2n}$ and $r_{2m}$ are different from $T-1$, the smallest power of $\lambda$ is then $r_{2n}+r_{2m}+\alpha T$ so the second term of \eqref{Eq:PBTCycBCD} does not contribute to the Poisson bracket of $\J^0_{2n}(x)$ with $\J^0_{2m}(y)$, as $\alpha \geq 1$.

If $r_{2n}$ or $r_{2m}$ is equal to $T-1$ then there will be a contribution from this term involving other objects than only the $\J
^0_k$'s, preventing us from constructing charges in involution. Thus, we will only consider the currents $\J^0_{2k}(x)$ such that $r_{2k} \neq T-1$. We then have
\begin{equation*}
\left\lbrace \Q^0_{2n}, \Q^0_{2m} \right\rbrace = 0,
\end{equation*}
where $\Q^0_{2k}$ is the integral of the current $\J^0_{2k}(x)$.\\

We have thus extracted a tower of local charges in involution from the Lax matrix around the origin. Just as in the non-cyclotomic case, these charges are integrals of some polynomials of even degrees in the fields appearing in the Lax matrix. The main difference with the non-cyclotomic case is the fact that, in general, we do not have a current of any even degree. More precisely, we `dropped' the currents of degree $2n$, for all $n$ such that $r_{2n}= T-1$. Recall from appendix \ref{App:ExtSigma} that in the case of an algebra of type B, C or D, we have $\epsilon=1$ and $\kappa=1$. Thus $r_{2n}$ is the remainder of the euclidian division of $-2n$ by $T$, which means that $r_{2n}=T-1$ if and only if $2n \equiv 1 \, [T]$. In particular, we see that there is no drop of any degrees if $T$ is even.

\subsection{Algebra of type A and $\s$ inner}
\label{Sec:TypeATrivial}

Let us now suppose that $\g$ is $\sl(d,\C)$ and $\s$ is inner. In this case, we have the generalised completeness relation 
\eqref{Eq:CompRelSl}. Moreover, we also have a similar identity for the partial Casimir $C\ti{12}^{(0)}$, derived as follows. Recall that for any $Z\in F$, $Z-\frac{1}{d}\Tr(Z)\Id$ 
belongs to $\g$. Moreover, we note that the identity $\Id$ is in the 
grading zero $F^{(0)}$ for $\s$ inner (cf. appendix \ref{App:ExtSigma}). Using equation \eqref{Eq:CompRelPart}, we then have
\begin{equation}\label{Eq:CompRelPartSl}
\Tr\ti{2} \left( C^{(0)}\ti{12} Z\ti{2} \right) = \pi^{(0)}(Z) - \frac{1}{d} \Tr(Z) \Id.
\end{equation}
Using equations \eqref{Eq:CompRelSl} and \eqref{Eq:CompRelPartSl} in the Poisson bracket \eqref{Eq:PBTCyc}, we obtain
\begin{align} \label{Eq:PBTCycA1}
& \left\lbrace \Tc_n(\lambda,x), \Tc_m(\lambda,y) \right\rbrace
 = nm \lambda^{T-2} \zeta'(\lambda^T) \Tr\Bigl( S_{n-1}(\lambda,x) S_{m-1}(\lambda,y) \Bigr) \delta'_{xy} \\
& \hspace{6pt} - nm \lambda^{-2} \zeta(\lambda^T) \Tr\Bigl( S^{(0)}_{n-1}(\lambda,x) S^{(0)}_{m-1}(\lambda,y) \Bigr) \delta'_{xy}
 + \frac{nm}{d} \frac{ \zeta(\lambda^T) - \lambda^{T} \zeta'(\lambda^T)}{\lambda^2} \Tc_{n-1}(\lambda,x) \Tc_{m-1}(\lambda,y) \delta'_{xy} \notag
\end{align}

The Poisson bracket of $\J^0_n(x)$ with $\J^0_m(y)$ is obtained by extracting the coefficient of $\lambda^{r_n+r_m}$ in the above equation.
To treat the second term on the right hand side of this equation, we follow the discussion of the previous subsection \ref{Sec:CycBCD}. The smallest power of $\lambda$ appearing in this term is $\alpha T-2+q_n+q_m$ and if we restrict to $n$ and $m$ such that $r_n$ and $r_m$ are different from $T-1$, this power is strictly greater than $r_n+r_m$. The term then does not contribute to the Poisson bracket $\left\lbrace \J^0_n(x), \J^0_m(y) \right\rbrace$. Let us turn to the third term on the right hand side of equation \eqref{Eq:PBTCycA1}. It can be seen from equation \eqref{Eq:ZetaAsymptotic} that $\zeta(\lambda^T)-\lambda^T \zeta'(\lambda^T) = O(\lambda^{2T})$. The smallest power of $\lambda$ that can appear in this term is thus $2T-2+r_{n-1}+r_{m-1}$, which is always greater than $2T-2$ and therefore strictly greater than $r_n+r_m$ if $r_n$ and $r_m$ are different from $T-1$.\\

In conclusion, only the first term of the right hand side of \eqref{Eq:PBTCycA1} contributes to the Poisson bracket $\left\lbrace \J^0_n(x), \J^0_m(y) \right\rbrace$, which then has the same structure as in the previous subsection. Integrating this bracket over $x$ and $y$, we recognise the integral of a total derivative proportional to $\p_x\Tc_{n+m-2}(\lambda,x)$, which then vanishes, assuming appropriate boundary conditions. Thus, for any $n$ and $m$ such that $r_n$ and $r_m$ are different from $T-1$, we have
\begin{equation*}
\left\lbrace \Q^0_n, \Q^0_m \right\rbrace = 0
\end{equation*}
with $\Q^0_k$ the integral of the current $\J^0_k(x)$. As in the subsection 
\ref{Sec:CycBCD}, we have $\epsilon=1$ and $\kappa=1$ for $\s$ inner. It follows that the integers $n$ such that $r_n=T-1$ (for which we do not consider the charge $\Q^0_n$) are the ones equal to 1 modulo $T$.

\subsection{Algebra of type A and $\s$ not inner}
\label{Sec:TypeANonTrivial}

Finally, let us treat the case where $\g=\sl(d,\C)$ and $\s$ not inner. 
In particular, this implies that $T$ is even and we shall write $T=2S$ in this subsection. We still 
have the generalised completeness relation \eqref{Eq:CompRelSl}. As $\s$ is not inner, 
we have $\s(\Id)=-\Id$ and hence $\pi^{(0)}(\Id)=0$. We deduce that in this case, the partial completeness relation \eqref{Eq:CompRelPart} actually holds for any $Z\in F$. Equation \eqref{Eq:PBTCyc} then gives
\begin{align} \label{Eq:PBTCycA2}
& \hspace{-5pt}\left\lbrace \Tc_n(\lambda,x), \Tc_m(\lambda,y) \right\rbrace  = nm \lambda^{T-2} \zeta'(\lambda^T) \Tr\Bigl( S_{n-1}(\lambda,x) S_{m-1}(\lambda,y) \Bigr) \delta'_{xy} \\
& \hspace{28pt} - nm \lambda^{-2} \zeta(\lambda^T) \Tr\Bigl( S^{(0)}_{n-1}(\lambda,x) S^{(0)}_{m-1}(\lambda,y) \Bigr) \delta'_{xy}  - \frac{nm}{d} \lambda^{T-2} \zeta'(\lambda^T) \Tc_{n-1}(\lambda,x) \Tc_{m-1}(\lambda,y) \delta'_{xy}. \notag
\end{align}
We follow the method of the previous subsections and look for the power $r_n+r_m$ of $\lambda$ in the right hand side of this bracket. As explained in subsection \ref{Sec:CycBCD}, the second term does not contribute when we restrict to $r_n$ and $r_m$ different from $T-1$.\\

The first term is treated as in the case of a non-cyclotomic point: using the identity $f(y)\delta'_{xy}=f(x)\delta'_{xy}+\p_x\bigl(f(x)\bigr) \delta_{xy}$ and the equation \eqref{Eq:DerS}, we find
\begin{equation}\label{Eq:SxS}
\Tr\Bigl( S_{n-1}(\lambda,x) S_{m-1}(\lambda,y) \Bigr) \delta'_{xy} = \Tc_{n+m-2}(\lambda,x) \delta'_{xy} + \frac{m-1}{n+m-2} \p_x \bigl( \Tc_{n+m-2}(\lambda,x) \bigr) \delta_{xy}.
\end{equation}
The powers of $\lambda$ appearing in the power series expansion of the first term are then of the form $r_{n+m-2} - 2 + aT$, with $a \in \Z_{\geq 1}$. One has $r_{n+m-2} \equiv r_n +r_m + 2 \, [T]$ and $0 \leq r_{n+m-2} \leq T-1$. Moreover, $r_n+r_m+2$ is always between $0$ and $2T-2$ if we suppose $r_n$ and $r_n$ different from $T-1$. If $ 0 \leq r_n+r_m+2 < T $, we have $r_{n+m-2}=r_n+r_m+2$ and the powers $r_{n+m+2}-2+ aT$ are then all strictly greater than $r_n+r_m$. If $T \leq r_n+r_m+2 \leq 2T-2$ then $r_{n+m-2}=r_n+r_m+2-T$ and the power $r_{n+m-2}-2+aT$ is equal to $r_n+r_m$ if and only if $a=1$.

Finally, let us consider the third term on the right hand side of \eqref{Eq:PBTCycA2}. The powers of $\lambda$ in its power series expansion are of the form $r_{n-1}+r_{m-1}-2+ aT$, with $a \in \Z_{\geq 1}$. Note that $r_{k-1} \equiv r_k + 1 + S \equiv r_k + 1 - S \, [T]$. Thus, $r_{k-1} = r_k + 1 + S$ if $0 \leq r_k < S-1$ and $r_{k-1} = r_k +1 - S$ if $S-1 \leq r_k \leq T -1$. We then conclude that the power $r_{n-1}+r_{m-1}-2+aT$ is equal to $r_n+r_m$ if and only if $r_n + 1 -S \geq 0$, $r_m + 1 - S \geq 0$ and $a=1$.

Combining all the above results, we find a closed expression for the Poisson bracket of the currents $\J^0_n(x)$ and $\J^0_m(y)$ when $r_n$ and $r_m$ are different form $T-1$, specifically
\begin{align}\label{Eq:PBJCycA}
& \left\lbrace \J^0_n(x), \J^0_m(y) \right\rbrace
= \theta_{r_n+r_m+2-T} \, nm\zeta'(0) \left( \J^0_{n+m-2}(x) \delta'_{xy} + \frac{m-1}{n+m-2} \p_x \left( \J_{n+m-2}^{0}(x) \right) \delta_{xy}  \right) \notag \\
 & \hspace{20pt} - \theta_{r_n+1-S}\, \theta_{r_m+1-S} \, \frac{nm}{d}\zeta'(0)  \Bigl(  \J_{n-1}^{0}(x)\J_{m-1}^{0}(x)\delta'_{xy} + \J_{n-1}^{0}(x) \p_x \left( \J_{m-1}^{0}(x) \right) \delta_{xy} \Bigr),
\end{align}
where $\theta_k=1$ if $k\in\Z_{\geq 0}$ and $\theta_k=0$ if $k\in\Z_{<0}$.\\

As in the case of a non-cyclotomic zero, one can construct charges in involution by taking integrals of new currents $\K^0_n(x)$, constructed as polynomials of the currents $\J^0_m(x)$ for $m \leq n$. The method in the present case is similar: one can construct the currents $\K^0_n$ recursively by asking that the corresponding charge $\Q^0_n$ Poisson commutes with $\Q^0_m$ for all $m<n$. One of the main difference with the non-cyclotomic case is the fact that we do not consider a current $\K^0_n$ when $r_n=T-1$ (we say that such a degree $n$ drops from the construction). The second difference is the presence of the terms $\theta_k$ in the Poisson bracket \eqref{Eq:PBJCycA} compared to \eqref{Eq:PBJTypeA}. Since these terms depend on the numbers $r_k$, the construction of the $\K^0_n$'s will depend on $T$. As an illustration, we give here the expression of the first $\K^0_n$'s for $T$ equal to 2 and 4.

In the case $T=2$, we have $\kappa=2$ hence $r_n=0$ for all $n\geq 2$. Thus, there is no drop of any current due to the condition $r_n\neq T-1$ and the current $\J^0_n(x)$ is simply the evaluation of $\Tc_n(\lambda,x)$ at $\lambda=0$. Moreover, since $2-T$ and $1-S$ are both zero for $T=2$, we note that all $\theta_k$ terms in the Poisson bracket \eqref{Eq:PBJCycA} are equal to 1. The construction for the $\K^0_n$'s is therefore the same as in the non-cyclotomic case and their expression is given by \eqref{Eq:KJ}.

Let us now consider the case $T=4$. We find $\kappa=3$ and $r_n \equiv n \, [4]$ and therefore drop the currents $\J^0_{4k+3}(x)$. Constructing the $\K^0_n$'s recursively we find
\begin{align}\label{Eq:KJCyc4}
&\K^0_2 = \J^0_2, \;\;\; \K^0_4=\J^0_4, \;\;\; \K^0_5=\J^0_5, \\
&\K^0_6 = \J^0_6 - \frac{15}{4d}\J^0_2\J^0_4, \;\;\; \K^0_8 = \J^0_8 - \frac{7}{4d}\bigl(\J_4^0\bigr)^2 \notag
\end{align}
where we dropped the currents $\K^0_3$ and $\K^0_7$ and more generally all currents $\K^0_{4k+3}$.\\

Comparing these currents to the ones constructed in the non-cyclotomic case \eqref{Eq:KJ}, one can observe a pattern in the cyclotomic procedure. Here also, the current $\K^0_n$ is constructed by correcting $\J^0_n$ with monomials $\J^0_{m_1} \ldots \J^0_{m_p}$ such that $m_1+\ldots+m_p=n$. We observe that not all the monomials in the non-cyclotomic corrections appear among the cyclotomic ones but the ones that do have the same coefficients (for example $15/4d$ for the $\J_2\J_4$ correction of $\J_6$). Moreover, we note that a monomial $\J^0_{m_1} \ldots \J^0_{m_p}$ appearing in the non-cyclotomic procedure is also present in the cyclotomic expression if and only if $r_{m_1} + \ldots + r_{m_p} = r_n$.

The above observations are still found to hold for larger values of $T$ and $n$ (although we do not include the corresponding expression of $\K^0_n$ for conciseness). This allows one to find the currents $\K^0_n(x)$ without going through the recursive procedure if one already knows the result for the non-cyclotomic case. A more systematic approach to constructing higher conserved charges in involution in the cyclotomic case would be to find a generating function for the $\K^0_n$'s, generalising the results 
of subsection \ref{Sec:GenNonCyc}. This will be the subject of the next subsection.

\subsection{Generating function for type A with $\s$ not inner}
\label{Sec:CycGenerating}

In subsection \ref{Sec:GenNonCyc} we presented, following~\cite{Evans:1999mj}, the generating function for constructing the currents $\K_n(x)$ in the non-cyclotomic setting. In particular, we found that the relation between the $\K_n(x)$'s and the $\J_m(x)$'s is given by equation \eqref{Eq:KGen2}. In the previous subsection we showed how the currents $\K^0_n(x)$ could be constructed in the cyclotomic case from the knowledge of the corresponding result in the non-cyclotomic case. In particular, starting from the expression of $\K_n$ as a polynomial of the $\J_m$'s, we observed that $\K^0_n$ can be constructed in the same way by keeping monomials $\J^0_{m_1} \ldots \J^0_{m_p}$ with the same coefficient if and only if $r_n=r_{m_1}+\ldots+r_{m_p}$. This procedure for going from the non-cyclotomic to the cyclotomic setting has a natural interpretation in terms of equation \eqref{Eq:KGen2}. Indeed, the current $\K^0_n$ constructed above is equal to
\begin{equation}\label{Eq:KGenCyc1}
\hspace{-2pt}\K_n^{0}(x) = \left. \exp \left( - \frac{n-1}{d} \sum_{k=2}^{\infty} \frac{\mu^k}{k} \lambda^{r_k}\J_k^{0}(x) \right) \right|_{\mu^n \, \lambda^{r_n}},
\end{equation}
where the projection onto the term $\lambda^{r_n}$ ensures that we keep only the monomials satisfying the condition $r_n=r_{m_1}+\ldots+r_{m_p}$.\\

Recall that $\lambda^{r_k} \J_k^0(x)$ is the first term in the power series expansion of $\Tc_k(\lambda,x)$. Moreover, the next terms start with a $(r_k+T)^{\rm th}$ power of $\lambda$. Since $r_n < T$, such terms can be added to the exponent in equation \eqref{Eq:KGenCyc1} without changing the left hand side as they cannot contribute to a $\lambda^{r_n}$-term. We may therefore also write
\begin{equation}\label{Eq:KGenCyc2}
\hspace{-2pt}\K_n^{0}(x) = \left. \exp \left( - \frac{n-1}{d} \sum_{k=2}^{\infty} \frac{\mu^k}{k} \Tc_k(\lambda,x) \right) \right|_{\mu^n \, \lambda^{r_n}}.
\end{equation}
In terms of the definitions \eqref{Eq:DefF} and \eqref{Eq:DefA} of $F(\lambda,\mu,x)$ and $A(\lambda,\mu,x)$ and equation \eqref{Eq:PowF}, we can further re-express equation \eqref{Eq:KGenCyc2} as
\begin{equation}\label{Eq:KGenCyc3}
\K_n^{0}(x) = \left. A(\lambda,\mu,x)^{(n-1)/d} \right|_{\mu^n \, \lambda^{r_n}},
\end{equation}
or again
\begin{equation*}
\K_n^{0}(x) = \left. \exp \left( \frac{n-1}{d} F(\lambda,\mu,x) \right) \right|_{\mu^n \, \lambda^{r_n}}.
\end{equation*}
Starting from equation \eqref{Eq:PBTCycA2}, using equation \eqref{Eq:SxS} and the identity $f(y)\delta'_{xy}=f(x)\delta'_{xy}+\p_x\bigl(f(x)\bigr) \delta_{xy}$, we get
\begin{equation}\label{Eq:PBTTypeA}
\hspace{-4pt}\left\lbrace \Tc_n(\lambda,x), \Tc_m(\lambda,y) \right\rbrace = nm \, \Omega_{nm}(\lambda,x,y) + nm \, \Delta_{nm}(\lambda,x,y)
\end{equation}
where
\begin{subequations}
\begin{align}
\Omega_{nm}(\lambda,x,y)
&=  \lambda^{T-2}\zeta'(\lambda) \left( \Tc_{n+m-2}(\lambda,x) \delta'_{xy}  - \frac{1}{d}\Tc_{n-1}(\lambda,x)\Tc_{m-1}(\lambda,x)\delta'_{xy}\right. \label{Eq:DefOmega} \\
 & \hspace{40pt} \left. + \frac{m-1}{n+m-2} \p_x \left( \Tc_{n+m-2}(\lambda,x) \right) \delta_{xy} - \frac{1}{d}\Tc_{n-1}(\lambda,x) \p_x \left( \Tc_{m-1}(\lambda,x) \right) \delta_{xy} \right), \notag \\[2mm]
 \Delta_{nm}(\lambda,x,y) &=\lambda^{-2} \zeta(\lambda^T) \Tr\Bigl( S^{(0)}_{n-1}(\lambda,x) S^{(0)}_{m-1}(\lambda,y) \Bigr) \delta'_{xy}. \label{Eq:DefDelta}
\end{align}
\end{subequations}

We want to compute the Poisson brackets between the charges $\Q^0_n$ defined as integrals of the currents \eqref{Eq:KGenCyc3}. To begin with, note that the Poisson bracket between $F(\lambda,\mu,x)$ and $F(\lambda,\nu,y)$ can be obtained from equations \eqref{Eq:PowF} and \eqref{Eq:PBTTypeA}. We then find that
\begin{equation}\label{Eq:PBApAq}
\left\lbrace A(\lambda,\mu,x)^p, A(\lambda,\nu,y)^q \right\rbrace = pq \, A(\lambda,\mu,x)^p A(\lambda,\nu,y)^q \sum_{k,l=2}^{\infty} \bigl(\Omega_{kl}(\lambda,x,y)+\Delta_{kl}(\lambda,x,y) \bigr) \mu^k \nu^l.
\end{equation}
Up to a global factor and treating the spectral parameter $\lambda$ as an external parameter, $\Omega_{nm}(\lambda,x,y)$ has the same structure as the right hand side of equation \eqref{Eq:PBJTypeA}, which as we saw already had the same structure as equation (4.5) of~\cite{Evans:1999mj}. This equation (4.5) is used in~\cite{Evans:1999mj} to compute the Poisson brackets of the generating functions (equations (4.13) to (4.15)). The method developed in~\cite{Evans:1999mj} for computing these Poisson brackets then applies to the terms $\Omega_{kl}$ in equation \eqref{Eq:PBApAq} and gives a similar result, up to the global factor and the dependence on $\lambda$ that we keep. Specifically, we find
\begin{align} \label{Eq:PBGenCyc}
& \hspace{-30pt}\int \dd x \int \dd y \left\lbrace A(\lambda,\mu,x)^{p_n} \Bigr|_{\mu^n}, A(\lambda,\nu,y)^{p_m} \Bigr|_{\nu^m} \right\rbrace \\
& \hspace{30pt} = \, \left.  p_n p_m \, \lambda^{T-2}\zeta'(\lambda) \, \mu\nu \int \dd x \; A(\lambda,\mu,x)^{p_n} \p_x \bigl( A(\lambda,\nu,x)^{p_m} \bigr) h_{nm}(\mu,\nu) \;\; \right|_{\mu^n\nu^m} \notag  \\
& \hspace{50pt} + \left. p_n p_m \int \dd x \int \dd y  \; \sum_{k=2}^n \sum_{l=2}^m \Delta_{kl}(\lambda,x,y) \, A(\lambda,\mu,x)^{p_n} \Bigl|_{\mu^{n-k}} A(\lambda,\nu,y)^{p_m} \;\; \right|_{\nu^{n-l}}, \notag
\end{align}
where $h_{nm}(\mu,\nu)$ was defined in equation \eqref{Eq:hnm def}.

The first term on the right hand side of \eqref{Eq:PBGenCyc}, proportional to $h_{nm}(\mu,\nu)$, vanishes when $p_k = \frac{k-1}{d}$ for all $k \in \Z_{\geq 2}$, as expected. It therefore remains to show that the second term also does not contribute when we restrict to the $(r_n+r_m)$-th power of $\lambda$. From equation \eqref{Eq:DefDelta}, we see that the powers of $\lambda$ appearing in the power series expansion of $\Delta_{kl}(\lambda,x,y)$ are of the form $q_k+q_l-2+aT$, with $a\,\geq\alpha\geq\, 1$ and $q_n$ defined in subsection \ref{Sec:PBCurrentsCyc}.

The equivariance property \eqref{Eq:EquiT} can be rewritten as $\Tc_n(\omega\lambda,x)=\omega^{-n\kappa} \Tc_n(\lambda,x)$. In terms of the generating function $F(\lambda,\mu,x)$, this can be re-expressed as $F(\omega\lambda,\mu,x)=F(\lambda,\omega^{-\kappa} \mu,x)$. Thus, we also have $A(\omega\lambda,\mu,x)^p=A(\lambda,\omega^{-\kappa} \mu,x)^p$. Finally, we deduce that
\begin{equation}\label{Eq:EquiAGen}
A(\omega\lambda,\mu,x)^p \Bigr|_{\mu^k} = \omega^{r_k} A(\lambda,\mu,x)^p \Bigr|_{\mu^k}
\end{equation}
In particular, this implies that the powers of $\lambda$ appearing in $A(\lambda,\mu,x)^p \bigr|_{\mu^k}$ are of the form $r_k+bT$ with $b\geq 0$. In conclusion, the powers of $\lambda$ in
\begin{equation}\label{Eq:DeltaAA}
\Delta_{kl}(\lambda,x,y) \, A(\lambda,\mu,x)^p \Bigl|_{\mu^{n-k}} A(\lambda,\nu,y)^q \Bigr|_{\nu^{n-l}}
\end{equation}
are of the form $q_k+r_{n-k}+q_l+r_{m-l}-2+cT$, with $c\geq 1$.

Recall from subsection \ref{Sec:PBCurrentsCyc} that $q_k \equiv r_k+1 \, [T]$, and therefore $q_k+r_{n-k} \equiv r_n + 1 \, [T]$. Suppose now that $r_n$ and $r_m$ are different from, and so in particular strictly less than, $T-1$. As $q_k+r_{n-k}$ is always positive and congruent to $r_n+1$ modulo $T$, which is strictly less than $T$, it then follows that $q_k+r_{n-k} \geq r_n +1$. Similarly, we have $q_l+r_{m-l} \geq r_m+1$ and we thus deduce that $q_k+r_{n-k}+q_l+r_{m-l}-2+cT$ is greater than or equal to $r_n + r_m + T$. We deduce that the term \eqref{Eq:DeltaAA} cannot contribute to the $(r_n+r_m)^{\rm th}$ power of $\lambda$, as required.\\

In conclusion, we have found that, for any $n$ and $m$ such that $r_n$ and $r_m$ are different from $T-1$, one has $\left\lbrace \Q^0_n, \Q^0_m \right\rbrace = 0$, where the charge $\Q^0_k$ is defined as the integral of the current $\K_k^0(x)$ given by \eqref{Eq:KGenCyc3}.
Recall that the current $\J^0_n(x)$ is constructed as the coefficient of $\lambda^{r_n}$ in the power series expansion of $\Tc_n(\lambda,x)$. Similarly, one can rewrite equation \eqref{Eq:KGenCyc3} as
\begin{equation}\label{Eq:KW}
\K^0_n(x) = \W_n(\lambda,x)\Bigr|_{\lambda^{r_n}},
\end{equation}
with $\W_n$ defined in equation \eqref{Eq:DefW}.

\subsection{Summary}
\label{Sec:SummaryCyc}

Let us summarise the results of this section, as we did for non-cyclotomic zeros in subsection \ref{Sec:SummaryNonCyc}. In general, we define a charge $\Q_n^0$, associated with a current $\K_n^0(x)$, by
\begin{equation}\label{Eq:SummaryCyc}
\Q_n^0 = \int \dd x \; \K_n^0(x), \hspace{30pt} \text{ with } \hspace{30pt}
\K_n^0(x) = \W_n(\lambda,x) \Bigr|_{\lambda^{r_n}},
\end{equation}
where the definition of $\W_n(\lambda,x)$ and $r_n$ depends on the type of 
$\g$ and whether $\s$ is inner or not.

For $\g$ of type B, C or D (see subsection \ref{Sec:CycBCD}) and for $\g$ of type A and $\s$ 
 inner (see subsection \ref{Sec:TypeATrivial}), we simply choose 
 $\W_n(\lambda,x)=\Tc_n(\lambda,x)$, so that $\K_n^0(x)=\J_n^0(x)$. In this case, 
 we consider $r_n$ as the remainder of the euclidian division of $-n$ by $T$. When 
 $\g$ is of type A and $\s$ is not inner (the case discussed in subsections \ref{Sec:TypeANonTrivial} and \ref{Sec:CycGenerating}), we choose $\W_n(\lambda,x)$ as given by equation \eqref{Eq:SerieW}. In this case, $r_n$ is defined as the remainder of the euclidian division of $-n\big( 1+\frac{T}{2} \big)$ by $T$. \\

The equivariance properties \eqref{Eq:EquiT} and \eqref{Eq:EquiAGen} imply that, 
independently of the type of $\g$ and  of $\s$ being inner or not, we have
\begin{equation}
\W_n(\omega\lambda,x) = \omega^{r_n}\W_n(\lambda,x),
\end{equation}
for $\W_n$ defined as above. It therefore follows that the powers of $\lambda$ appearing in $\W_n(\lambda,x)$ are of the form $r_n+kT$, with $k\in\Z_{\geq 0}$. In particular, the current $\K^0_n(x)$ of equation \eqref{Eq:SummaryCyc} is the coefficient of the smallest power of $\lambda$ in $\W_n(\lambda,x)$.\\

As in the non-cyclotomic case, we restrict the degree $n$ of the currents $\K^0_n(x)$ to some subset $\E_0$ of $\Z_{\geq 2}$. More precisely, $n$ belongs to $\E_0$ if $n-1$ is an exponent of the affine algebra $\widehat{\g}$ and $r_n$ is different from $T-1$ (with the exception of the exponents related to the Pfaffian in type D, as in subsection \ref{Sec:SummaryNonCyc}). The results of this section can be summarised as the following theorem.

\begin{theorem} \label{thm: involution of Q0s}
Let $n,m\in\E_0$. Then the charges $\Q^0_n$ and $\Q^0_m$ are in involution, \textit{i.e.} we have
\begin{equation*}
\left\lbrace \Q^0_n, \Q^0_m \right\rbrace = 0.
\end{equation*}
\end{theorem}

There is a subtlety in the definition of $\W_n(\lambda,x)$ for $\g$ of type A 
and $\s$ inner. Indeed, in this case the current $\K_n^0(x)$ is extracted just from $\Tc_n(\lambda,x)$ as recalled above. Yet in section \ref{Sec:NonCycZero}, the current $\K^{\lambda_0}_n$ at a non-cyclotomic point was extracted instead from $\W_n(\lambda,x)$ which differs from $\Tc_n(\lambda,x)$. Thus, for this case, we  choose the appropriate definition of $\W_n(\lambda,x)$ depending on whether the regular zeros of the considered model are cyclotomic or not.\\

We end this section by an open question. For a non-cyclotomic regular zero $\lambda_0$ and $\g$ of type B, C or D, we considered local charges in involution $\Q_n^{\lambda_0}$ built as the integral of the currents $\J_n^{\lambda_0}(x)$ (see subsection \ref{Sec:NonCycZeroBCD}). However, based on the results of~\cite{Evans:1999mj}, we also exhibited a more general family of local charges in involution $\Q^{\lambda_0}_n(\xi)$, depending on a free parameter $\xi\in\R$ and whose corresponding currents $\K_n^{\lambda_0}(\xi,x)$ are constructed as polynomials in the $\J^{\lambda_0}_k$'s.

In the subsection \ref{Sec:CycBCD} of the present section, where we deal with a cyclotomic regular zero at the origin for $\g$ of type B, C or D, the charges $\Q^0_n$ are also constructed simply as integrals of the currents $\J^0_n(x)$. It is thus natural to ask if there exists in this case a more general family of charges $\Q^0_n(\xi)$, depending on a free parameter $\xi$, as for the non-cyclotomic case. Moreover, for $\g$ of type A and $\s$ inner (as treated in subsection \ref{Sec:TypeATrivial}), the charges $\Q^0_n$ are also integrals of the currents $\J_n^0(x)$ (we do not need to construct more complicated currents to obtain charges in involution). It is thus also natural to look for a more general family of charges $\Q^0_n(\xi)$. This would be an interesting result, as it would exhibit an important qualitative difference between the non-cyclotomic case and the cyclotomic one (with $\s$ inner), for $\g$ of type A.

We expect these one-parameter families (for $\g$ of type B, C and D or for $\g$ of type A with $\s$ inner) to exist. More precisely, we expect them to be given by the first non-zero coefficient in the power series expansions of a suitable generating function, depending on $\xi$, around the cyclotomic regular zero $\lambda=0$. As for the non-cyclotomic case, we will focus in this article on the local charges which do not depend on a free parameter $\xi$ (as described at the beginning of this subsection), for the same reasons as the ones discussed at the end of subsection \ref{Sec:NonCycZeroBCD} for a non-cyclotomic zero.

\section{Properties of the local charges} \label{Sec:PrOfLocCha}

\subsection{Algebra of local charges in involution}
\label{Sec:AlgebraLoc}

In the previous sections, we constructed a tower of local charges in involution at every regular zero $\lambda_0 \in \Zc$. More precisely, we constructed currents $\K_n^{\lambda_0}(x)$, with degrees $n$ in some subset $\E_{\lambda_0}$ of $\Z_{\geq 2}$, such that the charges $\Q^{\lambda_0}_n$ defined as the integral of $\K^{\lambda_0}_n(x)$ are in involution with one another. In this subsection, we study the whole algebra of local charges in involution, formed by all the $\Q^{\lambda_0}_n$ for $\lambda_0\in \Zc$ and $n\in\E_{\lambda_0}$. More precisely, we prove that currents $\K_n^{\lambda_0}(x)$ and $\K_m^{\mu_0}(y)$ extracted at different regular zeros are in involution. We establish this result for regular zeros in $\Zc$, excluding the point at infinity. If infinity is a regular zero then one can also extract charges in involution $\Q^\infty_n$, using the results of subsection \ref{Sec:Infinity}. The Poisson brackets of the corresponding currents with the currents at finite regular zeros involve some subtleties and will be treated separately at the end of the subsection.\\

In general, the currents $\K_n^{\lambda_0}(x)$ and $\K_m^{\mu_0}(y)$ are constructed as polynomials of the currents $\J_n^{\lambda_0}(x)$ and $\J_m^{\mu_0}(y)$. It is therefore sufficient to prove that the Poisson bracket of $\J_n^{\lambda_0}(x)$ and $\J_m^{\mu_0}(y)$ is zero. The currents $\J_n^{\lambda_0}(x)$ and $\J_m^{\mu_0}(y)$ are extracted from $\Tc_n(\lambda,x)$ and $\Tc_m(\mu,y)$, whose Poisson bracket is given by equation \eqref{Eq:PBT}. We can suppose that $\mu_0$ is different from $0$ and thus is a non-cyclotomic point, so that $\J^{\mu_0}_m(y) = \Tc_m(\mu_0,y)$. Using the fact that $U\ti{12}(\lambda,\mu_0)=\varphi(\lambda)\Rc^0\ti{12}(\lambda,\mu_0)$ since $\varphi(\mu_0)=0$, we can evaluate equation \eqref{Eq:PBT} at $\mu=\mu_0$ to find
\begin{equation}\label{Eq:PBTJ}
\left\lbrace \Tc_n(\lambda,x), \J^{\mu_0}_m(y) \right\rbrace = -nm \, \varphi(\lambda) \Tr\ti{12} \Bigl( \Rc^0\ti{12}(\lambda,\mu_0) S_{n-1}(\lambda,x)\ti{1}S_{m-1}(\mu_0,y)\ti{2} \Bigr) \delta'_{xy}.
\end{equation}
We will now treat separately the cases $\lambda_0$ cyclotomic or $\lambda_0$ non-cyclotomic.\\

Suppose that $\lambda_0$ is non-cyclotomic so that $\J^{\lambda_0}_n(x)$ is simply the evaluation of $\Tc_n(\lambda,x)$ at $\lambda=\lambda_0$. Recall from paragraph \ref{Sec:Equi} that, by construction of the set $\Zc$, as $\lambda_0$ and $\mu_0$ are different elements of $\Zc$, the cyclotomic orbits $\Z_T\lambda_0$ and $\Z_T\mu_0$ are disjoint. In particular, this means that $\Rc^0\ti{12}(\lambda,\mu_0)$ is regular at $\lambda=\lambda_0$. Indeed, by equation \eqref{Eq:RCyc} the poles of $\Rc^0\ti{12}(\lambda,\mu_0)$ are the points of the orbit $\Z_T\mu_0$. Moreover, $S_{n-1}(\lambda,x)$ is regular at $\lambda=\lambda_0$ and we have $\varphi(\lambda_0)=0$. Thus, evaluating equation \eqref{Eq:PBTJ} at $\lambda=\lambda_0$ we find that the currents $\J^{\lambda_0}_n(x)$ and $\J^{\mu_0}_m(y)$ are in involution, as expected.\\

Let us now treat the cyclotomic case where $\lambda_0=0$. In this case, $\J^0_n(x)$ is the coefficient of $\lambda^{r_n}$ in the power series expansion of $\Tc_n(\lambda,x)$ (cf. subsection \ref{Sec:EquivT}). The Poisson bracket of $\J^0_n(x)$ with $\J^{\mu_0}_m(y)$ is then the coefficient of $\lambda^{r_n}$ in equation \eqref{Eq:PBTJ}. Recall from section \ref{Sec:CycZero} that for $n\in\E_0$, we have $r_n < T-1$. Yet, in equation \eqref{Eq:PBTJ}, $\Rc^0\ti{12}(\lambda,\mu_0)$ and $S_{n-1}(\lambda,x)$ are regular at $\lambda=0$ and $\varphi(\lambda)=O(\lambda^{T-1})$, hence the involution of $\J^0_n(x)$ and $\J^{\mu_0}_m(y)$. In conclusion, we have proved

\begin{theorem}\label{Thm:DiffZeros}
Let $\lambda_0,\mu_0 \in \Zc$ and let $n\in\E_{\lambda_0}$ and $m\in\E_{\mu_0}$. Then, if $\lambda_0\neq\mu_0$, we have
\begin{equation*}
\left\lbrace \J^{\lambda_0}_n(x), \J^{\mu_0}_m(y) \right\rbrace = 0.
\end{equation*}
\end{theorem}
Combining this theorem with the results of previous sections, we conclude that the local charges $\Q^{\lambda_0}_n$, for all $\lambda_0\in\Zc$ and $n\in\E_{\lambda_0}$, are in involution with one another.\\

We now turn to the case where one of the regular zeros is the point at infinity. In this case, the current $\J^\infty_m(y)$ is extracted from the Lax matrix $\Lc^\infty(\alpha,y)$. From the Poisson brackets \eqref{Eq:PBR} and \eqref{Eq:PBLC}, we get
\begin{align}\label{Eq:PBLLI}
& \left\lbrace \Lc(\lambda,x)\ti{1}, \Lc^\infty(\alpha,y)\ti{2} \right\rbrace = \left[ \Rct\ti{12}(\lambda,\alpha^{-1}), \Lc(\lambda,x)\ti{1} \right] \delta_{xy} - \left[ \Rc\ti{21}(\alpha^{-1},\lambda), \Lc^\infty(\alpha,y)\ti{2} \right] \delta_{xy} \\
& \hspace{60pt}  - \; \bigl( \Rct\ti{12}(\lambda,\alpha^{-1}) + \Rc\ti{21}(\alpha^{-1},\lambda) \bigr) \delta'_{xy}  + \alpha^{-1}\psi(\alpha)^{-1} \left[ \Rc\ti{21}(\alpha^{-1},\lambda), \Cc(x)\ti{2} \right] \delta_{xy}, \notag
\end{align}
with the matrix $\Rct$ defined in \eqref{Eq:DefRct}.

In the following, we will say that an equation is true \textit{weakly}, and we will then use the symbol $\approx$ instead of $=$, if the equation holds when one puts the field $\Cc(x)$ to zero. This denomination and its interest will be made clear when studying $\Z_T$-coset models, in which case the field $\Cc(x)$ is interpreted as a gauge constraint. Note, in particular, that the last term of equation \eqref{Eq:PBLLI} vanishes weakly. Using Corollary \ref{Cor:PBTr}, we may then compute the Poisson brackets of $\Tc_n(\lambda,x)$ with
\begin{equation*}
\Tc^\infty_m(\alpha,y)=\Tr\bigl(S^\infty_n(\alpha,x)\bigr)=\Tr\bigl( \psi(\alpha)^n \Lc^\infty(\alpha,x)^n \bigr)
\end{equation*}
weakly, to find
\begin{equation}\label{Eq:PBTTI}
\left\lbrace \Tc_n(\lambda,x), \Tc^\infty_m(\alpha,y) \right\rbrace \approx -nm \, \Tr\ti{12} \Bigl( V\ti{12}(\lambda,\alpha) S_{n-1}(\lambda,x)\ti{1}S^\infty_{m-1}(\alpha,y)\ti{2} \Bigr) \delta'_{xy},
\end{equation}
where
\begin{equation*}
V\ti{12}(\lambda,\alpha) = - \alpha^{-2} \varphi(\lambda) \Rct^0\ti{12}(\lambda,\alpha^{-1}) + \psi(\alpha) \Rc^0\ti{21}(\alpha^{-1},\lambda).
\end{equation*}
Suppose first that $\lambda_0\in\Zc$ is non-cyclotomic. We have $\varphi(\lambda_0)=0$, and hence
\begin{equation*}
\left\lbrace \J^{\lambda_0}_n(x), \Tc^\infty_m(\alpha,y) \right\rbrace \approx -nm \, \psi(\alpha)  \Tr\ti{12} \Bigl( \Rc^0\ti{21}(\alpha^{-1},\lambda_0) S_{n-1}(\lambda_0,x)\ti{1}S^\infty_{m-1}(\alpha,y)\ti{2} \Bigr) \delta'_{xy}.
\end{equation*}
The Poisson bracket between $\J^{\lambda_0}_m(x)$ and $\J^\infty_m(y)$ is then obtained, weakly, by extracting the coefficient of $\alpha^{r_m}$ in the equation above. For $m\in\E_\infty$, we have $r_m < T-1$. Yet $\psi(\alpha)=O(\alpha^{T-1})$ and $ \Rc^0\ti{21}(\alpha^{-1},\lambda_0)$ and $S^\infty_{m-1}(\alpha,y)$ are regular at $\alpha=0$. Thus $\J^{\lambda_0}_n(x)$ and $\J^\infty_m(y)$ are weakly in involution.

It remains to consider the case where $\lambda_0=0$. In this case, $\J^0_n(x)$ is the coefficient of $\lambda^{r_n}$ in $\Tc_n(\lambda,x)$ and we restrict to $n$ such that $r_n < T-1$. We have to find the coefficient of $\lambda^{r_n}\alpha^{r_m}$ in equation \eqref{Eq:PBTTI}. Due to the presence of $\varphi(\lambda)$ or $\psi(\alpha)$ in the two terms appearing in $V\ti{12}(\lambda,\alpha)$, we see that either the power of $\lambda$ or the power of $\alpha$ is greater than $T-1$ and thus cannot contribute to the term $\lambda^{r_n}\alpha^{r_m}$. In conclusion, we have the following theorem.

\begin{theorem}\label{Thm:InvolutionInfinity}
Suppose that infinity is a regular zero of the model. Let $\lambda_0\in\Zc$, $n\in\E_{\lambda_0}$ and $m\in\E_{\infty}$. Then we have
\begin{equation*}
\left\lbrace \J^{\lambda_0}_n(x), \J^{\infty}_m(y) \right\rbrace \approx 0.
\end{equation*}
\end{theorem}

Combining this theorem with the results of previous sections, we see that the local charges $\Q^{\lambda_0}_n$, for all $\lambda_0\in\Zc\cup\lbrace\infty\rbrace$ and $n\in\E_{\lambda_0}$, are (at least weakly) in involution with one another.\\

We thus constructed a large algebra of local charges in involution, composed of the charges $\Q^{\lambda_0}_n$, with $\lambda_0$ regular zeros and $n\in\E_{\lambda_0}$. Since these charges are local, they are also in involution with the momentum $\Pc$ of the theory, whose Poisson bracket generates the derivative with respect to the spatial coordinate $x$. We have not yet discussed the conservation properties of these charges. For the models that we consider as examples in this article (listed in paragraph \ref{Sec:Examples}), we will see in section \ref{Sec:Applications} that the Hamiltonian $\Hc$ of the theory always belongs to the algebra of local charges described above. It therefore follows that all these charges are conserved. More precisely, we will see that $\Hc$ is always a linear combination of the quadratic charges $\Q^{\lambda_0}_2$ and the momentum $\Pc$.

\subsection{Gauge invariance}
\label{Sec:Gauge}

In this subsection, we anticipate the application of the construction developed in the previous sections to integrable $\s$-models on $\Z_T$-coset spaces. In those models, infinity is a regular zero and the corresponding field $\Cc(x)$ (cf. subsection \ref{Sec:Infinity}) is a gauge constraint. We prove here that the currents $\K_n^{\lambda_0}(x)$ constructed at regular zeros $\lambda_0$ in the previous sections are gauge invariant, in the sense that they Poisson commute with the constraint $\Cc(y)$. As the $\K_n^{\lambda_0}$'s are polynomials of the $\J_n^{\lambda_0}$'s, it is enough to prove the following theorem.

\begin{theorem}\label{Thm:Gauge}
Suppose that infinity is a regular zero. Let $\lambda_0 \in \Zc\cup\lbrace \infty \rbrace$ and $n\in\E_{\lambda_0}$. Then, we have
\begin{equation*}
\left\lbrace \J^{\lambda_0}_n(x), \Cc(y) \right\rbrace = 0.
\end{equation*}
\end{theorem}

\begin{proof}
Let us first suppose that $\lambda_0$ is different from infinity. The current $\J^{\lambda_0}_n(x)$ is extracted from $\Tc_n(\lambda,x)=\varphi(\lambda)^n\Lc(\lambda,x)^n$. Recall the Poisson bracket between the Lax matrix $\Lc(\lambda,x)$ and $\Cc(y)$, given by equation \eqref{Eq:PBLC}. By Corollary \ref{Cor:PBTr} we then have
\begin{equation*}
\left\lbrace \Tc_n(\lambda,x), \Cc(y) \right\rbrace = - n\, \varphi(\lambda) \, \Tr\ti{1} \bigl( C^{(0)}\ti{12} S_{n-1}(\lambda,x)\ti{1} \bigr) \delta'_{xy}.
\end{equation*}

If $\lambda_0\neq 0$, $\J_n^{\lambda_0}(x)$ is simply $\Tc_n(\lambda_0,x)$. And since $\lambda_0$ is a regular zero, $S_{n-1}(\lambda,x)$ is regular at $\lambda=\lambda_0$ and $\varphi(\lambda_0)=0$. Evaluating the above Poisson bracket at $\lambda=\lambda_0$, we get the involution of $\J^{\lambda_0}_n(x)$ and $\Cc(y)$. If $\lambda_0=0$, $\J^0_n(x)$ is the coefficient of $\lambda^{r_n}$ in $\Tc_n(\lambda,x)$ and, since $n\in\E_0$, we have $r_n < T-1$. Moreover, since $S_{n-1}(\lambda,x)$ is regular at $\lambda=0$ and $\varphi(\lambda)=O(\lambda^{T-1})$, the $\lambda^{r_n}$-term in the Poisson bracket above is then zero, as required.

Finally, let us treat the case $\lambda_0=\infty$, for which $\J^\infty_n(x)$ is given by the coefficient of $\alpha^{r_n}$ in $\Tc^\infty_n(\alpha,x)= \Tr\big( \psi(\alpha)^n \Lc^\infty(\alpha,x)^n \big)$. Using the definition of $\Lc^\infty$ and the Poisson brackets \eqref{Eq:PBLC} and \eqref{Eq:PBCC}, we find
\begin{equation*}
\left\lbrace \Lc^\infty(\alpha,x)\ti{1}, \Cc(y)\ti{2} \right\rbrace = \bigl[ C^{(0)}\ti{12}, \Lc^\infty(\alpha,x)\ti{1} \bigr] \delta_{xy} - C^{(0)}\ti{12}\delta'_{xy}.
\end{equation*}
This bracket has the same structure as equation \eqref{Eq:PBLC}. Therefore, the case $\lambda_0=\infty$ is treated exactly in the same way than the case $\lambda_0=0$, which ends the proof.
\end{proof}

\subsection{Reality conditions}
\label{Sec:Reality}

To close this section let us discuss the reality conditions on the charges $\Q^{\lambda_0}_n$ extracted at regular zeros in the previous sections. In the examples of paragraph \ref{Sec:Examples}, we consider integrable $\s$-models with target space $G_0$ or a quotient of $G_0$, where $G_0$ is a real Lie group. If $\g_0$ is the Lie algebra of $G_0$, then the Lax matrix of the model is a $\g$-valued field, where $\g$ is the complexification of $\g_0$. In other words, $\g_0$ is a real form of the complex Lie algebra $\g$. The fact that the $\s$-models we consider are on the real form $G_0$ (or one of its quotient) is encoded on the Lax matrix as the following reality condition:
\begin{equation}\label{Eq:Reality}
\tau\bigl(\Lc(\lambda,x)\bigr) = \Lc(\bar{\lambda},x),
\end{equation}
where the bar denotes complex conjugation and $\tau$ is the involutive semi-linear automorphism of $\g$ characterizing the real form $\g_0$.
Moreover, we also have a reality condition on the twist function, which simply reads
\begin{equation}\label{Eq:TwistReal}
\overline{\varphi(\lambda)}=\varphi(\bar{\lambda}).
\end{equation}
In particular, if $\lambda_0$ is a zero of $\varphi$, its conjugate $\bar{\lambda}_0$ is also a zero of $\varphi$. Combining the reality conditions \eqref{Eq:Reality} and \eqref{Eq:TwistReal}, we also see that if $\lambda_0$ is a regular zero (see paragraph \ref{Sec:Equi}), $\bar\lambda_0$ is also a regular zero.
Thus, the regular zeros can be of two types: real ones $\lambda_0\in\R$ and conjugate pairs $\lambda_0$, $\bar\lambda_0$.
\\

We will use the reality condition \eqref{Eq:Reality} in a similar way to the way we used the equivariance property \eqref{Eq:EquiL} in subsection \ref{Sec:EquivT}. In particular, as we consider powers of the Lax matrix, which are not in the Lie algebra $\g$ in general, we will need to extend ``naturally'' the automorphism $\tau$ to the whole algebra $F$ of matrices acting on the defining representation of $\g$. This was done for the automorphism $\s$ in subsection \ref{Sec:EquivT} and appendix \ref{App:ExtSigma}. One can apply similar ideas to $\tau$, using the classification of real forms of the classical Lie algebras A, B, C and D. We do not present the details here and just summarise the results.

There exists an extension of $\tau$ on the whole algebra of matrices $F$, which coincides with $\tau$ when restricted to the Lie algebra $\g$, and that we shall still denote $\tau$. This extension is still an involutive semi-linear map of $F$ to itself. However, it is not in general an algebra homomorphism. The main properties of the extension $\tau$ that we will need are the following. There exists $\gamma\in\lbrace 1, -1 \rbrace$ such that
\begin{subequations}\label{Eq:PropTau}
\begin{align}
\tau(Z^n) &= \gamma^{n-1} \, \tau(Z)^n, \label{Eq:TauPow}\\
\Tr\bigl(\tau(Z)\bigr) &= \gamma\, \overline{\Tr(Z)}, \label{Eq:TauTr}
\end{align}
\end{subequations}
for any $Z\in F$. For every real form $\g_0$ of a classical algebra $\g$ we have $\gamma = 1$, except for the real forms $\mathfrak{su}(p,q,\R)$ of $\sl(d,\C)$ (with $p+q=d$), for which $\gamma=-1$. Using the properties \eqref{Eq:PropTau} with the reality conditions \eqref{Eq:Reality} and \eqref{Eq:TwistReal}, one finds that $\tau\bigl(S_n(\lambda,x)\bigr) = \gamma^{n-1} S_n(\bar{\lambda},x)$ and that
\begin{equation}\label{Eq:TReal}
\overline{\Tc_n(\lambda,x)} = \gamma^n \Tc_n(\bar{\lambda},x).
\end{equation}

Consider a regular zero $\lambda_0$. Suppose first that $\lambda_0$ is complex: its conjugate $\bar{\lambda}_0$ is then also a regular zero. According to the previous sections, we can extract two towers of (possibly complex) currents $\J^{\lambda_0}_n(x)$ and $\J^{\bar\lambda_0}_n(x)$ by evaluating $\Tc_n(\lambda,x)$ at $\lambda=\lambda_0$ or $\lambda=\bar\lambda_0$ (note that $\lambda_0$ cannot be a cyclotomic point as it is complex). However, according to equation \eqref{Eq:TReal}, these currents are not independent. Indeed, they are related by the reality condition
\begin{equation*}
\J^{\bar\lambda_0}_n(x) = \gamma^n \overline{ \J^{\lambda_0}_n(x) }.
\end{equation*}
Thus, considering linear combination of $\Q^{\lambda_0}_n$ and $\Q^{\bar\lambda_0}_n$, we extract from each pair $\lambda_0$, $\bar\lambda_0$ of complex regular zeros two towers of real charges in involution: $\Q^{\lambda_0}_n+\gamma^n\Q^{\bar\lambda_0}_n$ and $i\bigl(\Q^{\lambda_0}_n-\gamma^n\Q^{\bar\lambda_0}_n\bigr)$.\\

Suppose now that $\lambda_0$ is a real and non-cyclotomic regular zero. Equation \eqref{Eq:TReal} then imposes the reality condition
\begin{equation}\label{Eq:RealZero}
\J^{\lambda_0}_n(x) = \gamma^n  \overline{ \J^{\lambda_0}_n(x)}.
\end{equation}
Thus, the current $\J^{\lambda_0}_n$ is either real or pure imaginary. In each case, we can extract only one tower of real local charges.
Consider now the case where $\lambda_0$ is the origin and thus a cyclotomic real point. The current $\J^0_n(x)$ is then the coefficient of $\lambda^{r_n}$ in the power series expansion of $\Tc_n(\lambda,x)$. Yet, this coefficient is also the one of $\bar{\lambda}^{r_n}$ in the power series expansion of $\Tc_n(\bar{\lambda},x)$. The reality condition \eqref{Eq:TReal} then implies that equation \eqref{Eq:RealZero} also holds for $\lambda_0=0$.

Finally, let us discuss the case where $\lambda_0$ is infinity, which we consider as a real point. From the reality conditions \eqref{Eq:Reality} and \eqref{Eq:TwistReal}, we find that the field $\Cc(x)$ defined in subsection \ref{Sec:Infinity} is real, in the sense that $\tau\bigl(\Cc(x)\bigr)=\Cc(x)$. We then obtain reality conditions on the Lax matrix $\Lc^\infty(\alpha,x)$ and the twist function $\psi(\alpha)$ similar to equations \eqref{Eq:Reality} and \eqref{Eq:TwistReal}. As a result we can apply the above discussion, since the point at infinity in the variable $\lambda$ corresponds to the origin in the variable $\alpha$, and conclude that equation \eqref{Eq:RealZero} also holds for $\lambda_0=\infty$.\\

\noi To summarise this subsection, we have shown that one can extract:
\begin{itemize}\setlength\itemsep{0.2em}
\item one tower of real local charges for each real regular zero $\lambda_0$,
\item two towers of real local charges for each pair $\lambda_0$, $\bar\lambda_0$ of complex regular zeros.
\end{itemize}
In other words, one can extract as many towers of real charges as there are regular zeros.

\section{Integrable hierarchies and zero curvature equations}
\label{Sec:IntHierZeroCurv}

In the previous sections, we constructed a infinite set of local charges $\Q^{\lambda_0}_n$ in involution, with $\lambda_0$ regular zeros. It induces an infinite set of commuting Hamiltonian flows, defined by $\left\lbrace \Q^{\lambda_0}_n, \cdot \right\rbrace$. In this section, we show that these flows generate a hierarchy of integrable equations. More precisely, we associate with each charge $\Q^{\lambda_0}_n$ a connection
\begin{equation*}
\nabla^{\lambda_0}_n = \left\lbrace \Q^{\lambda_0}_n, \cdot \right\rbrace + \M^{\lambda_0}_n (\lambda,x)
\end{equation*}
which commutes with the connection $\nabla_x = \p_x+\Lc(\lambda,x)$. We show that the connections $\nabla^{\lambda_0}_n$ also commute with one another for finite regular zeros $\lambda_0$. The commutativity of these connections takes the form of zero curvature equations. In particular, we will use the zero curvature equations involving $\Lc(\lambda,x)$ and the $\M^{\lambda_0}_n(\lambda,x)$'s to prove that the local charges $\Q^{\lambda_0}_n$ are in involution with the non-local charges extracted from the monodromy of $\Lc(\lambda,x)$.

\subsection{Zero curvature equations with $\Lc$}
\label{Sec:ZCEL}

The starting point of this article is an integrable system with Lax matrix $\Lc(\lambda,x)$ and Hamiltonian $\Hc$. The dynamical equations of this system are generated by the Poisson bracket with $\Hc$. They are encoded in the form of a zero curvature equation
\begin{equation}\label{Eq:ZCEH}
\left\lbrace \Hc, \Lc(\lambda,x) \right\rbrace - \p_x \M(\lambda,x) + \left[ \M(\lambda,x), \Lc(\lambda,x) \right] = 0
\end{equation}
on the Lax matrix $\Lc(\lambda,x)$ and a $\g$-valued matrix $\M(\lambda,x)$. In this subsection, we study the dynamics of the Lax matrix under the Hamiltonian flows generated by the local charges $\Q^{\lambda_0}_n$ constructed in the previous sections. More precisely, we show that these dynamics also take the form of a zero curvature equation on $\Lc(\lambda,x)$:

\begin{theorem}\label{Thm:ZCEL}
Let $\lambda_0\in\Zc$ and $n\in\E_{\lambda_0}$. There exists a matrix $\M^{\lambda_0}_n(\lambda,x)$ such that we have the zero curvature equation
\begin{equation*}
\left\lbrace \Q^{\lambda_0}_n, \Lc(\lambda,x) \right\rbrace - \p_x \M^{\lambda_0}_n(\lambda,x) + \left[ \M^{\lambda_0}_n(\lambda,x), \Lc(\lambda,x) \right] = 0.
\end{equation*}
\end{theorem}

\begin{proof}
Let us apply the second result of Corollary \ref{Cor:PBTr} to the $r/s$-system \eqref{Eq:PBR}. Using the form \eqref{Eq:DefR} of the $\Rc$-matrix, we find
\begin{align}\label{Eq:PBLT}
& \left\lbrace \Lc(\lambda,x), \Tc_n(\mu,y) \right\rbrace = n \left[ \Tr\ti{2}\Bigl( \Rc^0\ti{12}(\lambda,\mu) S_{n-1}(\mu,y)\ti{2} \Bigr), \Lc(\lambda,x) \right] \delta_{xy}  \\
& \hspace{75pt} - n \, \Tr\ti{2} \Bigl( \Rc^0\ti{12}(\lambda,\mu) S_{n-1}(\mu,y)\ti{2} \Bigr) \delta'_{xy}  - n \frac{\varphi(\mu)}{\varphi(\lambda)} \, \Tr\ti{2} \Bigl( \Rc^0\ti{21}(\mu,\lambda)  S_{n-1}(\mu,y)\ti{2} \Bigr) \delta'_{xy}. \notag
\end{align}

Consider first the case where $\lambda_0$ is a non-cylotomic regular zero. Evaluating the equation above at $\mu=\lambda_0$ and using $\varphi(\lambda_0)=0$, we have
\begin{equation}\label{Eq:ZCEJ}
\left\lbrace \Lc(\lambda,x), \J^{\lambda_0}_n(y) \right\rbrace
 = \left[ \mathcal{N}^{\lambda_0}_n(\lambda,y), \Lc(\lambda,x) \right] \delta_{xy}  - \mathcal{N}^{\lambda_0}_n(\lambda,y) \delta'_{xy},
\end{equation}
where
\begin{equation}\label{Eq:DefNNonCyc}
\mathcal{N}^{\lambda_0}_n(\lambda,x) = n \, \Tr\ti{2} \Bigl( \Rc^0\ti{12} (\lambda,\lambda_0) S_{n-1}(\lambda_0,x)\ti{2} \Bigr).
\end{equation}

Suppose now that $\lambda_0$ is the origin, which is a cyclotomic point, in which case $\J^0_n(y)$ is constructed as the coefficient of $\mu^{r_n}$ in the power series expansion of $\Tc_n(\mu,y)$. Moreover, as $n\in\E_0$, we have $r_n<T-1$ (see section \ref{Sec:CycZero}). 
The Poisson bracket $\left\lbrace \Lc(\lambda,x), \J^0_n(y) \right\rbrace$ is thus the $\mu^{r_n}$-term in equation \eqref{Eq:PBLT}. We have $\varphi(\mu)=O(\mu^{T-1})$ and $r_n <T-1$, thus the last term of equation \eqref{Eq:PBLT} cannot contribute to $\mu^{r_n}$. Thus, we also have equation \eqref{Eq:ZCEJ} for $\lambda_0=0$, with
\begin{equation*}
\Nc^0_n(\lambda,x) = n \, \Tr\ti{2} \Bigl( \Rc^0\ti{12} (\lambda,\mu) S_{n-1}(\mu,x)\ti{2} \Bigr) \Bigr|_{\mu^{r_n}}.
\end{equation*}

We will say that $\mathcal{N}^{\lambda_0}_n$ is the Lax matrix associated with the charge defined as the integral of the current $\J_n^{\lambda_0}$. Equation \eqref{Eq:ZCEJ} implies a zero curvature equation for the evolution of $\Lc(\lambda,x)$ under the Hamiltonian flow of this charge. In general, the charge $\Q^{\lambda_0}_n$ is not the integral of $\J_n^{\lambda_0}$ but of $\K_n^{\lambda_0}$ (see previous sections). Recall that $\K_n^{\lambda_0}$ is a polynomial in the $\J_m^{\lambda_0}$'s. We construct the Lax matrix $\M^{\lambda_0}_n(\lambda,x)$ associated with $\Q_n^{\lambda_0}$ by assigning any monomial $\J^{\lambda_0}_{m_1} \ldots \J^{\lambda_0}_{m_p}$ in this polynomial to the matrix
\begin{equation*}
\sum_{k=1}^p \Bigl( \prod_{j\neq k} \J_{m_j}^{\lambda_0}(x) \Bigr) \mathcal{N}^{\lambda_0}_{m_k} (\lambda,x).
\end{equation*}
Using the fact that the Poisson bracket is a derivation, we find from equation \eqref{Eq:ZCEJ} that
\begin{equation}\label{Eq:ZCEK}
\left\lbrace \Lc(\lambda,x), \K^{\lambda_0}_n(y) \right\rbrace
 = \left[ \M^{\lambda_0}_n(\lambda,y), \Lc(\lambda,x) \right] \delta_{xy} - \M^{\lambda_0}_n(\lambda,y) \delta'_{xy}.
\end{equation}
After integration over $y$, we get the required zero curvature equation.
\end{proof}

Thus, the Hamiltonian flows of the charges $\Q^{\lambda_0}_n$ generate dynamical equations that can be recast in the form of zero curvature equations. In conclusion, we have constructed a hierarchy of integrable systems with Lax matrix $\Lc(\lambda,x)$ and Hamiltonians $\Q^{\lambda_0}_n$. The zero curvature equations of Theorem \ref{Thm:ZCEL} can be seen as the commutativity of the connections
\begin{equation}\label{Eq:Nabla}
\nabla^{\lambda_0}_n = \left\lbrace \Q^{\lambda_0}_n, \cdot \right\rbrace + \M^{\lambda_0}_n (\lambda,x)
\end{equation}
with $\nabla_x = \p_x + \Lc(\lambda,x)$. This connection $\nabla_x$ can be thought of as the connection associated with the local momentum $\Pc$ of the theory. As already mentioned, we will see in section \ref{Sec:Applications} that for the models we consider, the Hamiltonian is given by a linear combination $\Hc=\sum_{\lambda_0\in\Zc} a_{\lambda_0} \Q^{\lambda_0}_2 + b \Pc$ of the quadratic charges $\Q^{\lambda_0}_2$ and the momentum $\Pc$. Therefore, the matrix $\M(\lambda,x)$ of equation \eqref{Eq:ZCEH} can be constructed as $\sum_{\lambda_0\in\Zc} a_{\lambda_0} \M^{\lambda_0}_2(\lambda,x) + b \Lc(\lambda,x)$.\\

Theorem \ref{Thm:ZCEL} only treats the case of finite regular zeros $\lambda_0$. Let us also briefly discuss what happens when $\lambda_0=\infty$. In this case, $\J^\infty_n(x)$ is extracted from the Lax matrix $\Lc^\infty(\alpha,x)$. Since this matrix satisfies an $r/s$-system with twist function $\psi(\alpha)$, one can apply the method developed here. Doing so we find that the dynamics of $\Lc^\infty(\alpha,x)$ under the Hamiltonian flow of $\Q^\infty_n$ takes the form of a zero curvature equation. Moreover, starting with the Poisson bracket \eqref{Eq:PBLLI} and working weakly, we also find a weak curvature equation
\begin{equation*}
\left\lbrace \Q^{\infty}_n, \Lc(\lambda,x) \right\rbrace - \p_x \M^{\infty}_n(\lambda,x) + \left[ \M^{\infty}_n(\lambda,x), \Lc(\lambda,x) \right] \approx 0,
\end{equation*}
where the matrix $\M^{\infty}_n(\lambda,x)$ is constructed from
\begin{equation*}
\Nc^\infty_n(\lambda,x) = - n \, \alpha^{-2} \, \Tr\ti{2} \Bigl( \Rct^0\ti{12} (\lambda,\alpha^{-1}) S^\infty_{n-1}(\alpha,x)\ti{2} \Bigr) \Bigr|_{\alpha^{r_n}}
\end{equation*}
in the same way as $\M^{\lambda_0}_n(\lambda,x)$ was built from $\Nc^{\lambda_0}_n(\lambda,x)$ for a finite regular zero $\lambda_0$. In other words, Theorem \ref{Thm:ZCEL} also applies for $\lambda_0=\infty$ when Poisson brackets are considered weakly.\\

Let us end this subsection by stating a few properties of the Lax matrix $\M^{\lambda_0}_n(\lambda,x)$. Using the equivariance property \eqref{Eq:EquiR}, we find that
\begin{equation*}
\s \left( \M^{\lambda_0}_n(\lambda,x) \right) = \M^{\lambda_0}_n(\omega\lambda,x).
\end{equation*}
The Lax matrix $\M^{\lambda_0}_n$ thus satisfies the same equivariance property \eqref{Eq:EquiL} as $\Lc$. Recall that the Lax matrix $\Nc^0_n(\lambda,x)$ is extracted as the $\mu^{r_n}$-term in
\begin{equation}\label{Eq:DefN}
\Nc_n(\mu \,;\lambda,x) =  n \, \Tr\ti{2} \Bigl( \Rc^0\ti{12} (\lambda,\mu) S_{n-1}(\mu,x)\ti{2} \Bigr).
\end{equation}
Consider the equivariance properties \eqref{Eq:EquiS} and
\begin{equation*}
\s\ti{2} \Rc^0\ti{12}(\lambda,\mu) = \omega\Rc^0\ti{12}(\lambda,\omega\mu).
\end{equation*}
Combining it with the fact that $\Tr\bigl(\s(Y)\s(Z)\bigr)=\Tr(YZ)$ for any matrices $Y,Z\in F$ (see appendix \ref{App:ExtSigma}), we find that
\begin{equation}\label{Eq:EquiN}
\Nc_n(\omega\mu\,;\lambda,x) = \omega^{r_n} \Nc_n(\mu\,;\lambda,x).
\end{equation}
Therefore, the power series expansion of $\Nc_n(\mu\,; \lambda,x)$ in $\mu$ contains powers of the form $r_n+kT$, with $k\in\Z_{\geq 0}$. In particular, $\Nc^0_n(\lambda,x)$ is the coefficient of the smallest power in this expansion, in the same way as $\J^0_n(x)$ is in the expansion of $\Tc_n(\mu,x)$.

Let us define $\M_n(\mu\,;\lambda,x)$ from $\Nc_n(\mu\,;\lambda,x)$ and $\Tc_n(\mu,x)$ in the same way we constructed $\M_n^{\lambda_0}(\lambda,x)$ from $\Nc^{\lambda_0}_n(\lambda,x)$ and $\J^{\lambda_0}_n(x)$. In particular, $\M^{\lambda_0}_n(\lambda,x)$ is the evaluation of $\M_n(\mu\,;\lambda,x)$ at $\mu=\lambda_0$. From equations \eqref{Eq:EquiTrA} and \eqref{Eq:EquiN}, we find the following equivariance property
\begin{equation}\label{Eq:EquiM}
\M_n(\omega\mu\,;\lambda,x) = \omega^{r_n} \M_n(\mu\,;\lambda,x).
\end{equation}
So $\M^0_n(\lambda,x)$ is the coefficient of the first term $\mu^{r_n}$ in the power series expansion of $\M_n(\mu\,;\lambda,x)$.

\subsection{Involution with non-local charges}
\label{Sec:NonLocal}

In this subsection, we use the result of the previous one to prove that the local charges $\Q^{\lambda_0}_n$ are in involution with the non-local charges extracted from the monodromy of the Lax matrix $\Lc(\lambda,x)$. This monodromy is defined as the path-ordered exponential
\begin{equation*}
T(\lambda) = \Pexp \left( -\int \dd z \, \Lc(\lambda,z) \right),
\end{equation*}
where the integral is taken on the real line $\R$ or the circle $S^1$, depending on the coordinate space of the model. Consider also the partial transfer matrices
\begin{equation*}
T(\lambda\;;x,y) = \Pexp \left( -\int_y^x \dd z \, \Lc(\lambda,z) \right).
\end{equation*}
These matrices satisfy the initial condition $T(\lambda\;;x,x)=\Id$ and the differential equations
\begin{subequations}\label{Eq:DerPexp}
\begin{align}
\p_x T(\lambda\;;x,y) &= -\Lc(\lambda,x) T(\lambda\;;x,y), \\
\p_y T(\lambda\;;x,y) &= T(\lambda\;;x,y) \Lc(\lambda,y).
\end{align}
\end{subequations}
Moreover, the variation of $T$ under a infinitesimal variation $\delta\Lc$ of $\Lc$ is given by
\begin{equation*}
\delta T(\lambda\;;x,y) = -\int_y^x \dd z \, T(\lambda\;;x,z) \delta\Lc(\lambda,z)T(\lambda\;;z,y).
\end{equation*}

This formula allows one to compute derivatives of $T$ and in particular its Poisson bracket with the local charge $\Q^{\lambda_0}_n$, for $\lambda_0\in\Zc$ and $n\in\E_{\lambda_0}$. Specifically, we have
\begin{equation}\label{Eq:PBPexp}
\left\lbrace \Q^{\lambda_0}_n, T(\lambda\;;x,y) \right\rbrace
 = -\int_y^x \dd z \, T(\lambda\;;x,z) \left\lbrace \Q^{\lambda_0}_n, \Lc(\lambda,z) \right\rbrace T(\lambda\;;z,y).
\end{equation}
The Poisson bracket of $\Q^{\lambda_0}_n$ and $\Lc(\lambda,z)$ is given by Theorem \ref{Thm:ZCEL}. Using this together with the equations \eqref{Eq:DerPexp} and \eqref{Eq:PBPexp}, we find
\begin{equation*}
\left\lbrace \Q^{\lambda_0}_n, T(\lambda\;;x,y) \right\rbrace = T(\lambda\;;x,y) \M^{\lambda_0}_n(\lambda,y) - \M^{\lambda_0}_n(\lambda,x) T(\lambda\;;x,y).
\end{equation*}
If the spatial coordinate is taken on the real line (from $-\infty$ to $\infty$) and the fields are assumed to be decreasing at infinity fast enough, we get
\begin{equation*}
\left\lbrace \Q^{\lambda_0}_n, T(\lambda) \right\rbrace = 0,
\end{equation*}
\textit{i.e.} the whole monodromy $T(\lambda)$ is in involution with $\Q^{\lambda_0}_n$. If the spatial coordinate is taken on the circle (from $0$ to $2\pi$) and the fields are assumed to be periodic, we get
\begin{equation*}
\left\lbrace \Q^{\lambda_0}_n, T(\lambda) \right\rbrace = \left[ T(\lambda), \M^{\lambda_0}_n(\lambda,0) \right].
\end{equation*}
In this case, $\Q^{\lambda_0}_n$ Poisson commutes with any central function of $T(\lambda)$, \textit{e.g.} the traces $\Tr\bigl(T(\lambda)^k\bigr)$ and the determinant $\det\bigl(T(\lambda)\bigr)$. Thus, we have

\begin{theorem}\label{Thm:NonLocal}
The monodromy $T(\lambda)$ (resp. the central functions of $T(\lambda)$) is in involution with the local charges $\Q^{\lambda_0}_n$ for $\lambda_0\in\Zc$ and $n\in\E_{\lambda_0}$, if the spatial coordinate is taken on the real line (resp. the circle). In particular, it is conserved.
\end{theorem}
\begin{proof}
It just remains to prove the conservation of the non-local charges. This follows from the fact that the Hamiltonian $\Hc$ can be expressed as a linear combination of the quadratic charges $\Q^{\lambda_0}_2$ and the momentum $\Pc$.
\end{proof}

Once again, this theorem applies only for finite regular zeros $\lambda_0$. Following a similar argument to the one given in the previous subsection, it also holds for the charges $\Q^\infty_n$ if we consider Poisson brackets only weakly.

\subsection{Zero curvature equations between the $\M^{\lambda_0}_n$'s} \label{sec: ZC eq}

In subsection \ref{Sec:ZCEL}, we showed that the dynamics of the Lax matrix $\Lc(\lambda,x)$ under the Hamiltonian flow of the local charge $\Q^{\lambda_0}_n$ takes the form of a zero curvature equation with a matrix $\M^{\lambda_0}_n(\lambda,x)$. We thus exhibited a hierarchy of integrable equations, corresponding to the commutativity of the connections $\nabla^{\lambda_0}_n$ with $\nabla_x$. This can be seen as the compatibility condition of the two auxiliary linear problems $\nabla_x\Psi=0$ and $\nabla^{\lambda_0}_n \Psi=0$, with $\Psi$ a function on the phase space, valued in the connected and simply connected Lie group with Lie algebra $\g$. In this subsection, we prove that the connections $\nabla^{\lambda_0}_n$ and $\nabla^{\mu_0}_m$ also commute with one another (except when $\lambda_0$ is finite and $\mu_0=\infty$). This can be seen as the simultaneous compatibility of all auxiliary linear problems $\nabla^{\lambda_0}_n \Psi=0$ and it takes the form of zero curvature equations:

\begin{theorem}\label{Thm:ZCEM}
Let $\lambda_0,\mu_0\in\Zc$, $n\in\E_{\lambda_0}$ and $m\in\E_{\mu_0}$. We have the zero curvature equation
\begin{equation*}
\left\lbrace \Q^{\lambda_0}_n, \M^{\mu_0}_m(\lambda,x) \right\rbrace - \left\lbrace \Q^{\mu_0}_m , \M^{\lambda_0}_n(\lambda,x) \right\rbrace  + \left[ \M^{\lambda_0}_n(\lambda,x), \M^{\mu_0}_m(\lambda,x) \right] = 0.
\end{equation*}
\end{theorem}

This subsection is entirely devoted to the proof of Theorem \ref{Thm:ZCEM}. After stating some general results, we will treat separately the cases $\lambda_0\neq\mu_0$ and $\lambda_0=\mu_0$. Note that for the latter, we only have a complete proof for an algebra $\g$ of type B, C and D. For $\g$ of type A, we verified Theorem \ref{Thm:ZCEM} for the first degrees $n$ and $m$ and conjecture that it holds more generally for any $n$ and $m$. To improve the clarity of the subsection, some technical details of the proof are presented in appendix \ref{App:Xi}.

Here also the theorem concerns the finite regular zeros $\lambda_0$ and $\mu_0$. The method presented in this subsection also applies for $\lambda_0=\mu_0=\infty$ as $\Lc^\infty(\lambda,x)$ also satisfies an $r/s$-system with twist function (Theorem \ref{Thm:PBLcI}). However, the theorem does not hold when $\lambda_0$ is finite and $\mu_0=\infty$, even if Poisson brackets are considered only weakly.

\subsubsection{Some general results}
\label{Sec:ZCEMGeneral}

Let us consider the Poisson bracket \eqref{Eq:PBLT}. It can be rewritten
\begin{align*}
& \hspace{-3pt}\left\lbrace S(\lambda,x), \Tc_m(\mu,y) \right\rbrace = - m \, \Tr\ti{2} \Bigl( U\ti{12}(\lambda,\mu) S_{m-1}(\mu,y)\ti{2} \Bigr) \delta'_{xy} \notag \\
& \hspace{140pt} + m \left[ \Tr\ti{2}\Bigl( \Rc^0\ti{12}(\lambda,\mu) S_{m-1}(\mu,x)\ti{2} \Bigr), S(\lambda,x) \right] \delta_{xy},
\end{align*}
with $S(\lambda,x)=S_1(\lambda,x)=\varphi(\lambda)\Lc(\lambda,x)$. Starting from this Poisson bracket, we elevate $S(\lambda,x)$ to the power $n-1$ and find, using the fact that the Poisson bracket and the commutator are derivations, that
\begin{align}\label{Eq:PBSnT}
&\left\lbrace S_{n-1}(\lambda,x), \Tc_m(\mu,y) \right\rbrace = m \left[ \Tr\ti{2}\Bigl( \Rc^0\ti{12}(\lambda,\mu) S_{m-1}(\mu,x)\ti{2} \Bigr), S_{n-1}(\lambda,x) \right] \delta_{xy} \\
&\hspace{150pt} - m \sum_{k=0}^{n-2} S_k(\lambda,x) \Tr\ti{2} \Bigl( U\ti{12}(\lambda,\mu) S_{m-1}(\mu,y)\ti{2} \Bigr)  S_{n-2-k}(\lambda,x) \delta'_{xy}. \notag
\end{align}
Recall the definition \eqref{Eq:DefN} of $\Nc_n(\lambda\, ; \rho,x)$. From the Poisson bracket \eqref{Eq:PBSnT}, using the cyclicity of the trace, we find
\begin{equation*}
\left\lbrace \Nc_n(\lambda\, ; \rho,x), \Tc_m(\mu,y) \right\rbrace = \Gamma^{\lambda\mu}_{nm}(\rho,x) \delta_{xy} + \Xi^{\lambda\mu}_{nm}(\rho,x,y),
\end{equation*}
where
\begin{align}
&\Gamma^{\lambda\mu}_{nm}(\rho,x) = nm \Tr\ti{23} \Bigl( \bigl[ \Rc^0\ti{12}(\rho,\lambda), \Rc^0\ti{23}(\lambda,\mu) \bigr]   S_{n-1}(\lambda,x) \ti{2} S_{m-1}(\mu,y)\ti{3} \Bigr), \label{Eq:DefGamma} \\
&\Xi^{\lambda\mu}_{nm}(\rho,x,y) = nm \Tr\ti{23} \Bigl( \Rc^0\ti{12}(\rho,\lambda) S_{m-1}(\mu,y)\ti{3} \sum_{k=0}^{n-2} S_k(\lambda,x)\ti{2} U\ti{23}(\lambda,\mu) S_{n-2-k}(\lambda,x)\ti{2} \Bigr) \delta'_{xy}, \label{Eq:DefXi}
\end{align}
with $U$ defined in equation \eqref{Eq:DefU}. Let us introduce
\begin{align*}
&\Y^{\lambda\mu}_{nm}(\rho,x,y) = \left[ \Nc_n(\lambda\, ; \rho,x), \Nc_m(\mu\,;\rho,x) \right]\delta_{xy} \\
& \hspace{90pt} + \left\lbrace \Tc_n(\lambda,y), \Nc_m(\mu\,;\rho,x) \right\rbrace  - \left\lbrace \Tc_m(\mu,y), \Nc_n(\lambda\, ; \rho,x) \right\rbrace . \notag
\end{align*}
It contains a term equal to $\delta_{xy}$ times
\begin{equation*}
\left[ \Nc_n(\lambda\, ; \rho,x), \Nc_m(\mu\,;\rho,x) \right] + \Gamma^{\lambda\mu}_{nm}(\rho,x) - \Gamma^{\mu\lambda}_{mn}(\rho,x).
\end{equation*}
One can show from equations \eqref{Eq:DefN} and \eqref{Eq:DefGamma} that this is equal to
\begin{equation*}
\Tr\ti{23}\Bigl( \Upsilon\ti{123}(\rho,\lambda,\mu) S_{n-1}(\lambda,x)\ti{2} S_{m-1}(\mu,x)\ti{3} \Bigr),
\end{equation*}
with
\begin{align*}
&\hspace{-3pt}\Upsilon\ti{123}(\rho,\lambda,\mu) = \left[ \Rc^0\ti{12}(\rho,\lambda), \Rc^0\ti{13}(\rho,\mu) \right]   + \left[ \Rc^0\ti{12}(\rho,\lambda), \Rc^0\ti{23}(\lambda,\mu) \right] + \left[ \Rc^0\ti{32}(\mu,\lambda), \Rc^0\ti{13}(\rho,\mu) \right].
\end{align*}
This terms vanishes as $\Rc^0$ is a solution of the classical Yang-Baxter equation \eqref{Eq:CYBE}. We are therefore simply left with
\begin{equation}\label{Eq:Y}
\Y^{\lambda\mu}_{nm}(\rho,x,y) = \Xi^{\lambda\mu}_{nm}(\rho,x,y) - \Xi^{\mu\lambda}_{mn}(\rho,x,y).
\end{equation}

The currents $\J_k$ are extracted from $\Tc_k$. But in general, the charges are constructed from currents $\K_k$ which are extracted from $\W_k$, where the definition of $\W_k$ depends on $\g$ and $\s$ (see  subsections \ref{Sec:SummaryNonCyc} and \ref{Sec:SummaryCyc}). In particular, we have $\K^{\lambda_0}_k(x)=\W_k(\lambda_0,x)$ for a non-cyclotomic regular zero $\lambda_0$. For the origin, which is cyclotomic, $\K^0_k(x)$ is the coefficient of $\lambda^{r_k}$ in $\W_k(\lambda,x)$. Let us define
\begin{align}
&\Zc^{\lambda\mu}_{nm}(\rho,x,y) = \left[ \M_n(\lambda\, ; \rho,x), \M_m(\mu\,;\rho,x) \right]\delta_{xy} \\
& \hspace{90pt} + \left\lbrace \W_n(\lambda,y), \M_m(\mu\,;\rho,x) \right\rbrace  - \left\lbrace \W_m(\mu,y), \M_n(\lambda\, ; \rho,x) \right\rbrace . \notag
\end{align}
Using the expression of $\W_k$ and $\M_k$ in terms of $\Tc_k$ and $\Nc_k$, we see that $\Zc^{\lambda\mu}_{nm}(\rho,x,y)$ contains several types of terms:
\begin{enumerate}
\item commutators $\left[\Nc_k(\lambda\,;\rho,x),\Nc_l(\mu\,;\rho,x)\right]$, multiplied by polynomials in the $\Tc_j$'s,
\item $\Gamma^{\lambda\mu}_{kl}(\rho,x)$ and $\Gamma^{\mu\lambda}_{lk}(\rho,x)$, multiplied by polynomials in the $\Tc_j$'s,
\item $\Xi^{\lambda\mu}_{kl}(\rho,x,y)$ and $\Xi^{\mu\lambda}_{lk}(\rho,x,y)$, multiplied by polynomials in the $\Tc_j$'s,
\item $\left\lbrace \Tc_k(\lambda,x), \Tc_l(\mu,y) \right\rbrace$ and $\left\lbrace \Tc_k(\mu,x), \Tc_l(\lambda,y) \right\rbrace$, multiplied by polynomials in $\Tc_j$'s and $\Nc_j$'s.
\end{enumerate}
Moreover, the terms of type 1 and 2 are always ultralocal, \textit{i.e.} proportional to $\delta_{xy}$. It can be seen that these terms always combine into polynomials of $\Tc_j$ multiplied by
\begin{equation*}
\left[ \Nc_k(\lambda\, ; \rho,x), \Nc_l(\mu\,;\rho,x) \right] + \Gamma^{\lambda\mu}_{kl}(\rho,x) - \Gamma^{\mu\lambda}_{lk}(\rho,x).
\end{equation*}
As explained above, this vanishes by virtue of the classical Yang-Baxter equation. Therefore, $\Zc^{\lambda\mu}_{nm}(\rho,x,y)$ is composed only of terms of type 3 and 4.

\subsubsection{Zero curvature equation at different regular zeros}

Let us now prove Theorem \ref{Thm:ZCEM} when $\lambda_0$ and $\mu_0$ are different regular zeros. Since we are not considering here the point at infinity (see discussion after Theorem \ref{Thm:ZCEM}), at least one of them is non-cyclotomic, say $\mu_0$. Recall that $U\ti{23}(\lambda,\mu_0)=\varphi(\lambda)\Rc^0(\lambda,\mu_0)$, as $\varphi(\mu_0)=0$.\\

Consider first the case where $\lambda_0$ is also non-cyclotomic. We will prove the that zero curvature equation of Theorem \ref{Thm:ZCEM} holds by showing that $\Zc^{\lambda_0\mu_0}_{nm}(\rho,x,y)$ vanishes. As explained above, it contains two types of terms. The ones of types 4 contain Poisson brackets between currents $\J^{\lambda_0}_k$ and $\J^{\mu_0}_l$. According to Theorem \ref{Thm:DiffZeros}, these brackets are all zeros. As $\lambda_0$ and $\mu_0$ are two distinct elements of $\Zc$, the cyclotomic orbits $\Z_T\lambda_0$ and $\Z_T\mu_0$ are disjoint and thus $\Rc^0(\lambda,\mu_0)$ is regular at $\lambda=\lambda_0$. We then have $U\ti{23}(\lambda_0,\mu_0)=0$, as $\varphi(\lambda_0)=0$. We deduce from this that $\Xi^{\lambda_0\mu_0}_{kl}(\rho,x,y) = 0$ and similarly $\Xi^{\mu_0\lambda_0}_{lk}(\rho,x,y) = 0$, \textit{i.e.} the terms of type 3 also vanish. Thus $\Zc^{\lambda_0\mu_0}_{nm}(\rho,x,y)=0$, as required. \\

Suppose now that $\lambda_0$ is the origin and hence a cyclotomic point. Recall that $\K^0_n(x)$ and $\M^0_n(\rho,x)$ are the coefficients of $\lambda^{r_n}$ in respectively $\W_n(\lambda,x)$ and $\M_n(\lambda\,;\rho,x)$. Thus, it is enough to show that there is no term $\lambda^{r_n}$ in $\Zc^{\lambda\mu_0}_{nm}(\rho,x,y)$ to prove Theorem \ref{Thm:ZCEM} in this case. Recall that $\Tc_k(\lambda,x)$ and $\Nc_k(\lambda\,;\rho,x)$ contain powers of $\lambda$ of the form $r_k + aT$ with $a\in\Z_{\geq 0}$. As $r_n \leq T-2 <T$ for $n\in\E_{\lambda_0}$, the powers with $a\geq 1$ cannot contribute to the $\lambda^{r_n}$-term. Following the discussion at the end of paragraph \ref{Sec:ZCEMGeneral}, the term $\lambda^{r_n}$ of $\Zc^{\lambda\mu_0}_{nm}(\rho,x,y)$ is thus composed of polynomials in the $\J^0_k$'s and $\Nc^0_k$'s times Poisson brackets of $\J^0_k$ with $\J^{\mu_0}_l$ or terms of the form
\begin{equation*}
\Xi^{\lambda\mu_0}_{kl}(\rho,x,y)\Bigr|_{\lambda^{r_k}} \;\;\; \text{ or } \;\;\; \Xi^{\mu_0\lambda}_{lk}(\rho,x,y)\Bigr|_{\lambda^{r_k}},
\end{equation*}
for $k$ such that $r_k < T-1$. According to Theorem \ref{Thm:DiffZeros}, the Poisson brackets of such $\J^0_k$ with $\J^{\mu_0}_l$ vanish. Moreover, $\Xi^{\lambda\mu_0}_{kl}(\rho,x,y)$ is proportional to $\varphi(\lambda)\Rc^0(\lambda,\mu_0)$. Yet, $\varphi(\lambda)=O(\lambda^{T-1})$ and $r_k<T-1$, hence $\Xi^{\lambda\mu_0}_{kl}(\rho,x,y)\Bigr|_{\lambda^{r_k}}=0$. Similarly $\Xi^{\mu_0\lambda}_{lk}(\rho,x,y)\Bigr|_{\lambda^{r_k}}=0$. Thus, the coefficient of $\lambda^{r_n}$ in $\Zc^{\lambda\mu_0}_{nm}(\rho,x,y)$ vanishes, as required. This ends the proof of Theorem \ref{Thm:ZCEM} for different regular zeros $\lambda_0$ and $\mu_0$.

\subsubsection{Zero curvature equations at a non-cyclotomic regular zero}

Let us now prove Theorem \ref{Thm:ZCEM} for $\lambda_0=\mu_0$. We start with the case where $\lambda_0$ is a non-cyclotomic point. We then want to show that $\displaystyle\int \dd y \,\Zc^{\lambda_0\lambda_0}_{nm}(\rho,x,y)=0$.\\

As in section \ref{Sec:NonCycZero}, we treat separately the Lie algebras of type B, C and D and the Lie algebras of type A. Suppose first that $\g$ is of type B, C or D. In this case, the currents $\K^{\lambda_0}_{2n}$ are equal to the currents $\J^{\lambda_0}_{2n}$ (see subsections \ref{Sec:NonCycZeroBCD} and \ref{Sec:SummaryNonCyc}) and the corresponding Lax matrices $\M^{\lambda_0}_{2n}$ are equal to the matrices $\Nc^{\lambda_0}_{2n}$. Thus, $\Zc^{\lambda_0\lambda_0}_{2n\,2m}(\rho,x,y)$ is simply equal to $\Y^{\lambda_0\lambda_0}_{2n\,2m}(\rho,x,y)$ (see paragraph \ref{Sec:ZCEMGeneral}). According to equation \eqref{Eq:Y}, we have
\begin{equation*}
\Y^{\lambda_0\lambda_0}_{2n\,2m}(\rho,x,y) = \Xi^{\lambda_0\lambda_0}_{2n\,2m}(\rho,x,y) - \Xi^{\lambda_0\lambda_0}_{2m\,2n}(\rho,x,y),
\end{equation*}
where $\Xi$ was defined in equation \eqref{Eq:DefXi}. To avoid cluttering the argument in the present paragraph with too many technicalities, we postpone the details of the computation of $\Xi^{\lambda_0\lambda_0}_{2n\,2m}(\rho,x,y)$ in appendix \ref{App:XiNonCyc}. We find
\begin{equation*}
\Xi^{\lambda_0\lambda_0}_{2n\,2m}(\rho,x,y) = \frac{\varphi'(\lambda_0)}{T} \frac{4nm(1-2n)}{2n+2m-2}\Nc^{\lambda_0}_{2n+2m-2}(\rho,x) \delta'_{xy} + f^{\lambda_0}_{2n\,2m}(\rho,x) \delta_{xy},
\end{equation*}
where the function $f^{\lambda_0}_{2n\,2m}$ satisfies $f^{\lambda_0}_{2n\,2m}=f^{\lambda_0}_{2m\,2n}$ (cf. appendix \ref{App:XiNonCyc}). It then follows that
\begin{equation*}
\Y^{\lambda_0\lambda_0}_{2n\,2m}(\rho,x,y) = \frac{\varphi'(\lambda_0)}{T} \frac{8nm(m-n)}{2n+2m-2}\Nc^{\lambda_0}_{2n+2m-2}(\rho,x) \delta'_{xy},
\end{equation*}
from which we deduce that $\displaystyle \int \dd y \; \Y^{\lambda_0\lambda_0}_{2n\,2m}(\rho,x,y) = 0$, as required.\\

Suppose now that $\g$ is of type A. In this case, the currents $\K^{\lambda_0}_n$ are different from the currents $\J^{\lambda_0}_n$ and we therefore have to consider $\Zc^{\lambda_0\lambda_0}_{nm}$ rather than simply $\Y^{\lambda_0\lambda_0}_{nm}$. According to the discussion at the end of paragraph \ref{Sec:ZCEMGeneral}, it contains polynomials in the $\J^{\lambda_0}_p$'s and $\Nc^{\lambda_0}_p$'s, multiplied by either $\Xi^{\lambda_0\lambda_0}_{kl}(\rho,x,y)$ or $\bigl\lbrace \J^{\lambda_0}_k(x), \J^{\lambda_0}_l(y) \bigr\rbrace$. This last Poisson bracket is given by equation \eqref{Eq:PBJTypeA} and is expressed in terms of the $\J^{\lambda_0}_p$'s. As for type B, C and D, we compute the expression of $\Xi^{\lambda_0\lambda_0}_{kl}$ in appendix \ref{App:XiNonCyc}. We find
\begin{align*}
&\Xi^{\lambda_0\lambda_0}_{k\,l}(\rho,x,y) = - \frac{\varphi'(\lambda_0)}{T} \frac{kl(k-1)}{k+l-2} \Nc^{\lambda_0}_{k+l-2}(\rho,x) \delta'_{xy} \\
& \hspace{145pt}+ \frac{\varphi'(\lambda_0)}{d T} kl \J^{\lambda_0}_{l-1}(y) \Nc^{\lambda_0}_{k-1}(\rho,x) \delta'_{xy} + f^{\lambda_0}_{kl}(\rho,x) \delta_{xy},
\end{align*}
for some function $f^{\lambda_0}_{kl}$ such that $f^{\lambda_0}_{kl}=f^{\lambda_0}_{lk}$.

Hence $\Zc^{\lambda_0\lambda_0}_{nm}$ can be expressed in terms of the $\J^{\lambda_0}_p$'s and $\Nc^{\lambda_0}_p$'s, up to terms involving $f^{\lambda_0}_{kl}$. The latter are always of the form
\begin{equation*}
\alpha \J^{\lambda_0}_{p_1}(x) \ldots \J^{\lambda_0}_{p_q}(x) f^{\lambda_0}_{kl}(\rho,x) \delta_{xy},
\end{equation*}
with $\alpha$ a constant. Moreover, one can check that for any such term, there is also a similar one but with an opposite sign and $k$ and $l$ interchanged. Using the symmetry property $f^{\lambda_0}_{kl}=f^{\lambda_0}_{lk}$, one can then conclude that these terms always vanish.

Therefore $\Zc^{\lambda_0\lambda_0}_{nm}$ can be expressed in terms of the $\J^{\lambda_0}_p$'s and $\Nc^{\lambda_0}_p$'s. Using the first explicit expressions \eqref{Eq:KJ} for the current $\K^{\lambda_0}_n$ and the corresponding expressions for the matrices $\M^{\lambda_0}_n$, it can be check directly that $\displaystyle \int \dd y \; \Zc^{\lambda_0\lambda_0}_{nm}(\rho,x,y)=0$ for the first few degrees $n, m$. Specifically, we have checked this for degrees $n$ and $m$ up to 7. In particular, we observed that we could not have chosen different coefficients in equation \eqref{Eq:KJ} for these zero curvature equations to hold (in the same way that these coefficients were uniquely fixed by requiring the involution of $\Q^{\lambda_0}_n$ and $\Q^{\lambda_0}_m$). Based on these strong observations, we conjecture that it holds for any $n,m\in\E_{\lambda_0}$.

\subsubsection{Zero curvature equations at a cyclotomic regular zero}

Finally, let us prove Theorem \ref{Thm:ZCEM} for $\lambda_0=\mu_0=0$, which is a cyclotomic point. Remember that $\K^0_n(x)$ and $\M^0_n(\rho,x)$ are extracted as the coefficient of $\lambda^{r_n}$ in the power series expansion of $\W_n(\lambda,x)$ and $\M_n(\lambda\,;\rho,x)$ where $r_n$ is the smallest power appearing in these expansions. That is, Theorem \ref{Thm:ZCEM} for $\lambda_0=\mu_0=0$ is equivalent to the statement that
\begin{equation*}
\int \dd y \; \Zc^{\lambda\lambda}_{nm} (\rho,x,y) \Bigr|_{\lambda^{r_n+r_m}} = 0.
\end{equation*}

Let us start with the case of a Lie algebra $\g$ of type B, C or D, for which $\Zc^{\lambda\lambda}_{2n\,2m}=\Y^{\lambda\lambda}_{2n\,2m}$. According to equation \eqref{Eq:Y}, we have
\begin{equation*}
\Y^{\lambda\lambda}_{2n\,2m}(\rho,x,y) = \Xi^{\lambda\lambda}_{2n\,2m}(\rho,x,y) - \Xi^{\lambda\lambda}_{2m\,2n}(\rho,x,y).
\end{equation*}
The computation of $\Xi^{\lambda\lambda}_{2n\,2m}\bigr|_{\lambda^{r_n+r_m}}$ is performed in appendix \ref{App:XiCyc}. The final result is
\begin{align*}
&\Xi^{\lambda\lambda}_{2n\,2m}(\rho,x,y)\Bigr|_{\lambda^{r_{2n}+r_{2m}}} = f^{(0)}_{2n\,2m}(\rho,x) \delta_{xy} - \theta_{r_{2n}+r_{2m}+2-T}\, \zeta'(0) \frac{4nm(2n-1)}{2n+2m-2}\Nc^0_{2n+2m-2}(\rho,x) \delta'_{xy}, \notag
\end{align*}
with $f^{(0)}_{2n\,2m}$ a function symmetric under the exchange of $n$ and $m$. By virtue of this symmetry we find that the terms involving $f$ disappear in $\Y^{\lambda\lambda}_{2n\,2m}(\rho,x,y)\bigr|_{\lambda^{r_n+r_m}}$, while the other terms vanish when integrated over $y$, as required.\\

Consider now $\g=\sl(d,\C)$ of type A. The construction of the currents 
$\K^0_k$ depends on $\s$ being inner or not 
(see subsections \ref{Sec:TypeATrivial}, \ref{Sec:TypeANonTrivial} and \ref{Sec:SummaryCyc}). If 
 $\s$ is inner, then the currents $\K^0_k$ and $\W_k$ are equal to the currents $\J^0_k$ and $\Tc_k$. In this case, we have
\begin{equation*}
\Zc^{\lambda\lambda}_{nm}(\rho,x,y)=\Y^{\lambda\lambda}_{nm}(\rho,x,y) = \Xi^{\lambda\lambda}_{nm}(\rho,x,y) - \Xi^{\lambda\lambda}_{mn}(\rho,x,y).
\end{equation*}
The expression for $\Xi^{\lambda\lambda}_{nm}(\rho,x,y)\bigr|_{\lambda^{r_n+r_m}}$ is given by equation \eqref{Eq:XiATrivial} of appendix \ref{App:XiCyc}. It has the same structure as in the case of types B, C and D: the same arguments then apply and we conclude that the integration of $\Y^{\lambda\lambda}_{nm}(\rho,x,y)\bigr|_{\lambda^{r_n+r_m}}$ over $y$ vanishes.\\

Finally, consider $\g=\sl(d,\C)$ of type A with $\s$ not inner. In this case, the currents $\K^0_k(x)$ and $\W_k(\lambda,x)$ are constructed as polynomials of respectively $\J^0_k(x)$ and $\Tc_k(\lambda,x)$. The corresponding structure of $\Zc^{\lambda\mu}_{nm}$ is discussed at the end of paragraph \ref{Sec:ZCEMGeneral}. In particular, $\Zc^{\lambda\lambda}_{nm}(\rho,x,y)$ is composed of two types of terms:
\begin{itemize}
\item $\Xi^{\lambda\lambda}_{kl}(\rho,x,y)$, multiplied by polynomials in the $\Tc_j(\lambda,\bm{\cdot})$'s,
\item $\bigl\lbrace \Tc_k(\lambda,x), \Tc_l(\lambda,y) \bigr\rbrace$, multiplied by polynomials in the $\Tc_j(\lambda,\bm{\cdot})$'s and $\Nc_j(\lambda \, ; \rho,\bm{\cdot})$'s.
\end{itemize}

We want to extract the coefficient of $\lambda^{r_n+r_m}$ in $\Zc^{\lambda\lambda}_{nm}(\rho,x,y)$. Recall that the powers of $\lambda$ appearing in $\Tc_j(\lambda,\bm{\cdot})$ and $\Nc_j(\lambda \, ; \rho,\bm{\cdot})$ are of the form $r_j+aT$, with $a\in\Z_{\geq 0}$, and that the coefficients corresponding to $a=0$ are $\J^0_j(\bm{\cdot})$ and $\Nc^0_j(\rho,\bm{\cdot})$. In the two types of terms mentioned above, one can check that the terms with $a>0$ will not contribute to the coefficient of $\lambda^{r_n+r_m}$. More precisely, $\Zc^{\lambda\lambda}_{nm}(\rho,x,y)\bigr|_{\lambda^{r_n+r_m}}$ is composed of two types of terms:
\begin{itemize}
\item $\Xi^{\lambda\lambda}_{kl}(\rho,x,y)\bigr|_{\lambda^{r_k+r_l}}$, multiplied by polynomials in the $\J^0_j(\bm{\cdot})$'s,
\item $\bigl\lbrace \J^0_k(x), \J^0_l(y) \bigr\rbrace$, multiplied by polynomials in the $\J^._j(\bm{\cdot})$'s and $\Nc^0_j(\rho,\bm{\cdot})$'s.
\end{itemize}
The Poisson brackets $\bigl\lbrace \J^0_k(x), \J^0_l(y) \bigr\rbrace$ are given by equation \eqref{Eq:PBJCycA}. The expression for $\Xi^{\lambda\lambda}_{kl}(\rho,x,y)\bigr|_{\lambda^{r_k+r_l}}$ is worked out in appendix \ref{App:XiCyc} and reads
\begin{align*}
\Xi^{\lambda\lambda}_{kl}(\rho,x,y)\Bigr|_{\lambda^{r_{k}+r_{l}}} & = f^{(0)}_{kl}(\rho,x) \delta_{xy} - \theta_{r_{k}+r_{l}+2-T}\, \zeta'(0) \frac{kl(k-1)}{k+l-2}\Nc^0_{k+l-2}(\rho,x) \delta'_{xy} \\
& \hspace{65pt} - \theta_{r_{k}+1-S}\theta_{r_{l}+1-S}\, \frac{\zeta'(0)}{d}\, kl \J^0_{l-1}(y) \Nc^0_{k-1}(\rho,x) \delta'_{xy}, \notag
\end{align*}
where $f^{(0)}_{kl}$ is a function invariant under the interchange of $k$ and $l$.

The rest of the argument follows closely that given in the non-cyclotomic case. Specifically, the terms containing $f^{(0)}_{kl}$ are seen to vanish by virtue of this symmetry property. We can thus express $\Zc^{\lambda\lambda}_{nm}(\rho,x,y)\bigr|_{\lambda^{r_n+r_m}}$ in terms of the $\J^0_k$'s and $\Nc^0_k$'s only. One then can check explicitly that this expression vanishes when integrated over $y$, as required. We verified this for the first few degrees $n$ and $m$ (up to 8) and different values of $T$ (from 2 to 6). We therefore conjecture that this is also true for any $n,m \in \E_0$ and any $T$.

\section{Applications}
\label{Sec:Applications}

In paragraph \ref{Sec:Examples}, we gave a list of integrable $\s$-models which fit the framework of the present article. In this section, we apply the methods developed in the previous sections to these particular examples, analyse the results and compare them to some existing work in the literature. These models were recently re-interpreted as particular examples of so-called dihedral affine Gaudin models~\cite{Vicedo:2017cge}. We explain in the last part of this section how the framework of dihedral affine Gaudin models is particularly suited to apply the methods of the present article.

\subsection{Principal chiral model and its deformations}
\label{Sec:PCM}

Let us start with the simplest integrable $\s$-model, the Principal Chiral Model (PCM). The study of local charges of the PCM is already well known and was treated in the reference~\cite{Evans:1999mj}: these results were the principal motivation and guideline for the present article. In particular, one of the aims was to generalise the construction of~\cite{Evans:1999mj} to a wider class of models, among which are the integrable two-parameters deformations of the PCM (dPCM). We shall discuss the latter in this subsection.\\

The integrable structure of the dPCM was discussed in subsection \ref{Sec:Examples}. 
In this case, $\s=\Id$ so that  $T=1$. Their Lax matrix and twist function are given by equations \eqref{Eq:LaxPCM} and \eqref{Eq:TwistdPCM} respectively. In the language of this article, the regular zeros of these deformed models are $+1$ and $-1$. The evaluation of $\varphi(\lambda)\Lc(\lambda,x)$ at these zeros gives the fields $J_\pm(x)$ of equation \eqref{Eq:ChiralFieldsdPCM}.

The local charges $\Q^{\pm 1}_n$ constructed in the present article are related to the traces of powers of these fields. In the undeformed case, \textit{i.e.} for the PCM, these fields coincide with the fields $j_\pm = -g^{-1} \p_\pm g$ used in reference~\cite{Evans:1999mj} to construct the local charges. Thus, the method presented in this article gives back the results of~\cite{Evans:1999mj} for the PCM, as expected. In the deformed case, it generalises these results, while keeping a similar structure in the construction: in particular, we obtain two towers of local charges in involution, corresponding to the two chiralities of the model, and the spin of these charges is still related to the exponents of the affine Kac-Moody algebra $\widehat{\g}$, as it was in the PCM case~\cite{Evans:1999mj}.\\

The Hamiltonian and momentum of the deformed PCM are given by equations \eqref{Eq:HamMomPCM}. One can check that these are related to the quadratic charges $\Q^{\pm 1}_2$ as follows
\begin{subequations}
\begin{align*}
\Hc_{\text{dPCM}} & = - \frac{\Q^{+1}_2}{2\varphi'_{\text{dPCM}}(+1)} + \frac{\Q^{-1}_2}{2\varphi'_{\text{dPCM}}(-1)} , \\
\Pc_{\text{dPCM}} & = - \frac{\Q^{+1}_2}{2\varphi'_{\text{dPCM}}(+1)} -\frac{\Q^{-1}_2}{2\varphi'_{\text{dPCM}}(-1)} .
\end{align*}
\end{subequations}
In particular, the Hamiltonian belongs to the algebra of local charges in involution so that these charges are conserved (see also the discussion at the end of subsection \ref{Sec:AlgebraLoc}).

This observation also allows one to recover the temporal component $\M(\lambda,x)$ of the Lax pair of the model (see the paragraph below equation \eqref{Eq:Nabla}). More precisely, the equation of motion of the dPCM can be recast as the Lax equation \eqref{Eq:ZCEH}, where
\begin{equation*}
\M_{\text{dPCM}}(\lambda_,x)= - \frac{\M^{+1}_2(\lambda,x)}{2\varphi'_{\text{dPCM}}(+1)} + \frac{\M^{-1}_2(\lambda,x)}{2\varphi'_{\text{dPCM}}(-1)} = \frac{j_0(x) + \lambda j_1(x)}{1-\lambda^2}.
\end{equation*}
This zero curvature equation \eqref{Eq:ZCEH} is the first among a whole hierarchy of integrable equations generated by the local charges $\Q^{\pm 1}_n$ (cf. subsection \ref{Sec:ZCEL}).\\

In particular, this result was used in subsection \ref{Sec:NonLocal} to show that the local charges $\Q^{\pm 1}_n$ are in involution with the non-local charges extracted from the monodromy of the Lax matrix $\Lc(\lambda,x)$ (see Theorem \ref{Thm:NonLocal}). In~\cite{Evans:1999mj}, it was shown that the local charges of the undeformed PCM Poisson commute with the non-local charges generating the classical Yangian symmetry of the model. In the framework of this article, if we consider the model on the real line $\R$, we expect these non-local charges to be extracted from the expansion of (a gauge transformation of) the monodromy around the pole $\lambda=0$ of the twist function of the PCM.

For Yang-Baxter deformations ($k=0$ and $\eta\neq 0$, see paragraph \ref{Sec:Examples}), this Yangian symmetry gets deformed to a quantum affine symmetry~\cite{Delduc:2016ihq,Delduc:2017brb}. In particular, studying the monodromy around the poles $\pm i \eta$ of the twist function of the Yang-Baxter model, one can extract a $q$-deformed affine Poisson-Hopf algebra $\mathscr U_q(\widehat{\g})$. We have therefore proved that this algebra of non-local charges is in involution with the algebra of local charges consisting of the $\Q^{\pm 1}_n$'s.\\

As mentioned in the paragraph \ref{Sec:Examples}, the PCM and its deformation are defined on a real Lie group $G_0$, whose Lie algebra $\g_0$ is a real form of $\g$. This real form is characterised by a semi-linear involutive automorphism $\tau$. The Lax matrix \eqref{Eq:LaxPCM} of these models satisfies the reality condition \eqref{Eq:Reality}. Moreover, the twist function \eqref{Eq:TwistdPCM} verifies the reality condition \eqref{Eq:TwistReal} and the regular zeros of the model ($+1$ and $-1$) are real. Thus, the discussion of the subsection \ref{Sec:Reality} applies to these models and the charges $\Q^{\pm 1}_n$ are real (possibly up to a redefinition of some $\Q^{\pm 1}_n$ by a factor of $i$, depending on $\tau$).

\subsection[Bi-Yang-Baxter $\s$-model]{Bi-Yang-Baxter $\bm{\s}$-model}
\label{Sec:BYB}

There exists another two-parameter deformation of the PCM, the so-called bi-Yang-Baxter (bYB) $\s$-model \cite{Klimcik:2008eq,Klimcik:2014bta}. Its Hamiltonian integrability was established in~\cite{Delduc:2015xdm}, by viewing it as a deformation of the $\Z_2$-coset model $G_0\times G_0 / G_{0,\text{diag}}$. In this formulation, the bYB model falls into the framework of the present article but has to be considered as a model with a gauge constraint (this is also the approach of the reference~\cite{Vicedo:2017cge}). This can be treated with the methods developed here, using the results of subsections \ref{Sec:Infinity} and \ref{Sec:Gauge} to take into account the gauge constraint. The analysis of the local charges for this model would then be close to the one for the $\Z_2$-coset model that we perform in the next subsection.

In this subsection, we choose to treat the bYB model in its gauge fixed formulation, which is more natural if we see it as a deformation of the PCM. As we will see, this formulation does not fit exactly within the framework of this article, but we will explain how one can overcome this difficulty by further relaxing the general assumptions we made. Even though this generalisation could have been done throughout the entire article, we chose here to work in a more restricting but more common framework for clarity and simplicity. In this regard, the present subsection is also used to illustrate, \textit{a posteriori}, how the methods and results we found apply under the generalised conditions.\\

The Hamiltonian integrability of the bYB model in its gauge fixed formulation was studied in the last section of the article~\cite{Delduc:2015xdm}. In order to see this model as a deformation of the PCM, one should first perform a change of variables from the conventions of~\cite{Delduc:2015xdm}: more precisely, we relate the spectral parameter $z$ of~\cite{Delduc:2015xdm} with the spectral parameter $\lambda$ of this article by the involutive transformation
\begin{equation*}
z = \frac{1-\lambda}{1+\lambda}.
\end{equation*}
In particular, after this change of spectral parameter, the Lax matrix of the gauge fixed-model takes the form
\begin{equation}\label{Eq:LaxBYB}
\Lc_{\text{bYB}}(\lambda,x) = \frac{j_1(x) + \lambda j_0(x)}{1-\lambda^2} + j_\infty(x),
\end{equation}
for some $\g$-valued fields $j_0$, $j_1$ and $j_\infty$. In particular, the field $j_\infty$ vanishes when the deformation parameters $\eta$ and $\tilde{\eta}$ go to zero. We then recover the Lax matrix \eqref{Eq:LaxPCM} of the undeformed PCM. In this sense, the gauge fixed model with spectral parameter $\lambda$ describes an integrable deformation of the PCM.\\

The equation (3.8) of~\cite{Delduc:2015xdm} gives the twist function of the bYB model in terms of the spectral parameter $z$, which we shall denote $\chi_{\text{bYB}}(z)$ here. We define a corresponding function  $\varphi_{\text{bYB}}(\lambda)$ of the spectral parameter $\lambda$ by the relation
\begin{equation*}
\chi_{\text{bYB}}(z) \, \dd z = \varphi_{\text{bYB}}(\lambda) \, \dd\lambda.
\end{equation*}
It is found to be of the form
\begin{equation*}
\varphi_{\text{bYB}}(\lambda) = 8K\frac{1-\lambda^2}{(a\lambda^2+b)(c\lambda^2+d)},
\end{equation*}
where $a$, $b$, $c$ and $d$ depend on the deformation parameters $\eta$ and $\tilde{\eta}$. This is to be compared with the twist function \eqref{Eq:TwistdPCM} (where $k=A=0$) of the PCM. The presence of a global factor in the twist function is directly related to the global factor $K$ in the definition of the action of the model: it can be reabsorbed by setting $K$ to a particular value.

In the undeformed case $\eta=\tilde{\eta}=0$, one has $b=c=0$. Thus, up to the global factor, the function $\varphi_{\text{bYB}}(\lambda)$ coincides with the twist function of the PCM in this case. The bi-Yang-Baxter deformation thus has for effect to deform the poles of the twist function of the PCM: the two double poles at $0$ and $\infty$ in $\varphi_{\text{PCM}}(\lambda) \, \dd\lambda$ split into four simple poles 
on the imaginary axis. The zeros of the twist function stay undeformed in this procedure: indeed, the zeros of $\varphi_{\text{bYB}}$ are $+1$ and $-1$, as for the PCM and its deformations. Considering the expression \eqref{Eq:LaxBYB} of the Lax matrix of the bYB model, we see that these zeros are regular.\\

It was shown in~\cite{Delduc:2015xdm} that the Lax matrix of the gauge fixed model satisfies a non-ultralocal Poisson bracket of the form \eqref{Eq:PBR}. Using the parameters of~\cite{Delduc:2015xdm} and after performing the change of spectral parameter described above, we find that the $\Rc$-matrix describing this Poisson bracket takes the form
\begin{equation*}
\Rc\ti{12}(\lambda,\mu) = \varphi_{\text{bYB}}(\mu)^{-1}\Rc^{\text{bYB}}\ti{12}(\lambda,\mu),
\end{equation*}
where
\begin{equation}\label{Eq:RBYB}
\Rc^{\text{bYB}}\ti{12}(\lambda,\mu) = \Rc^0\ti{12}(\lambda,\mu) - \frac{c\mu}{c\mu^2+d} C\ti{12} - \frac{\tilde{\eta}}{c\mu^2+d} \widetilde{R}\ti{12}.
\end{equation}
In the above equation, $\widetilde{R}$ is the non-split solution of the modified CYBE 
considered in~\cite{Delduc:2015xdm} and we use the standard $\Rc^0$-matrix 
\begin{equation*}
\Rc^0\ti{12}(\lambda,\mu)=\frac{C\ti{12}}{\mu-\lambda}.
\end{equation*}
In the undeformed case, $\tilde{\eta}$ and $c$ vanishes and we recover the integrable structure of the PCM, as described in paragraph \ref{Sec:Examples}. However, in general, the matrix $\Rc^{\text{bYB}}$ is different from $\Rc^0$. Thus, the bYB model does not fit exactly within the framework described in subsection \ref{Sec:Model}.

Let us note that $\Rc^{\text{bYB}}$ is still a solution of the classical Yang-Baxter equation (CYBE) \eqref{Eq:CYBE}, using the fact that $\widetilde{R}$ is a non-split solution of the modified classical Yang-Baxter equation (see~\cite{Delduc:2015xdm}) and the identity
\begin{equation*}
\bigl[ C\ti{ij}, X\ti{i} + X\ti{j} \bigr] = 0,
\end{equation*}
true for any $X\in\g$. Note that the coefficients of $C\ti{12}$ and $\widetilde{R}\ti{12}$ in equation \eqref{Eq:RBYB} are strongly constrained by the requirement that $\Rc^{\text{bYB}}$ fulfils the CYBE.\\

As explained above, $\Rc^{\text{bYB}}$ is different from $\Rc^0$ and therefore we cannot directly apply the results of this article. However, going through the details of 
the proofs of these results for a non-constrained model with $T=1$, we see that the only properties of the matrix $\Rc^0$ that we used are the CYBE (for the zero curvature equations), the fact that $\Rc^0\ti{12}(\lambda,\mu)$ is holomorphic at pairs $(\lambda_0,\mu_0)$ of distinct regular zeros and the asymptotic property \eqref{Eq:RAsymptotic} near a regular zero $\mu = \lambda_0$. The matrix $\Rc^{\text{bYB}}$ also satisfies the CYBE, as explained above. Moreover, one easily checks that it also verifies the holomorphy condition and the asymptotic property mentioned above. Thus, the results we found in this article also apply to the bYB model.

This is a general observation: we can also treat the models where the matrix $\Rc^0$ is replaced by a matrix $\Rc'$ satisfying some similar properties. More precisely, we require that $\Rc'$
\begin{itemize}\setlength\itemsep{0.2em}
\item obeys the CYBE \eqref{Eq:CYBE},
\item is holomorphic at $(\lambda_0,\mu_0)$ with $\lambda_0$ and $\mu_0$ different regular zeros in $\Zc$,
\item verifies the asymptotic property \eqref{Eq:RAsymptotic} around non-cyclotomic regular zeros,
\item satisfies the equation \eqref{Eq:UAround0} for $U\ti{12}(\lambda,\lambda)$ around a cyclotomic regular zero, up to a term $O(\lambda^{2T-3})$ (which would not contribute to some $(r_n+r_m)^{\rm th}$ power of $\lambda$ in \eqref{Eq:PBTCyc}).
\end{itemize}
In particular, let us consider a matrix $\Rc'$ of the form
\begin{equation}\label{Eq:R'}
\Rc'\ti{12}(\lambda,\mu) = \Rc^0\ti{12}(\lambda,\mu) + \mathcal{D}\ti{12}(\mu),
\end{equation}
like the matrix $\Rc^{\text{bYB}}$. Then $\Rc'\ti{12}(\lambda,\mu)$ is holomorphic for $\lambda$ and $\mu$ going to different regular zeros if $\mathcal{D}$ is holomorphic at any regular zero (this is for example the case for the bYB model). This condition also ensures that the asymptotic property \eqref{Eq:RAsymptotic} is satisfied by $\Rc'$. In the same way the condition on $U\ti{12}(\lambda,\lambda)$ is satisfied by $\Rc'$ if $\mathcal{D}\ti{12}(\lambda)+\mathcal{D}\ti{21}(\lambda)=O(\lambda^{T-2})$.

Let us note, however, that these conditions do not allow to treat the case where infinity is a regular zero in the same way that we did in this article (subsections \ref{Sec:Infinity} and \ref{Sec:CycZero}). This would require, among other conditions, that the asymptotic properties \eqref{Eq:RAsymptoticInfinity} at infinity are also satisfied by the matrix $\Rc'$. One can check that a matrix $\Rc'$ of the form \eqref{Eq:R'} can never satisfy the second property of equation \eqref{Eq:RAsymptoticInfinity}.\\

As explained above, we can apply the construction of local charges in involution to the bYB model. These local charges will be very similar to the ones of the PCM and its deformations, described in the previous subsection, so we shall not enter into much details here. Let us note that these charges are related to traces of powers of $j_0(x)\pm j_1(x)$, where $j_0$ and $j_1$ are the fields appearing in the Lax matrix \eqref{Eq:LaxBYB}. As in the case of the PCM (see previous subsection), the Hamiltonian and the momentum of the bYB model are related to the quadratic charges $\Q^{\pm 1}_2$ by the relation
\begin{subequations}
\vspace{-6pt}\begin{align*}
\Hc_{\text{bYB}} & = - \frac{\Q^{+1}_2}{2\varphi'_{\text{bYB}}(+1)} + \frac{\Q^{-1}_2}{2\varphi'_{\text{bYB}}(-1)} , \\
\Pc_{\text{bYB}} & = - \frac{\Q^{+1}_2}{2\varphi'_{\text{bYB}}(+1)} -\frac{\Q^{-1}_2}{2\varphi'_{\text{bYB}}(-1)} .
\end{align*}
\end{subequations}
In particular, the local charges constructed above are all conserved.

\subsection[$\Z_T$-coset models and their deformations]{$\bm{\Z_T}$-coset models and their deformations}

In this subsection, we discuss the construction of local charges in involution 
for $\Z_T$-coset models (and the deformations of $\Z_2$-coset models). These models 
were described in paragraph \ref{Sec:Examples}. The order $T$ of $\s$ 
is strictly greater than one. 
Their twist function and Lax matrix are given by equations \eqref{Eq:TwistZT} and \eqref{Eq:LaxZT}. Such local charges were constructed for symmetric spaces, \textit{i.e.} $\Z_2$-cosets, in references~\cite{Evans:2000qx} and~\cite{Evans:2005zd}: we shall compare these results with the ones of this article.

As already mentioned in the paragraph \ref{Sec:Examples}, the regular zeros of these models are the origin and infinity, which are both cyclotomic points. We shall therefore apply here the construction of section \ref{Sec:CycZero}. Moreover, all these models possess a gauge constraint $\Pi^{(0)}$ (see paragraph \ref{Sec:Examples}), which is identified with the field at infinity $\Cc(x)$ described in subsection \ref{Sec:Infinity}. The results of subsection \ref{Sec:Gauge} ensure that the densities of the local charges that we construct here are gauge invariant. Indeed, by Theorem \ref{Thm:Gauge}, these densities Poisson commute with the constraint $\Cc$.\\

As in the case of the PCM (see subsection \ref{Sec:PCM}), the degrees of the local charges are related to the exponents of the affine Kac-Moody algebra $\widehat{\g}$ plus one (here also, we do not consider the exponents corresponding to the Pfaffian for type D). However, as explained in section \ref{Sec:CycZero}, the fact that the regular zeros of the model are cyclotomic makes some of the exponents `drop out', in the sense that we cannot construct a charge of the corresponding degree. Recall that a degree $n$ (corresponding to an exponent $n-1$) drops out if $r_n$ is equal to $T-1$ (where $r_n$ was defined in subsection \ref{Sec:EquivT}).

Let us study this in more detail for the case of $\Z_2$-cosets. In particular, we shall compare this phenomenon of exponents dropping out with some results of reference~\cite{Evans:2000qx}. Indeed, in this reference, some local charges in involution were constructed for symmetric spaces (\textit{i.e.} $\Z_2$-cosets). These symmetric spaces correspond to quotients $G_0/G^\s_0$ of the real Lie group $G_0$ by the subgroup of fixed points under the involutive automorphism $\s$ (see paragraph \ref{Sec:Examples}). Such spaces were classified, up to isomorphism, for classical compact groups $G_0$.

In particular, the possible exponents (\textit{i.e.} the degrees minus one) of the local charges for each symmetric space of this classification were listed in Table 1 of~\cite{Evans:2000qx}: they form a (potentially proper) subset of the exponents of $\widehat{\g}$. A simple case by case computation of the integers $r_n$ for these symmetric spaces, and thus these automorphisms $\s$, shows that the exponents of $\widehat{\g}$ which do not appear in this list are exactly the exponents that drop out in the formalism of the present article. We therefore recover the structure of the degrees of local charges found in~\cite{Evans:2000qx} (except for the integer $h$ of~\cite{Evans:2000qx}, which we could not interpret in the present formalism).\\

An explicit computation of the traces of powers of $\varphi_{\Z_2}(\lambda)\Lc_{\Z_2}(\lambda,x)$ around the origin $\lambda=0$ shows that the charges constructed in this article coincide, up to some factors, with the ones constructed in reference~\cite{Evans:2000qx}. The two regular zeros $0$ and $\infty$ correspond to the two chiralities of the model. The article~\cite{Evans:2000qx} focused on one particular chirality. Here, we also have the Poisson brackets between the two towers of local charges constructed in this way. Indeed, according to Theorem \ref{Thm:InvolutionInfinity}, we show that these two towers of charges Poisson commute weakly.

This article also generalises the results of~\cite{Evans:2000qx} in different directions. First of all, the present formalism also allows to treat the integrable deformations of the $\Z_2$-coset model. Indeed, as explained in paragraph \ref{Sec:Examples}, the regular zeros of these models are still $0$ and $\infty$ and so the methods developed here still apply. The main generalisation is that this article does not restrict to (compact) symmetric spaces and also generalises the construction to any $\Z_T$-coset model. Finally, in this article we have also studied the hierarchy of equations induced by the flow of these local charges.

The reference~\cite{Evans:2005zd} deals with the local charges in involution for the supersymmetric models on symmetric spaces, working with a superalgebra $\g$. Although we did not consider such models in this article, we expect the construction to extend to these theories, by working with the Grassmann envelope of $\g$ and replacing all traces by supertraces. One should however be careful about how the automorphism $\s$ is extented to the whole matrix algebra (see appendix \ref{App:ExtSigma}) in these supersymmetric cases. Such considerations could allow the construction of local charges in involution for supersymmetric $\s$-models whose target space includes AdS manifolds, with possible applications to the hybrid formulations of string theory.\\

We end this subsection by observing that the Hamiltonian of the $\Z_T$-coset model is related to the quadratic charge $\Q^0_2$ at the origin and the momentum $\Pc_{\Z_T}$ of the theory by
\begin{equation*}
\Hc_{\Z_T} = \frac{\Q^0_2}{\zeta'(0)} + \Pc_{\Z_T},
\end{equation*}
where $\zeta$ was defined in equation \eqref{Eq:DefZeta} (note that this expression also holds for the deformed $\Z_2$ model). Thus, we conclude that the local charges constructed above are conserved, as they commute (at least weakly) with the Hamiltonian.

As the $\Z_T$-coset models are constrained models, their Hamiltonian is defined up to a term $\Tr\bigl( \mu(x) \Cc(x) \bigr)$, where $\mu$ is a $\g$-valued Lagrange multiplier. In this sense, $\Hc_{\Z_T}$ defined above is a particular choice of such a Hamiltonian, which generates a strong zero curvature equation \eqref{Eq:ZCEH} on the Lax matrix $\Lc$. Another choice of Hamiltonian involves the quadratic charge $\Q^\infty_2$ extracted at infinity, namely
\begin{equation*}
\widetilde{\Hc}_{\Z_T} = -\frac{\Q^\infty_2}{\zeta_\infty'(0)} - \Pc_{\Z_T},
\end{equation*}
where $\zeta_\infty$ is defined in the same way than $\zeta$ by $\zeta_\infty(\alpha^T)=\alpha \psi(\alpha)$. This Hamiltonian is weakly equal to $\Hc_{\Z_T}$ and generates a strong curvature equation on the Lax matrix $\Lc^\infty$.

\subsection{Dihedral affine Gaudin models}
\label{SubSec:DGAM}

Let us end this section by discussing briefly dihedral affine Gaudin models (DAGM) and 
their relation to the present article. These models were defined and studied recently 
in reference~\cite{Vicedo:2017cge}, extending the notion of cyclotomic 
Gaudin models of \cite{Vicedo:2014zza}. In particular, it was shown that all the integrable $\s$-models mentioned in the previous subsections and paragraph \ref{Sec:Examples} are examples of such DAGM. In fact, we will argue in this subsection that the construction of local charges of the present article automatically applies to a broader class of DAGM than just these integrable $\s$-models. We use the notations of~\cite{Vicedo:2017cge} and refer the interested reader to the original article for the details.\\

Reference~\cite{Vicedo:2017cge} defines a DAGM as a Hamiltonian field theory with a Poisson algebra of local observables $\widehat{S}_\ell(\widehat{\g}^{\mathcal{D}}_\C)$, depending on the following data:
\begin{itemize}\setlength\itemsep{0.15em}
\item a semi-simple Lie algebra $\g$ ;
\item an automorphism $\s$ of $\g$ of order $T$ and a semi-linear involutive automorphism $\tau$, forming the dihedral group $\Pi=D_{2T}$ ;
\item a set of points $\bm{z}\subset\mathbb{P}^1$ containing infinity (not to be confused with the set of regular zeros $\Zc$ of this article), associated with positive integers $n_x$ ($x\in\bm{z}$), encoded as a divisor $\mathcal{D}=\displaystyle\sum_{x\in\bm{z}} n_x x $ ;
\item complex numbers $\ell^x_p$ ($p=0,\ldots,n_x-1$) for all $x\in\bm{z}\setminus\hspace{-2pt}\lbrace\infty\rbrace$ and $\ell^\infty_q$ ($q=1,\ldots,n_\infty-1$), called the levels of the DAGM.
\end{itemize}
The dihedral group $\Pi$ acts on $\mathbb{P}^1$ \textit{via} the multiplication by $\omega$ and complex conjugation. The set $\bm{z}$ is chosen such that any two points of $\bm{z}$ are in disjoints orbits under the action of $\Pi$. We define the Gaudin poles $\Pi\bm{z}$ as the set of all images of points in $\bm{z}$ by $\Pi$. For this article, we shall restrict the discussion to the case where the highest levels $\ell^x_{n_x-1}$ are all non-zero (this is the main case treated in article~\cite{Vicedo:2017cge}).\\

The twist function $\varphi$ and Lax matrix $\Lc$ of the DAGM are introduced in the form of a ``natural'' connection $\nabla=\varphi(\lambda) \p_x + \varphi(\lambda)\Lc(\lambda,x)$. In particular, the twist function is uniquely defined by the levels $\ell^x_p$ of the model (see equation (4.44) of~\cite{Vicedo:2017cge} and below). The Lax matrix $\Lc(\lambda,x)$ then naturally satisfies an $r/s$-type Poisson bracket \eqref{Eq:PBR} with twist function $\varphi(\lambda)$. Moreover, the Lax matrix and the twist function of the DAGM are shown in~\cite{Vicedo:2017cge} to verify the equivariance properties \eqref{Eq:EquiL} and \eqref{Eq:TwistEqui} and the reality conditions \eqref{Eq:Reality} and \eqref{Eq:TwistReal}. Hence the DAGM automatically fits in the framework of the present article, described in subsection \ref{Sec:Model}.

Let us now look at the zeros of the twist function. In the construction of~\cite{Vicedo:2017cge}, the connection $\nabla$ defined above (and so in particular the twist function) has poles exactly at the Gaudin poles $\Pi\bm{z}$. The zeros of the twist function then do not belong to $\Pi\bm{z}$. Yet, the matrix component $S(\lambda,x)=\varphi(\lambda)\Lc(\lambda,x)$ of $\nabla$ has poles only at the points $\Pi\bm{z}$. Thus, all zeros of the twist function are regular zeros. The methods developed in this article thus apply naturally to the DAGM.\\

Finally, let us relate the discussion of the point at infinity in the present article, based on the results of subsection \ref{Sec:Infinity}, to its treatment in a DAGM given in~\cite{Vicedo:2017cge}. We will argue that the notions studied in that subsection can naturally be transposed to the DAGM setting. For that, we restrict attention to the case where $n_\infty=1$. As one can check from~\cite{Vicedo:2017cge}, this hypothesis is equivalent to the condition that $P(\alpha,x)$, defined in equation \eqref{Eq:DefP}, is regular at $\alpha=0$. With this condition fulfilled, we can then define the field $\Cc(x)$ as in equation \eqref{Eq:DefC}. Throughout the present article, based on the examples of $\Z_T$-coset models, this field $\Cc(x)$ was seen as a gauge constraint (with Theorem \ref{Thm:Gauge} stating that the local charges $\Q^{\lambda_0}_n$ are gauge invariant). This interpretation generalises to any DAGM, as the field $\Cc(x)$ computed above can be seen to coincide with the constraint defined in the paragraph 4.5.3 of the reference~\cite{Vicedo:2017cge}.

In order to define a constrained DAGM, it is required in~\cite{Vicedo:2017cge} that the levels $\ell^x_0$ satisfy some additional relation (4.64). In the context of this article, this relation ensures the regularity of the twist function $\psi(\alpha)$, defined in equation \eqref{Eq:Psi}, 
at $\alpha=0$. If we suppose that the DAGM is such that  $T>1$, then this regularity, combined with the equivariance property \eqref{Eq:EquiInfinity}, implies that $\psi(0)=0$. In this case, the point at infinity is then a regular zero of the model and the methods developed in the present article apply.

\section{Outlook}

In the present paper we showed how, in any integrable field theory described by an 
$r/s$-system with twist function $\varphi(\lambda)$ and underlying finite-dimensional Lie algebra $\g$ of classical type, one can assign an infinite tower of local conserved charges $\mathcal Q^{\lambda_0}_n$, $n \in \mathcal E_{\lambda_0}$, in involution to each regular zero $\lambda_0$ of $\varphi(\lambda)$. If the regular zero $\lambda_0$ is non-cyclotomic then the set of degrees $\mathcal E_{\lambda_0}$ of these charges consists of one plus the exponents of the untwisted affine Kac-Moody algebra $\widehat{\g}$ associated with $\g$ (excluding the Pfaffian in the case of type D). This generalises the results established in \cite{Evans:1999mj, Evans:2000hx} for the principal chiral model, with or without a Wess-Zumino term, on a real Lie group $G_0$ of classical type. On the other hand, if $\lambda_0$ is cyclotomic then some exponents may `drop' and the set of degrees $\mathcal E_{\lambda_0}$ forms a subset of the set of exponents of $\widehat{\g}$ plus one. In the case where the order of the cyclotomy is $T=2$, a simple computation reveals that the `drops' occur precisely when $\g$ is of type A and the automorphism $\sigma$ is inner. In particular, by performing a direct comparison with the results of \cite{Evans:2000qx} we find perfect agreement with the set of exponents naturally associated with a compact symmetric space $G_0 / G_0^\sigma$ listed in Table 1 of that paper (see also \cite{Evans:2001sz}) describing the degrees of the local charges in involution in a symmetric space $\sigma$-model. We were, however, unable to interpret the Coxeter number of the compact symmetric space $G_0 / G_0^\sigma$ in our present formalism. It would be interesting to understand whether there is a more fundamental connection between these two descriptions of the degrees of local charges in involution for symmetric space $\sigma$-models. This would also provide insight into the extension of the definition of the exponents to the case of $\Z_T$-cosets.

As noted in subsection \ref{SubSec:DGAM}, the regularity assumption on the 
zeros of the twist function is intimately related to a condition imposed on the levels 
of a dihedral affine Gaudin model in \cite{Vicedo:2017cge}, so that the results 
of the present paper automatically apply to $r/s$-systems arising in the latter framework.
In light of this fact, a noteworthy feature of the densities of the local charges $\mathcal Q^{\lambda_0}_n$ is that they were all obtained from the particular combination $\varphi(\lambda) \mathcal L(\lambda, x)$ of the twist function and Lax matrix of the model. Indeed, the significance of this expression is that it forms the matrix component of the connection $\nabla = \varphi(\lambda) \partial_x + \varphi(\lambda) \mathcal L(\lambda, x)$ characterising the dihedral affine Gaudin model. As argued in \cite{Evans:2001sz}, the invariant tensors used to build local charges in involution in the principal chiral model and symmetric space $\sigma$-models can be obtained directly from the Weyl-invariant tensors appearing in the leading derivative-free terms of the local charges in involution of $\widehat{\g}$-mKdV theory, given by the Drinfel'd-Sokolov construction. On the other hand, it follows from the results of \cite{Vicedo:2017cge} that $\widehat{\g}$-mKdV theory is itself a cyclotomic affine Gaudin model whose corresponding connection $\nabla$ is the starting point in the construction of Drinfel'd-Sokolov. It is therefore natural to speculate that a more general tower of local charges in involution, involving derivatives of the fields appearing in the Lax matrix $\mathcal L(\lambda, x)$, may be constructed starting from the dihedral affine Gaudin model connection $\nabla$ itself.

With the application to dihedral affine Gaudin models in mind, it is also interesting to note that the results of the present paper bear a strong resemblance to the structure of commuting integrals of motion in a (quantum) Gaudin model for a finite-dimensional Lie algebra $\g$. Indeed, it is well known from the work of Feigin, Frenkel and Reshetikhin \cite{Feigin:1994in} (see also \cite{Frenkel:2003qx, Frenkel:2004qy}) that the algebra of commuting integrals of motion in a Gaudin model with $N \in \Z_{\geq 1}$ sites, known as the Gaudin algebra, is generated by elements of the algebra of quantum observables $U(\g)^{\otimes N}$ whose degrees are given by exponents of $\g$ plus one. An important open problem is to obtain an analogue of this result in the case of quantum Gaudin models associated with affine Kac-Moody algebras. Although this is still currently out of reach, the results of the present article may furnish useful hints in this direction by providing a similar description of local Poisson commuting integrals of motion in classical (dihedral) affine Gaudin models.

\paragraph{Acknowledgments} This work is partially supported by the French Agence Nationale 
de la Recherche (ANR) under grant ANR-15-CE31-0006 DefIS.  S. L. acknowledges 
funding provided by a grant from R\'egion Auvergne-Rh\^one-Alpes. 

\appendix

\section{Extension of the automorphism to the whole space of matrices}
\label{App:ExtSigma}

We consider a Lie algebra $\g$ of classical type A, B, C or D, in its defining matricial representation. We therefore regard elements of $\g$ as acting linearly on a vector space $V$, \textit{i.e.} as element of the space $F$ of endomorphisms of $V$. Note that we can consider some connected matrix group $G \subset F$ whose Lie algebra is $\g$. Let $\s$ be an automorphism of $\g$, of finite order $T$. In this article, we are considering powers of elements of $\g$, which do not belong to $\g$ in general but are elements of $F$. Thus, we want to extend the automorphism $\s$ to the whole space of matrices $F$, in a ``natural way''.

\subsection{The conjugacy case}

Let us begin with the case where $\s$ is inner, \textit{i.e.} when $\s : X \in \g \mapsto Q X Q^{-1}$ for some $Q \in G$. Then the extension of $\s$ to $F$, which by a slight abuse of notation we still denote as $\s$, can be naturally defined as
\begin{equation}\label{Eq:ExtInner}
\begin{array}{rccc}
\s : & F & \longmapsto & F \\
                 & X & \longrightarrow & Q X Q^{-1}
\end{array}
\end{equation}
This covers the case of types B and C, as they do not have any non trivial diagram automorphism.\\

Let us now consider the algebra D$_n$, \textit{i.e.} $\g=\so(2n,\C)$, for $n\geq 5$. In this case, there 
always exists one non trivial diagram automorphism. However, this automorphism can be realised on the defining representation as an external conjugation: $\s: X \in \g \mapsto Q X Q^{-1}$ where $Q$ is not in the group $SO(2n,\C)$ but belongs to $O(2n,\C)$. In this case, the endomorphism $\s$ as defined in equation \eqref{Eq:ExtInner} still naturally extends $\s$ on $F$.

Let us say a few words on the algebra D$_4=\so(8,\C)$. It is known to have 6 diagram automorphisms, forming the \textit{triality}, isomorphic to the symmetric group $S_3$. One of them, of order 2, can also be realised as conjugation by a matrix $Q\in O(8,\C)$ and so extends to $F=M_8(\C)$ by equation \eqref{Eq:ExtInner}. The other non-trivial ones cannot be realised in any ``natural way'' on the defining representation and thus cannot easily be extended to the whole space $M_8(\C)$. We shall not consider them in this article. \\

We now describe the properties of the extension $\s$ defined in \eqref{Eq:ExtInner}. The fact that the automorphism $\s$ of $\g$ is of order $T$ is equivalent to the fact that $Q^T$ belongs to the centraliser of $\g$ in $F$,
\begin{equation*}
Z_F(\g) = \left\lbrace X \in F \; \text{s.t.} \; [X,Y]=0, \; \forall \, Y\in\g \right\rbrace.
\end{equation*}
By Schur's lemma, this implies that $Q^T = \lambda\,\Id$ for some $\lambda\in\C$. Therefore $\s$ on $F$ defined by \eqref{Eq:ExtInner} is also of order $T$. We shall make extensive use of the following five obvious properties of $\s$ as defined in \eqref{Eq:ExtInner}:
\begin{align*}
\s(XY) = \s(X)\s(Y), \;\;\;\;\;
\s(X^n) = \s(X)^n, \;\;\;\;\;
\s(\Id) = \Id, \\
\Tr\bigl(\s(X)\bigr) = \Tr(X), \;\;\;\;\;
\Tr\bigl(\s(X)\s(Y)\bigr) = \Tr(XY), \;\;\;
\end{align*}
for any $X,Y$ in $F$.

\subsection{The transpose case}

The last case that we have to treat is the one of a type A algebra, with 
$\s$ being not inner. We thus consider the defining representation $\g=\sl(d,\C)$. The action of $\s$ on $\g$ can then always be expressed as $\s : X\in\g \mapsto - Q X\Tp Q^{-1}$, where $X\Tp$ is the transpose of $X$ and $Q$ is a matrix in $SL(d,\C)$. Here also we can naturally extend $\s$ to an endomorphism of $F=M_d(\C)$, which we still denote $\s$, by letting
\begin{equation} \label{sigma F type A}
\begin{array}{rccc}
\s : & F & \longmapsto & F \\
                 & X & \longrightarrow & - Q \, X\Tp Q^{-1}
\end{array}.
\end{equation}

Once again, let us investigate the properties of $\s$. As the automorphism $\s$ of $\g$ is 
not inner, its order $T$ must be even, and we shall write $T=2S$. We note that $\s^2$ acts as conjugation by $R=Q (Q\Tp)^{-1}$. The fact that $\s^T=(\s^2)^S=\left. \Id \right|_{\g}$ is thus equivalent to the fact that $R^S$ belongs to the centraliser $Z_F(\g)$. Thus $R^S = \lambda \, \Id$ for some $\lambda\in\C$ and so $\s$ defined in \eqref{sigma F type A} is also of order $T$. We end the subsection by noting the following five properties of $\s$:
\begin{align*}
\s(XY) = - \s(Y)\s(X), \;\;\;\;\;
\s(X^n) = (-1)^{n-1}\s(X)^n, \;\;\;\;\;
\s(\Id) = - \Id, \\
\Tr\bigl(\s(X)\bigr) = -\Tr(X), \;\;\;\;\;
\Tr\bigl(\s(X)\s(Y)\bigr) = \Tr(XY), \hspace{40pt}
\end{align*}
for any $X,Y$ in $F$.

\section{Computation of $\Xi$}
\label{App:Xi}

In this appendix, we give the details of the computation in some particular cases of the term $\Xi^{\lambda\mu}_{nm}(\rho,x,y)$, defined by \eqref{Eq:DefXi}.

\subsection{At a non-cyclotomic regular zero}
\label{App:XiNonCyc}

We first suppose that $\lambda_0$ is a non-cyclotomic regular zero and that $\g$ is of type B, C or D. Recall that in this case, we constructed currents $\K^{\lambda_0}_{2n}=\J^{\lambda_0}_{2n}$ and the associated Lax matrices $\M^{\lambda_0}_{2n}=\Nc^{\lambda_0}_{2n}$. We want to compute $\Xi^{\lambda_0\lambda_0}_{2n\,2m}(\rho,x,y)$, starting from equation \eqref{Eq:DefXi}. Recall from section \ref{Sec:NonCycZero} that for a non-cyclotomic zero $\lambda_0$, one has $U\ti{23}(\lambda_0,\lambda_0)=-\frac{1}{T}\varphi'(\lambda_0)C\ti{12}$. Recall also from subsection \ref{Sec:NonCycZeroBCD} that $S_{2m-1}(\lambda_0,x)$ belongs to the Lie algebra $\g$, as it is an odd power of a matrix in $\g$. Using the completeness relation \eqref{Eq:CompRel}, we find
\begin{align}\label{Eq:XiNonCyc1}
\Xi^{\lambda_0\lambda_0}_{2n\,2m}(\rho,x,y) = - \frac{4nm\varphi'(\lambda_0)}{T} \Tr\ti{2} \left( \Rc^0\ti{12}(\rho,\lambda_0)  \sum_{k=0}^{2n-2} S_k(\lambda_0,x)\ti{2} S_{2m-1}(\lambda_0,y)\ti{2} S_{2n-2-k}(\lambda_0,x)\ti{2} \right) \delta'_{xy},
\end{align}
Using the identity $f(y)\delta'_{xy}=f(x)\delta'_{xy}+\bigl(\p_x f(x) \bigr) \delta_{xy}$ and the fact that $S_p(\lambda,x)S_q(\lambda,x)=S_{p+q}(\lambda,x)$, we get
\begin{equation*}
\Xi^{\lambda_0\lambda_0}_{2n\,2m}(\rho,x,y) = f^{\lambda_0}_{2n\,2m}(\rho,x) \delta_{xy} - \frac{4nm\varphi'(\lambda_0)(2n-1)}{T} \Tr\ti{2} \Bigl( \Rc^0\ti{12}(\rho,\lambda_0) S_{2n+2m-3}(\lambda_0,x)\ti{2} \Bigr) \delta'_{xy},
\end{equation*}
where
\begin{equation*}
f^{\lambda_0}_{2n\,2m}(\rho,x) = - \frac{4nm\varphi'(\lambda_0)}{T} \Tr\ti{2} \Bigl( \Rc^0\ti{12}(\rho,\lambda_0)  \sum_{k=0}^{2n-2} S_k(\lambda_0,x)\ti{2} \p_x\bigl(S_{2m-1}(\lambda_0,x)\ti{2}\bigr) S_{2n-2-k}(\lambda_0,x)\ti{2} \Bigr).
\end{equation*}
Recalling the definition \eqref{Eq:DefNNonCyc} of $\Nc^{\lambda_0}_p$, we obtain
\begin{equation}\label{Eq:XiNonCyc2}
\Xi^{\lambda_0\lambda_0}_{2n\,2m}(\rho,x,y) = \frac{\varphi'(\lambda_0)}{T} \frac{4nm(1-2n)}{2n+2m-2}\Nc^{\lambda_0}_{2n+2m-2}(\rho,x) \delta'_{xy} + f^{\lambda_0}_{2n\,2m}(\rho,x) \delta_{xy}.
\end{equation}
As $\p_x S_p(\lambda,x)=\sum_{l=0}^{p-1} S_{l}(\lambda,x)\p_x \bigl( S(\lambda,x) \bigr) S_{p-1-l}(\lambda,x)$, one can rewrite the function $f^{\lambda_0}_{2n\,2m}$ as
\begin{align*}
& f^{\lambda_0}_{2n\,2m}(\rho,x) = - \frac{4nm\varphi'(\lambda_0)}{T} \sum_{k=0}^{2n-2}\sum_{l=0}^{2m-2} \Tr\ti{2} \Bigl( \Rc^0\ti{12}(\rho,\lambda_0)  \\
& \hspace{150pt} S_{k+l}(\lambda_0,x)\ti{2} \p_x\bigl(S(\lambda_0,x)\ti{2}\bigr) S_{2n+2m-4-k-l}(\lambda_0,x)\ti{2} \Bigr).
\end{align*}
In particular, note that $f^{\lambda_0}_{2n\,2m}=f^{\lambda_0}_{2m\,2n}$.\\

Let us now compute $\Xi^{\lambda_0\lambda_0}_{nm}$ for a non-cyclotomic regular zero $\lambda_0$ and an algebra $\g$ of type A. As in the case of type B, C or D, we have $U\ti{23}(\lambda_0,\lambda_0)=-\frac{1}{T}\varphi'(\lambda_0)C\ti{23}$. Using the generalised completeness relation \eqref{Eq:CompRelSl} and the fact that $\J^{\lambda_0}_p(x)=\Tr\bigl( S_p(\lambda_0,x)\bigr)$, we find from equation \eqref{Eq:DefXi} that
\begin{align*}
\Xi^{\lambda_0\lambda_0}_{n\,m}(\rho,x,y) & = - \frac{nm\varphi'(\lambda_0)}{T} \Tr\ti{2} \Bigl( \Rc^0\ti{12}(\rho,\lambda_0)  \sum_{k=0}^{n-2} S_k(\lambda_0,x)\ti{2} S_{m-1}(\lambda_0,y)\ti{2} S_{n-2-k}(\lambda_0,x)\ti{2} \Bigr) \delta'_{xy} \\
& \hspace{40pt} + \frac{nm(n-1)\varphi'(\lambda_0)}{dT} \J^{\lambda_0}_{m-1}(y) \Tr\ti{2} \Bigl( \Rc^0\ti{12}(\rho,\lambda_0) S_{n-2}(\lambda_0,x) \Bigr) \delta'_{xy}.
\end{align*}
From the identity $f(y)\delta'_{xy}=f(x)\delta'_{xy}+\bigl(\p_x f(x) \bigr) \delta_{xy}$ and equation \eqref{Eq:DefNNonCyc}, we find
\begin{align*}
\Xi^{\lambda_0\lambda_0}_{n\,m}(\rho,x,y) & = - \frac{\varphi'(\lambda_0)}{T} \frac{nm(n-1)}{n+m-2} \Nc^{\lambda_0}_{n+m-2}(\rho,x) \delta'_{xy} \\
& \hspace{70pt}+ \frac{\varphi'(\lambda_0)}{dT} nm \J^{\lambda_0}_{m-1}(y) \Nc^{\lambda_0}_{n-1}(\rho,x) \delta'_{xy} + f^{\lambda_0}_{nm}(\rho,x) \delta_{xy},
\end{align*}
with
\begin{equation*}
f^{\lambda_0}_{nm}(\rho,x) = - \frac{nm\varphi'(\lambda_0)}{T} \Tr\ti{2} \Bigl( \Rc^0\ti{12}(\rho,\lambda_0)  \sum_{k=0}^{n-2} S_k(\lambda_0,x)\ti{2} \p_x\bigl(S_{m-1}(\lambda_0,x)\ti{2}\bigr) S_{n-2-k}(\lambda_0,x)\ti{2} \Bigr).
\end{equation*}
As in the case of type B, C or D, we can re-express $f^{\lambda_0}_{nm}$ as
\begin{equation*}
f^{\lambda_0}_{nm}(\rho,x) = - \frac{nm\varphi'(\lambda_0)}{T} \sum_{k=0}^{n-2}\sum_{l=0}^{m-2} \Tr\ti{2} \Bigl( \Rc^0\ti{12}(\rho,\lambda_0) S_{k+l}(\lambda_0,x)\ti{2} \p_x\bigl(S(\lambda_0,x)\ti{2}\bigr) S_{n+m-4-k-l}(\lambda_0,x)\ti{2} \Bigr).
\end{equation*}
In particular, note that $f^{\lambda_0}_{nm}=f^{\lambda_0}_{mn}$.

\subsection{Around a cyclotomic regular zero}
\label{App:XiCyc}

This subsection is devoted to the computation of $\Xi^{\lambda\lambda}_{nm}(\rho,x,y)$ around the origin $\lambda=0$ and more precisely to the computation of the coefficient of $\lambda^{r_n+r_m}$ in its series expansion. Our starting point is the definition \eqref{Eq:DefXi} of $\Xi^{\lambda\mu}_{nm}$. To evaluate this equation at $\mu=\lambda$,	we will need the expression of $U\ti{12}(\lambda,\lambda)$. We saw in section \ref{Sec:CycZero} that this is given by equation \eqref{Eq:UAround0}.\\

The presence of the partial Casimir $C^{(0)}\ti{12}$ in this equation will gives rise to projections of $S_{m-1}(\lambda,y)$ onto the grading $F^{(0)}=\lbrace Z\in F \, | \, \s(Z)=Z \rbrace$ of the matrix algebra. More precisely, the calculations will involve
\begin{equation}\label{Eq:Defg}
g^\lambda_{mn}(\rho,x,y) = nm\lambda^{-2} \zeta(\lambda^T) \Tr\ti{2} \Bigl( \Rc^0\ti{12}(\rho,\lambda) \sum_{k=0}^{n-2} S_k(\lambda,x)\ti{2} S^{(0)}_{m-1}(\lambda,y)\ti{2} S_{n-2-k}(\lambda,x)\ti{2} \Bigr) \delta'_{xy},
\end{equation}
where $S^{(0)}_{p}$ denotes the projection of $S_{p}$ on $F^{(0)}$. As we are computing the coefficient of $\lambda^{r_n+r_m}$ in $\Xi^{\lambda\lambda}_{nm}$, we will consider the $\lambda^{r_n+r_m}$-term of $g^\lambda_{nm}$. Let us show that this term is actually always zero. Using the conventions and results of the subsections \ref{Sec:EquivT} and \ref{Sec:PBCurrentsCyc}, in particular the integers $\alpha$ and $q_m$, we see that the smallest power of $\lambda$ appearing in $g^\lambda_{nm}$ is $a=\alpha T-2+q_m$, as the $S_p(\lambda,x)$'s are regular at $\lambda=0$. Recall that $r_n$ and $r_m$ are both strictly less than $T-1$ when $n,m\in\E_{0}$ and that in this case, we have $q_m=r_m+1$ (see subsection \ref{Sec:PBCurrentsCyc}). Thus $a = \alpha T-1+r_m$ and hence $a > r_n+r_m$ since $\alpha\geq 1$ and $T-1>r_n$. We can then conclude that the coefficient of $\lambda^{r_n+r_m}$ in $g^\lambda_{nm}$ vanishes, as announced.\\

We will also need the function
\begin{equation*}
f^{\lambda}_{nm}(\rho,x) = - nm\lambda^{T-2}\zeta'(\lambda) \sum_{k=0}^{n-2}\sum_{l=0}^{m-2} \Tr\ti{2} \Bigl( \Rc^0\ti{12}(\rho,\lambda)   S_{k+l}(\lambda,x)\ti{2} \p_x\bigl(S(\lambda,x)\ti{2}\bigr) S_{n+m-4-k-l}(\lambda,x)\ti{2} \Bigr),
\end{equation*}
similar to the function $f^{\lambda_0}_{nm}$ defined in the non-cyclotomic case (see previous subsection) and which possesses the same symmetry property $f^\lambda_{nm}=f^{\lambda}_{mn}$. As for $g^\lambda_{nm}$, we will use more precisely the function $f^{(0)}_{nm}=f^\lambda_{nm}\bigr|_{\lambda^{r_n+r_m}}$, which is also symmetric under the exchange of $n$ and $m$.\\

To go further in the computation, we will need to distinguish between the algebras of type B, C and D and the ones of type A. Let us start with types B, C and D. In this case, we restrict to degrees $2n$ and $2m$ (see subsections \ref{Sec:NonCycZeroBCD} and \ref{Sec:CycBCD}) and thus compute $\Xi^{\lambda\lambda}_{2n\,2m}$. Recall that $S_{2m-1}(\lambda,y)$ belongs to the Lie algebra $\g$, so that we can apply the completeness relations \eqref{Eq:CompRel} and \eqref{Eq:CompRelPart} to it. One then gets
\begin{equation*}
\Xi^{\lambda\lambda}_{2n\,2m}(\rho,x,y) = \widetilde{\Xi}^{\lambda\lambda}_{2n\,2m}(\rho,x,y) + g^\lambda_{2n\,2m}(\rho,x,y),
\end{equation*}
with $g^\lambda_{2n\,2m}$ defined in equation \eqref{Eq:Defg} and
\begin{equation*}
\widetilde{\Xi}^{\lambda\lambda}_{2n\,2m}(\rho,x,y) = - 4nm\lambda^{T-2} \zeta'(\lambda^T) \Tr\ti{2} \Bigl( \Rc^0\ti{12}(\rho,\lambda)  \sum_{k=0}^{2n-2} S_k(\lambda,x)\ti{2} S_{2m-1}(\lambda,y)\ti{2} S_{2n-2-k}(\lambda,x)\ti{2} \Bigr) \delta'_{xy}.\notag
\end{equation*}
The first term $\widetilde{\Xi}^{\lambda\lambda}_{2n\,2m}$ has the same structure as $\Xi^{\lambda_0\lambda_0}_{2n\,2m}$ studied in the previous subsection (see equation \eqref{Eq:XiNonCyc1}). Thus, the calculations of that subsection apply here and we get to an equation similar to \eqref{Eq:XiNonCyc2} for $\widetilde{\Xi}^{\lambda\lambda}_{2n\,2m}$. Namely, we have
\begin{align}\label{Eq:XiLambdaBCD}
&\Xi^{\lambda\lambda}_{2n\,2m}(\rho,x,y) = f^{\lambda}_{2n\,2m}(\rho,x) \delta_{xy} + g^{\lambda}_{2n\,2m}(\rho,x,y)  - \lambda^{T-2}\zeta'(\lambda^T) \frac{4nm(2n-1)}{2n+2m-2}\Nc_{2n+2m-2}(\lambda\,;\rho,x) \delta'_{xy},
\end{align}
where $f^\lambda_{2n\,2m}$ is defined above.

We now compute the coefficient of $\lambda^{r_{2n}+r_{2m}}$ in this expression. We showed above that $g^\lambda_{2n\,2m}$ does not contribute to this term and we have defined $f^{(0)}_{2n\,2m}$ as its contribution from $f^\lambda_{2n\,2m}$. Recall also that $\Nc_k(\lambda\,;\rho,x)$ has the same equivariance property as $\Tc_k(\lambda,x)$ (equation \eqref{Eq:EquiN}) so that its power series expansion starts with $\lambda^{r_k}$. Thus, the smallest power of $\lambda$ in the second line of equation \eqref{Eq:XiLambdaBCD} is greater than or equal to $T-2+r_{2n+2m-2}$. We have shown in subsection \ref{Sec:TypeANonTrivial} that this is equal to $r_{2n}+r_{2m}$ or $r_{2n}+r_{2m}+T$, depending on whether $r_{2n}+r_{2m}$ is greater than or strictly less than $T-2$. We find
\begin{equation}\label{Eq:XiBCD}
\Xi^{\lambda\lambda}_{2n\,2m}(\rho,x,y)\Bigr|_{\lambda^{r_{2n}+r_{2m}}} = f^{(0)}_{2n\,2m}(\rho,x) \delta_{xy} - \theta_{r_{2n}+r_{2m}+2-T}\, \zeta'(0) \frac{4nm(2n-1)}{2n+2m-2}\Nc^0_{2n+2m-2}(\rho,x) \delta'_{xy}.
\end{equation}~

Finally, let us study the case of a  Lie algebra of type A, \textit{i.e.} of $\g=\sl(d,\C)$. The term in $\Xi^{\lambda\lambda}_{nm}$ involving the Casimir $C\ti{12}$ is treated with the generalised completeness relation \eqref{Eq:CompRelSl}. In the same way, one has a generalised completeness relation for the partial Casimir $C^{(0)}\ti{12}$. This relation depends on whether the extension of $\s$ to the whole algebra of matrices fixes 
the identity $\Id$ or not, and thus whether  $\s$ is inner or not (see appendix \ref{App:ExtSigma} and subsections \ref{Sec:TypeATrivial} and \ref{Sec:TypeANonTrivial}). In general, one can write
\begin{equation*}
\Tr\ti{12}\left(C\ti{12}^{(0)}Z\ti{2}\right) = \pi^{(0)}(Z) - \frac{a}{d} \Tr(Z),
\end{equation*}
for any $Z\in M_d(\C)$, with $a=1$ if  $\s$ is inner and $a=0$ if not. Using this relation and the relation \eqref{Eq:CompRelSl}, one finds
\begin{equation*}
\Xi^{\lambda\lambda}_{nm}(\rho,x,y) = \widetilde{\Xi}^{\lambda\lambda}_{nm}(\rho,x,y) + g^\lambda_{nm}(\rho,x,y),
\end{equation*}
with $g^\lambda_{nm}$ defined in equation \eqref{Eq:Defg} and
\begin{align*}
\widetilde{\Xi}^{\lambda\lambda}_{nm}(\rho,x,y) & = - nm\lambda^{T-2} \zeta'(\lambda^T) \Tr\ti{2} \Bigl( \Rc^0\ti{12}(\rho,\lambda)  \sum_{k=0}^{n-2} S_k(\lambda,x)\ti{2} S_{m-1}(\lambda,y)\ti{2} S_{n-2-k}(\lambda,x)\ti{2} \Bigr) \delta'_{xy} \notag \\
& \hspace{20pt} + \frac{nm(n-1)}{d} \frac{\lambda^{T} \zeta'(\lambda^T)-a\zeta(\lambda^T)}{\lambda^2} \Tc_{m-1}(\lambda,y) \Tr\ti{2} \Bigl( \Rc^0\ti{12}(\rho,\lambda) S_{n-2}(\lambda,x)\ti{2} \Bigr)\delta'_{xy}.
\end{align*}
The first term in this expression is treated in the same way as in the case of types B, C and D. Moreover, we recognise in the second term the definition of $\Nc_{n-1}(\lambda\,;\rho,x)$. Finally, we obtain
\begin{align}\label{Eq:XiLambdaA}
&\hspace{-20pt}\Xi^{\lambda\lambda}_{nm}(\rho,x,y) = f^{\lambda}_{nm}(\rho,x) \delta_{xy} + g^{\lambda}_{nm}(\rho,x,y) - \lambda^{T-2}\zeta'(\lambda^T) \frac{nm(n-1)}{n+m-2}\Nc_{n+m-2}(\lambda\,;\rho,x) \delta'_{xy},\\
& \hspace{130pt} + \frac{nm}{d} \frac{\lambda^{T} \zeta'(\lambda^T)-a\zeta(\lambda^T)}{\lambda^2} \Tc_{m-1}(\lambda,y) \Nc_{n-1}(\lambda\,;\rho,x)\delta'_{xy}. \notag
\end{align}

We now compute the coefficient of $\lambda^{r_n+r_m}$ in $\Xi^{\lambda\lambda}_{nm}$. As explained at the beginning of this subsection, $g^\lambda_{nm}$ does not contribute to this term and the contribution of $f^\lambda_{nm}$ is defined as $f^{(0)}_{nm}$. The contribution from the third term is calculated as in the case of types B, C and D. In particular, it vanishes when $r_n+r_m$ is strictly less than $T-2$.

Finally, let us discuss the contribution of the last term. First of all, we note that if $a=1$, $\lambda^T\zeta'(\lambda^T)-a\zeta(\lambda^T)=O(\lambda^{2T})$. Thus the powers of $\lambda$ in this term are greater than $2T-2$. Yet, we have $r_n+r_m<2T-2$ for $n,m\in\E_0$, so this term does not contribute to the $\lambda^{r_n+r_m}$-term in this case. Hence 
for $\s$ inner, we have
\begin{equation}\label{Eq:XiATrivial}
\Xi^{\lambda\lambda}_{nm}(\rho,x,y)\Bigr|_{\lambda^{r_{n}+r_{m}}} = f^{(0)}_{nm}(\rho,x) \delta_{xy}  - \theta_{r_{n}+r_{m}+2-T}\, \zeta'(0) \frac{nm(n-1)}{n+m-2}\Nc^0_{n+m-2}(\rho,x) \delta'_{xy},
\end{equation}
as in the case of types B, C and D.

Suppose now that $\s$ is not inner, so that $a=0$. Then the smallest power of $\lambda$ in the last term of equation \eqref{Eq:XiLambdaA} is greater than or equal to $T-2+r_{n-1}+r_{m-1}$. In subsection \ref{Sec:TypeANonTrivial}, we have shown that this is equal to $r_n+r_m$ if both $r_n$ and $r_m$ are greater than $S-1$ and that it is strictly greater than $r_n+r_m$ otherwise. In conclusion, we find
\begin{align}\label{Eq:XiANonTrivial}
\Xi^{\lambda\lambda}_{nm}(\rho,x,y)\Bigr|_{\lambda^{r_{n}+r_{m}}} & = f^{(0)}_{nm}(\rho,x) \delta_{xy} - \theta_{r_{n}+r_{m}+2-T}\, \zeta'(0) \frac{nm(n-1)}{n+m-2}\Nc^0_{n+m-2}(\rho,x) \delta'_{xy} \\
& \hspace{45pt} + \theta_{r_{n}+1-S}\theta_{r_{m}+1-S}\, \zeta'(0)\, \frac{nm}{d} \J^0_{m-1}(y) \Nc^0_{n-1}(\rho,x) \delta'_{xy}. \notag
\end{align}

 \providecommand{\href}[2]{#2}\begingroup\raggedright\endgroup

\end{document}